\documentclass[acmsmall,nonacm]{acmart}

\usepackage{subfig}
\usepackage{xcolor}
\usepackage[ruled, linesnumbered]{algorithm2e}
\usepackage{balance}
\usepackage{listings}
\usepackage{multirow}
\usepackage{pifont}
\usepackage{cleveref}
\DeclareMathOperator*{\argmax}{argmax}
\newcommand{\ie}{\emph{i.e.,}\xspace}
\newcommand{\eg}{\emph{e.g.,}\xspace}
\newcommand{\etc}{\emph{etc.}\xspace}
\newcommand{\wrt}{\emph{w.r.t.}\xspace}
\newcommand{\resp}{\emph{resp.,}\xspace}

\newcommand{\eat}[1]{}
\def\EndOfExample{\nolinebreak\ \hfill\rule{1.5mm}{2.7mm}}
\definecolor{revR1}{RGB}{0,76,153}
\newcommand{\todo}[1]{\textcolor{red}{#1}}

\newcommand{\bcirclednumber}[1]{\ding{\numexpr181+#1\relax}}

\AtBeginDocument{%
  }

\setcopyright{none}

\begin{document}
\title{Towards Selecting the Informative Alternative Relational Query Plans for Database Education}

\author{Hu Wang}
\email{huwang@stu.xidian.edu.cn}
\affiliation{
  \institution{Xidian University}
  \city{Xi'an}
  \country{China}
}

\author{Hui Li}
\email{hli@xidian.edu.cn}
\affiliation{
  \institution{Xidian University}
  \city{Xi'an}
  \country{China}
}

\author{Sourav S. Bhowmick}
\email{assourav@ntu.edu.sg}
\affiliation{
  \institution{Nanyang Technological University}
  \city{Singapore}
  \country{Singapore}
}

\author{Zihao Ma}
\email{zihao_ma@stu.xidian.edu.cn}
\affiliation{
    \institution{Xidian University}
    \city{Xi'an}
    \country{China}
}

\renewcommand{\shortauthors}{Hu Wang et al.}
\begin{abstract}
Off-the-shelf \textsc{rdbms} typically expose only the query execution plan (\textsc{qep}) of an \textsc{sql} query, without presenting information about representative \textit{alternative query plans} (\textsc{aqp}s) considered during plan selection in a \emph{user-friendly} manner. Providing easy access to representative \textsc{aqp}s is valuable in database education, as it helps learners understand the plan choices made by a query optimizer, one of the several important components related to the topic of relational query processing. In this paper, we present a novel problem called\textit{ informative plan selection problem} (\textsc{tips}) which aims to discover a set of $k$ \textit{informative} \textsc{aqp}s from the underlying plan space so that the \textit{plan informativeness} of the set is maximized. Specifically, we explore two variants of the problem, \textit{batch} \textsc{tips} and \textit{incremental} \textsc{tips}, to cater to diverse learners. Due to the computational hardness of the problem, we present an approximation algorithm to address it efficiently while providing theoretical guarantees for the results. An extensive experimental study--including feedback from real-world learners and a three-year in-class evaluation of academic outcomes--demonstrates the effectiveness of our solutions for database education.
\end{abstract}



\maketitle

\section{Introduction}\label{secintro}
Given an \textsc{sql} query, a relational query engine selects its \textit{query execution plan} (\textsc{qep}) from many \textit{alternative query plans} (\textsc{aqp}s) based on their estimated costs. Although major off-the-shelf relational database systems (\textsc{rdbms}) expose the \textsc{qep} of a query to an end user, they do not reveal the \textsc{aqp}s considered by the underlying query optimizer in a \emph{user-friendly} manner. 
Easy exposition of \textsc{aqp}s is useful in several applications, such as database education~\cite{BL22} and database administration. Although an \textsc{rdbms} (\eg \textit{PostgreSQL}) may allow one to manually pose \textsc{sql} queries with various constraints on \textit{configuration parameters} (\eg \textit{SET enable\_hashjoin = \textsf{true}}) to view the corresponding \textsc{qep}s containing specific physical operators, this strategy requires not only familiarity with the syntax and semantics of the configuration parameters, but also a clear idea of the \textsc{aqp}s that one is interested in. This is often impractical to expect from end users in practice. For instance, in an academic institution where many learners are taking a database system course for the first time, it is unrealistic to assume that they will manually pose such queries accurately to explore \textsc{aqp}s. Consider the following motivating scenario.

\begin{example}\label{eg1}
Lena is currently enrolled in an undergraduate database systems course. She wishes to explore some of the representative \textsc{aqp}s for the following query from the TPC-H benchmark after viewing its \textsc{qep} as depicted in Fig.~\ref{fig:alter_plans}(a).
\begin{lstlisting}
SELECT s_acctbal, s_name, n_name, s_address, s_phone
FROM supplier, nation, region
WHERE s_acctbal>=5000 and s_nationkey = n_nationkey and n_regionkey = r_regionkey;
\end{lstlisting}
Lena wants to know whether there are \textsc{aqp}s that share similar (\resp different) logical plans and physical operators yet exhibit very different (\resp similar) estimated costs, and, if so, what such \textsc{aqp}s look like. Examples of these \textsc{aqp}s are depicted in Fig.~\ref{fig:alter_plans}(b)-(d).
\EndOfExample \end{example}

\begin{figure*}[t]
    \centering
    \subfloat[QEP]{\includegraphics[width=0.24\linewidth, height=4cm]{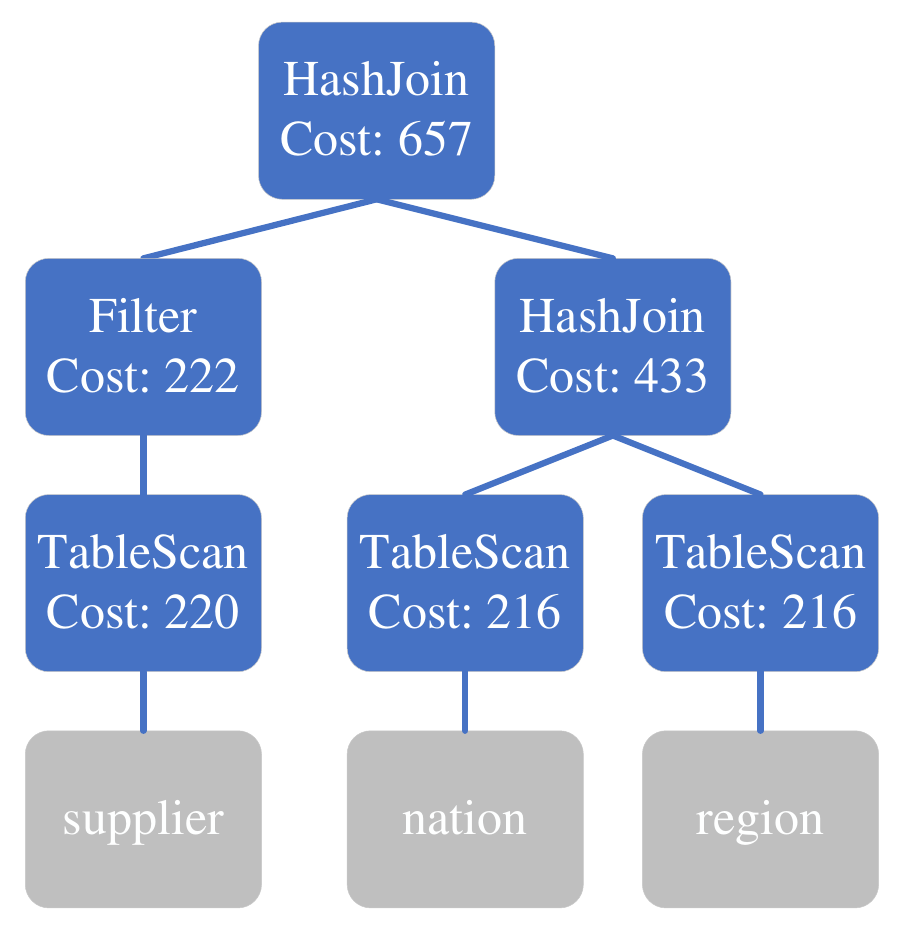}}
    \subfloat[AQP1]{\includegraphics[width=0.24\linewidth, height=4cm]{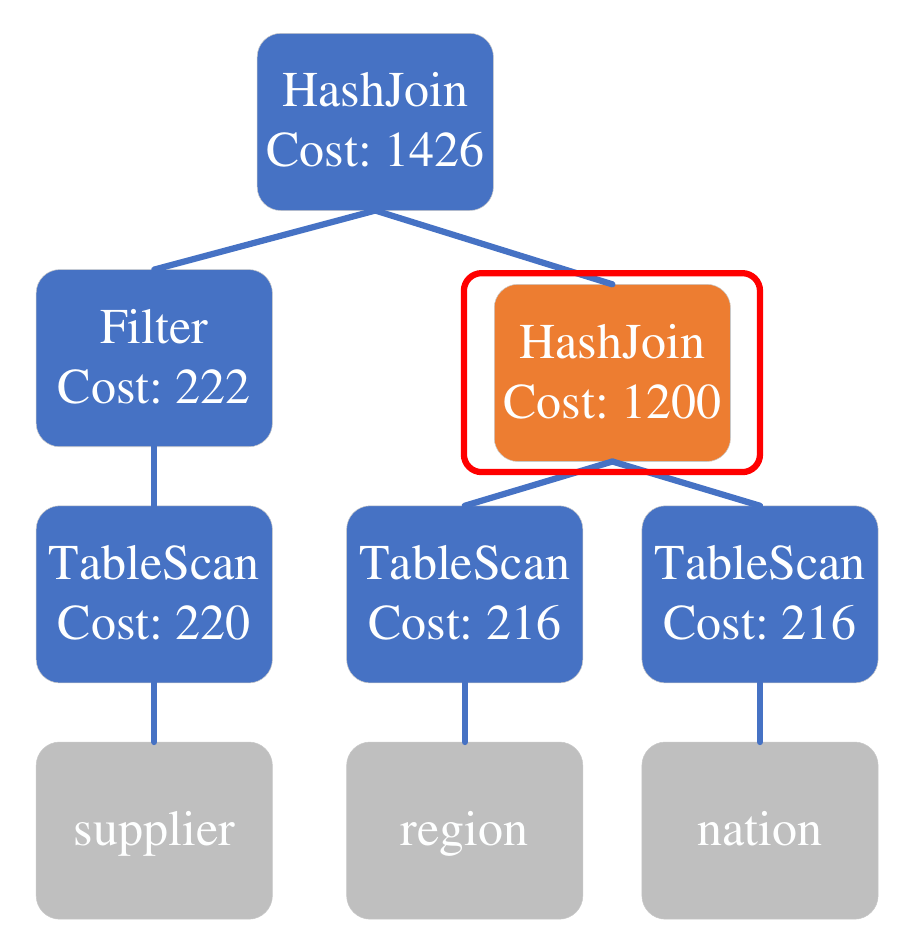}}
    \subfloat[AQP2]{\includegraphics[width=0.24\linewidth, height=4cm]{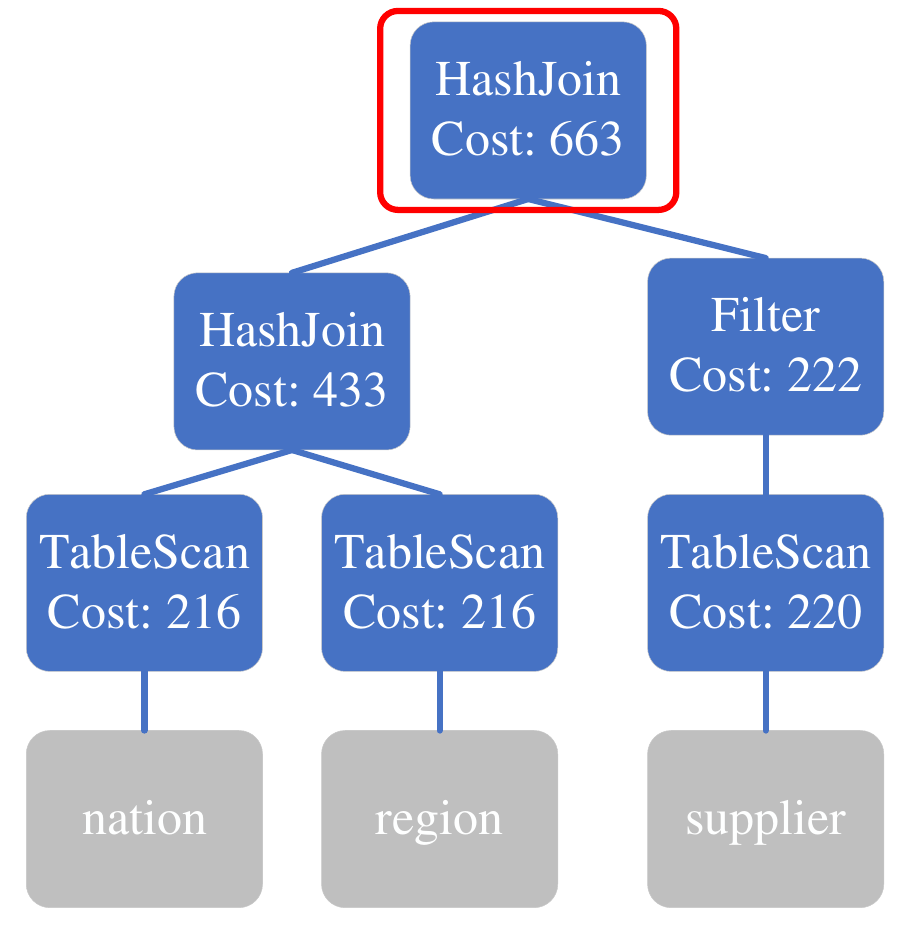}}
    \subfloat[AQP3]{\includegraphics[width=0.24\linewidth, height=4cm]{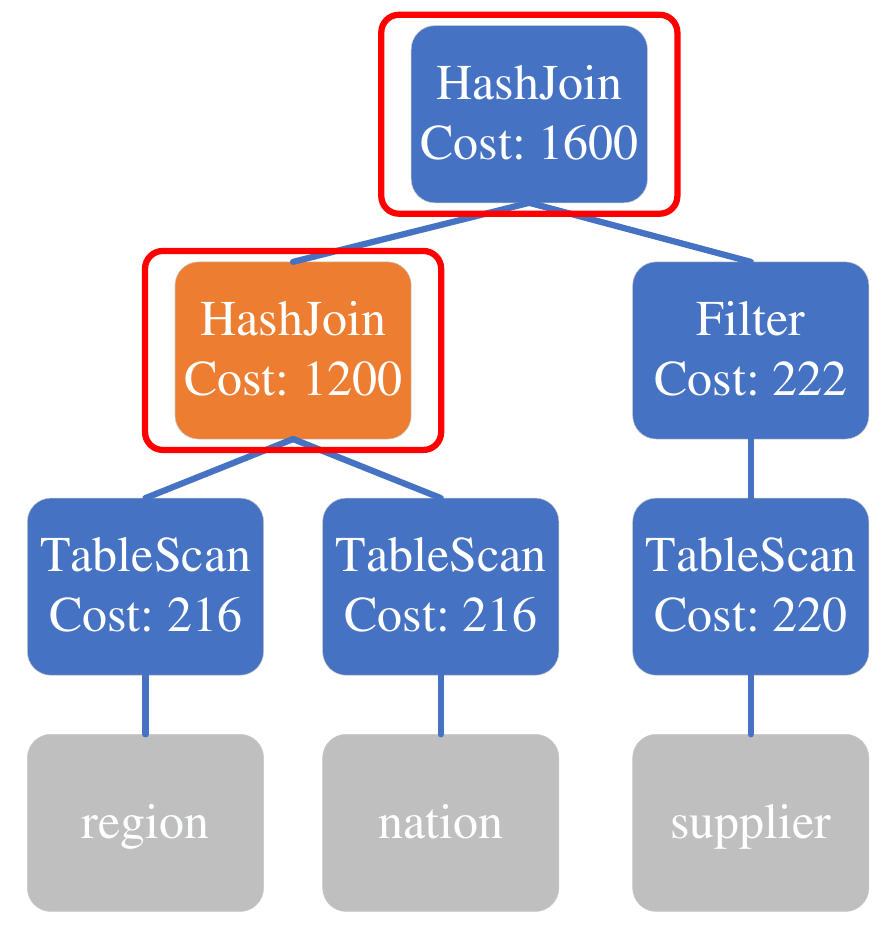}}
    \caption{Example of \textsc{qep} and alternative query plans (\textsc{aqp}s).}
    \label{fig:alter_plans}
\end{figure*}

\eat{\begin{figure}[t]
	\centering
	\includegraphics[width=0.6\linewidth]{DBSurvey.png}
	\caption{Survey results of learners.}
	\label{fig:survey}
\end{figure}}

Clearly, a user-friendly system that can facilitate exploration of \textit{representative} \textsc{aqp}s associated with a given query can greatly aid Lena in answering her queries. Unfortunately, scant attention has been paid by the data management research community to build such a system~\cite{BL22}.\eat{ Particularly, recent efforts for technological support to augment learning of relational query processing primarily focused on the natural language descriptions of \textsc{qep}s~\cite{neuron_liu_2018,towards_weiguo_2021} and visualization of the plan space~\cite{picasso_jayant_2010}.} In this paper, we present a novel solution for efficient selection of \textit{representative} \textsc{aqp}s for a given \textsc{sql} query. Our approach not only has received positive feedback from learners (Section~\ref{sec:ustudy}) but also demonstratively improves academic outcomes (Section~\ref{sec:outcome}).

Selecting a set of \textit{representative} \textsc{aqp}s is technically challenging~\cite{BL22}. First, the plan space can be prohibitively large~\cite{SC98}, and many \textsc{aqp}s are neither interesting nor useful to users (\eg learners). For instance, Lena should not have to inspect all \textsc{aqp}s to learn relational query processing\footnote{\scriptsize Learning relational query processing involves several components, including search algorithms, query rewriting, plan choices, and cost estimation. Our focus here is one component (\ie plan choices) within this larger set.}. Thus, exposing all plans is impractical. Instead, we should retrieve a representative subset based on its \textit{informativeness} to end users. 
However, informativeness is hard to quantify because it can vary across users. In Figures~\ref{fig:alter_plans}(b)-(d), for example, which plans should Lena be shown? Any informativeness measure must account for the plans (including the \textsc{qep}) a user has already seen, avoiding highly \textit{similar} plans, and reflect the user’s interests. Second, scanning all \textsc{aqp}s to find informative ones can be prohibitively expensive. We therefore need efficient techniques to select informative plans. This selection cannot be integrated into the optimizer’s enumeration step, since informative \textsc{aqp}s must be chosen not only based on the \textsc{qep} a user has seen, but also on learner feedback when available.
\eat{Thirdly, at first glance, it may seem that we can select $k>1$ alternative query plans where $k$ is a value specified by a user. Although this is a realistic assumption for many top-$k$ problems, our engagement with users reveals that they may not necessarily be confident in always specifying the value of $k$ (detailed in Section~\ref{sec:ustudy}). One may prefer to \textit{iteratively} view one plan-at-a-time and only cease exploration once they are satisfied with the understanding of the alternative plan choices made by the optimizer for a specific query. Hence, $k$ may not only be unknown \textit{apriori} but also the selection of an \textsc{aqp} in each iteration to enhance their understanding of different plan choices depends on the plans they have viewed thus far. This demands a flexible solution framework that can select alternative query plans in the absence or presence of $k$ value. Lastly, any proposed solution should be \textit{generic} so that it can be easily realized on top of the majority of existing \textsc{rdbms}. }
Third, it may seem natural to allow users to specify $k>1$ \textsc{aqp}s, as in standard top-$k$ problems. However, this introduces drawbacks in our setting. First, it returns the same set of \textsc{aqp}s for a query regardless of who issues it, though different users may find different plans informative. Second, our user study (Section~\ref{sec:ustudy}) shows that learners often lack confidence in choosing $k$. Instead, users should be able to inspect \textsc{aqp}s \textit{iteratively}, 
provide feedback on the usefulness of each plan, and stop once they are satisfied. Thus, $k$ is not only unknown \textit{a priori} but also depends on the plans examined so far. This calls for a flexible framework that supports both {\em batch\/} and {\em incremental\/} selection of \textsc{aqp}s. Finally, any solution should be sufficiently generic to be implemented on most existing \textsc{rdbms}.

In this paper, we formalize the aforementioned challenges into a novel problem called \textit{informa\underline{ti}ve  \underline{p}lan \underline{s}election} (\textsc{tips}) problem and present a solution to address it. Specifically, given an \textsc{sql} query and $k \geq 1$ (when $k$ is unspecified it is set to 1), it retrieves $k$ alternative query plans (\textsc{aqp}s) that\textit{ maximize} \textit{plan informativeness}. Such query plans are referred to as \textit{informative} query plans. To this end, we introduce two variants of \textsc{tips},  namely \textit{batch} \textsc{tips} (\textsc{b-tips}) and \textit{incremental} \textsc{tips} (\textsc{i-tips}), to cater to scenarios where $k$ is specified or unspecified by a user, respectively.  We introduce a novel concept called \textit{plan informativeness} to characterize and select informative \textsc{aqp}s. Specifically, it exploits the preferences of real-world users by mapping them to a novel \textit{distance} and \textit{relevance} measures of \textsc{aqps} \wrt the \textsc{qep} or the set of \textsc{aqp}s viewed by a user thus far. These measures are computed by exploiting \textit{differences} between plans \wrt their topology, physical operators, and estimated cost. This enables us to capture all the relevant dimensions associated with a query plan for any application. Next, given the computational hardness of the \textsc{tips} problem, we present a 2-approximation algorithm to efficiently address both variants of the problem. Specifically, it has a linear time complexity \wrt the number of candidate \textsc{aqp}s. Given that the number of such candidates can be very large for complex queries, we propose two novel \textit{strategies} to prune them efficiently. In this context, we adopt the candidate plan set retrieval framework of the \textsc{arena} system~\cite{arena} to retrieve the candidate plan space for a given \textsc{sql} query as it allows us to build a generic solution that can easily be realized on most existing \textsc{rdbms}.

To demonstrate the effectiveness of our \textsc{aqp} selection framework, we deployed it in a database education environment. Specifically, we analyze feedback from real-world learners enrolled in a database systems course, along with their academic performance \wrt their understanding of the \textsc{qep} selection process. Our analysis demonstrates the efficiency and effectiveness of our solutions in helping learners understand alternative query plans.

In summary, this paper makes the following contributions.
\begin{itemize}
\item We adopt a \textit{learner-centered} approach to characterize the notion of \textit{informative} query plans for database education by engaging real-world learners (Section~\ref{sec:feedback}), and present the notion of \textit{plan informativeness} (Section~\ref{sec:aqps}) to quantitatively model the informativeness of \textsc{aqp}s.  
\item We present a novel \textit{informative plan selection} (\textsc{tips}) problem to obtain a collection of \textsc{aqp}s for a given \textsc{sql} query by maximizing \textit{plan informativeness} (Section~\ref{sec:problem_def}).
\item We present approximate algorithms with quality guarantee that exploit plan informativeness and two \textit{pruning strategies} to address the \textsc{tips} problem efficiently (Section~\ref{ssec:plansel}). Our algorithms underpin any framework for exploration of informative \textsc{aqp}s for \textsc{sql} queries. 
\item We undertake an exhaustive performance study to demonstrate the superiority of the proposed algorithms over several representative baselines (Section~\ref{sec:experiment}). 
\item We conduct a user study and academic assessment involving real-world learners enrolled in a database systems course to demonstrate the usefulness and effectiveness of our solutions for pedagogical support (Section~\ref{sec:appl}). These solutions are one piece of a broader effort to address bottlenecks in learning SQL and relational query processing. 
\end{itemize}


\section{Informative Alternative Plans for Education}
\label{sec:feedback}
A key question that any alternative query plan (alternative plan for brevity) selection technique\eat{ for supporting database education} needs to address is as follows:  \textit{Which alternative plans are ``interesting'' or ``informative'' to a user?}\eat{ At first glance, it may seem that learners may overwhelmingly be interested in alternative plans the differ significantly from the \textsc{aqp}s \wrt cost. Is this really so?} Reconsider Example~\ref{eg1}. Fig.~\ref{fig:alter_plans}(b)-(d) depict three alternative plans for the query where the physical operator/join order differences are highlighted with red rectangles and significant cost differences are shown using orange nodes. Specifically, all three \textsc{aqp}s have different join orders \wrt the \textsc{qep}. However, \textit{AQP1} shares a similar logical plan with \textsc{qep} but has a substantially higher estimated cost. In contrast, \textit{AQP2} has an estimated cost that is very close to that of \textsc{qep}. \textit{AQP3} incurs a similarly high cost to \textit{AQP1}. Suppose $k=2$. Then, among \{\textit{AQP1}, \textit{AQP2}\}, \{\textit{AQP1}, \textit{AQP3}\},  and \{\textit{AQP2}, \textit{AQP3}\},  which one is the best choice? Intuitively, it is reasonable to choose the set that is most ``informative'' to an individual.

\begin{example} Reconsider Example~\ref{eg1}. 
By examining \textit{AQP1} or \textit{AQP3}, Lena learned that the \texttt{HASH JOIN} operator is sensitive to the order of the joined tables, reinforcing concepts she had previously encountered in her classroom lectures. In contrast, \textit{AQP2} shows that altering the join order does not necessarily result in substantially different estimated costs. Lena mistakenly assumed the costs would always differ, based on the examples presented in lectures. Hence, she found \textit{AQP2} informative. Consequently, \{\textit{AQP1}, \textit{AQP2}\} or \{\textit{AQP2}, \textit{AQP3}\} are potentially the most informative set for Lena. 
\EndOfExample \end{example}





\begin{table}[t]
    \centering
        \caption{Categorization for \textit{AP}s depending on their distances (\wrt 3 dimensions) to the \textsc{qep}.}
\label{tbl:info_plans}
\small
        \begin{tabular}{lccc} 
            \toprule
            \textit{Plan type} & \textit{LogiPln} & \textit{PhyOpr} & \textit{Cost}  \\ 
            \midrule
            \small
            $AP_{\uppercase\expandafter{\romannumeral1}}$ & small & small & small \\
            $AP_{\uppercase\expandafter{\romannumeral2}}$ & small & small & large \\
            $AP_{\uppercase\expandafter{\romannumeral3}}$ & small & large & small \\
            $AP_{\uppercase\expandafter{\romannumeral4}}$ & large & small & small \\    
            $AP_{\uppercase\expandafter{\romannumeral5}}$ & small & large & large \\
            $AP_{\uppercase\expandafter{\romannumeral6}}$ & large & small & large \\
            $AP_{\uppercase\expandafter{\romannumeral7}}$ & large & large & small \\
            $AP_{\uppercase\expandafter{\romannumeral8}}$ & large & large & large \\ 
            \bottomrule
        \end{tabular}
\end{table}

\subsection{Categorizing Alternative Query Plans}
We observe from the above example that the differences between a \textsc{qep} and \textsc{aqp}s embody valuable information that individuals can garner to supplement their knowledge. Observe that these differences manifest along three dimensions: the \textit{logical plan} (abbr. \textsf{LogiPln}, \ie the topology of the query plan tree), \textit{physical operators} (abbr. \textsf{PhyOpr}), and the estimated \textit{cost} (\textsf{Cost}) of the plan. Throughout the paper, we consistently use \textsf{LogiPln} to denote the logical plan and \textsf{PhyOpr} to denote physical-operator information. These dimensions determine the informativeness of \textsc{aqp}s. Since a query can have exponentially many \textsc{aqp}s, many of which are similar along these three dimensions, we group them into a concise set of \textit{categories} that characterize informative plans. Specifically, we define eight categories, $AP_{\uppercase\expandafter{\romannumeral1}}$ to $AP_{\uppercase\expandafter{\romannumeral8}}$, shown in Table~\ref{tbl:info_plans}, based on whether each dimension differs slightly or greatly from the \textsc{qep}. For instance, a plan that differs only slightly from the \textsc{qep} in terms of \textsf{LogiPln} and \textsf{PhyOpr}, but exhibits a substantial cost difference, falls into category $AP_{\uppercase\expandafter{\romannumeral2}}$. These eight categories cover \emph{all} possible combinations of differences from the \textsc{qep} along the three dimensions.

\subsection{Learner-centered Characterization of Informative AQPs} 
We posit that informativeness of an \textsc{aqp} cannot be determined accurately in practice without considering the interest of end users. Hence, for database education we adopt a learner-centered approach to understand what categories of plans are deemed as \textit{informative} or \textit{uninformative}. Our objective is to examine how learners perceive relationships between plans along the aforementioned dimensions, rather than to exhaustively cover every operator; additional operators, such as \texttt{UNION}, naturally fit into this categorization without modification. First, we survey 100 unpaid volunteers who have recently taken a database system course or are currently junior database engineers in industry to determine the informativeness of a set of \textsc{aqp}s. Note that we focus on both groups of volunteers as they cover current learners of a database systems course as well as those who have learned it in the recent past and applied their knowledge in industry. Second, we take the results of the survey as input to classify the aforementioned plan categories into informative or uninformative. Observe that any \textsc{aqp} can be assigned to one of the eight categories. Subsequently, we shall exploit these plan categories to select informative \textsc{aqp}s for any given query.   

 \begin{figure}[t]
	\centering
	\subfloat[Informative plans]{\includegraphics[width=0.45\linewidth]{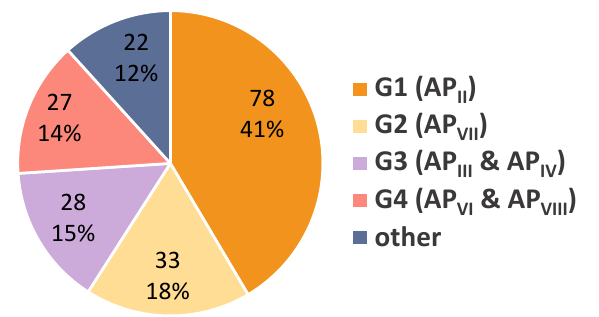}}
	\subfloat[Uninformative plans]{\includegraphics[width=0.51\linewidth]{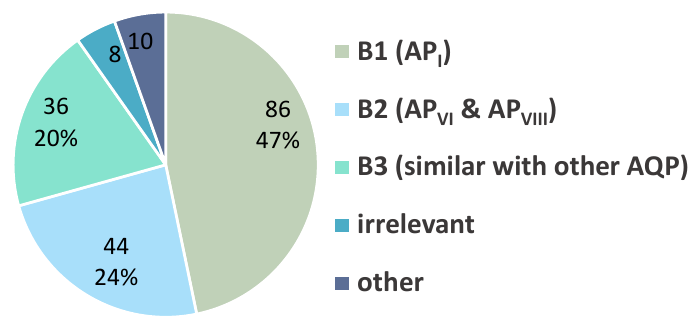}}
  \caption{Survey results.}
	\label{subfig:q4}
\end{figure}  

A key challenge in any human-centered approach is the number of items volunteers must evaluate. Given the exponential number of \textsc{aqp}s for a query, which plans--and how many--should be presented for feedback? Randomly selecting a large number of plans is ineffective: it may not adequately cover the plan categories and, as our interactions with volunteers suggest, large sets discourage feedback by requiring users to browse too many plans\eat{ (\todo{According to Table~\ref{tab:imdb_space}, for a given query, the query optimizer generates thousands of AQPs})}. 
We therefore use \textsc{aqp}s from two different \textsc{sql} queries on the IMDb dataset~\cite{good_viktor_2015}, each with 3–4 joins and different table scan methods, and ask all volunteers to provide feedback on them. For each query, we show the \textsc{qep} and 8 \textsc{aqp}s representing the plan categories. These \textsc{aqp}s are heuristically chosen based on their differences from the \textsc{qep} along the three dimensions. For example, we may select an \textsc{aqp} that exhibits small differences in \textsf{LogiPln} or \textsf{PhyOpr}, but large differences along the remaining dimensions. This strategy covers all plan categories while limiting volunteer effort to reviewing only 18 plans. To mitigate plan-specific factors beyond \textsf{PhyOpr}, \textsf{LogiPln}, and \textsf{Cost}, we present all volunteers with the same representative plan for each category because randomizing plans within a category could therefore bias preferences toward such incidental details rather than the intended category-level differences. Note that our strategy may occasionally select some similar alternative plans. We deliberately expose these plans to the volunteers to garner their feedback on viewing similar alternative plans. 

We ask the volunteers to choose the \emph{most} and \emph{least} informative plans given a \textsc{qep} and provide justifications for their choices. Specifically, volunteers classify a plan as informative if the \emph{specific} differences between it and the \textsc{qep} augment, reinforce, or rectify their knowledge related to the plan space and choices (considered by a relational query processor) that they may have acquired through classroom lectures, textbooks, or on-the-job experiences. Note that the aforementioned plan categories are not revealed to volunteers to avoid bias. \eat{We received 188 feedback from them.} We received 188 pieces of valid feedback from them (not every volunteer provided effective feedback for all plans). We map the specified informative and uninformative plans to the corresponding plan categories $AP_{\uppercase\expandafter{\romannumeral1}}$ to $AP_{\uppercase\expandafter{\romannumeral8}}$ and count the number of occurrences for each category. Fig.~\ref{subfig:q4}(a) reports the distribution of feedback. The eight categories are divided into four groups, \ie $G_1 - G_4$, depending on the feedback. The group entitled \textsf{other} contains feedback that is too sparse to be grouped as representative of the volunteers. Fig.~\ref{subfig:q4}(b) reports the distribution of uninformative plans and categories. $B_1$ represents plans that are very similar to the \textsc{qep} ($AP_{\uppercase\expandafter{\romannumeral1}}$) and therefore voted uninformative, as they do not convey any useful information. $B_2$ represents  $AP_{\uppercase\expandafter{\romannumeral6}}$ and $AP_{\uppercase\expandafter{\romannumeral8}}$. Both these categories have large cost differences, as well as at least one other dimension that differs significantly. Given that two of the three dimensions differ significantly (including cost), it is intuitive that these plans do not interest the volunteers. $B_3$ represents plans that are very similar to other alternative plans viewed by volunteers. In addition, $AP_{\uppercase\expandafter{\romannumeral5}}$ garnered too little feedback to be deemed informative or uninformative. Moreover, there are eight other pieces of feedback that are orthogonal to the problem addressed in this work (\eg feedback on the visual interface design). Hence, they are categorized as \textsf{``irrelevant''}.

\subsection{Informative \& Uninformative Plan Categories} 
In summary, our engagement with volunteers reveals that the majority find alternative plans of categories $AP_{\uppercase\expandafter{\romannumeral2}}-AP_{\uppercase\expandafter{\romannumeral4}}$  and $AP_{\uppercase\expandafter{\romannumeral7}}$ informative. However, $AP_{\uppercase\expandafter{\romannumeral1}}$, $AP_{\uppercase\expandafter{\romannumeral5}}$, $AP_{\uppercase\expandafter{\romannumeral6}}$, and $AP_{\uppercase\expandafter{\romannumeral8}}$ are not. Observe that the number of feedback responses classifying the latter two as uninformative is substantially higher than those classifying them as informative ($G_4$ vs. $B_2$). Hence, we consider them as uninformative. Volunteers also find similar alternative plans uninformative. In the sequel, we shall use these insights to guide the design of algorithms and \textit{pruning strategies} to automatically select informative \textsc{aqp}s.

\noindent\textbf{\textit{Remark}.} It is worth noting that the results of the above human-centered characterization of \textsc{aqp}s into different plan categories are application dependent. In this paper, we focus on database education, and therefore the results in Fig.~\ref{subfig:q4} should be considered within this context. Certainly, for a different application (\eg database administration), the distribution of plan categories for informative and uninformative \textsc{aqp}s may be different. Nevertheless, the proposed human-centered framework can still be leveraged to determine the distributions.

\section{Informative Plan Selection (\textsc{tips}) Problem} 
\label{sec:problem_def}
In this section, we formally define the \textsc{tips} problem. We begin with a brief background on relational query plans.

\subsection{Relational Query Plans}

Given an arbitrary \textsc{sql} query, an \textsc{rdbms} generates a \textit{query execution plan} (\textsc{qep}) to execute it.  A \textsc{qep}  consists of a collection of physical operators organized in the form of a tree, namely the \textit{physical operator tree} (operator tree for brevity). Fig.~\ref{fig:alter_plans}(a) depicts an example \textsc{qep}. Each physical operator, \eg \texttt{SEQUENTIAL SCAN}, \texttt{INDEX SCAN}, takes as input one or more data streams and produces an output one. A \textsc{qep} explicitly describes how the underlying \textsc{rdbms} shall execute the given query. Notably, given a \textsc{sql}, there are many different query plans, other than the \textsc{qep}, for executing it. For instance, there are several physical operators in a \textsc{rdbms} that implement a join, \eg \texttt{HASH JOIN}, \texttt{SORT MERGE}, \texttt{NESTED LOOP}. We refer to each of these different plans (other than the \textsc{qep}) as an \textit{alternative query plan} (\textsc{aqp}).  Figures~\ref{fig:alter_plans}(b)-(d) depict three examples of \textsc{aqp}. It is well-known that the plan space is non-polynomial~\cite{SC98}.

\subsection{Problem Definition} 
\eat{Given that the plan space is exponential, it is infeasible to reveal all alternative plans to a learner.} The \textit{informative plan selection} \textit{(\textsc{tips}) problem} aims to automatically select a small number of alternative plans that maximize \textit{plan informativeness}. These plans are referred to as \textit{informative alternative plans} \textit{(informative plans/\textsc{aqp}s} for brevity). The \textsc{tips} problem has two flavors: 

\begin{itemize}
	\item \emph{Batch \textsc{tips}}. A user may wish to view top-$k$ informative plans beside the \textsc{qep};
	\item \emph{Incremental \textsc{tips}}. A user may iteratively view an informative plan besides what has already been shown to him or her. Moreover, a user can \textit{rate} each viewed plan \wrt their informativeness. This rating impacts the succeeding \textsc{aqp} selection.  
\end{itemize}

Formally, we define these two variants of \textsc{tips} as follows. 

\begin{definition}\label{bapq}
{\em Given a \textsc{qep} $\pi^*$ and a budget $k$, the \textbf{batched informative plan selection (\textsc{b-tips}) problem} aims to find top-$k$ alternative plans $\Pi_{OPT}$ such that the plan informativeness of the plans in $\Pi_{OPT}\cup\{\pi^*\}$ is maximized, 
\begin{align*}
\Pi_{OPT} = \argmax_{\Pi} U(\Pi\cup\{\pi^*\}) \qquad \text{s.t.} \quad |\Pi|=k \quad \text{and} \quad \Pi\subseteq\Pi^*\setminus\{\pi^*\}
\end{align*}
where $\Pi^*$ denotes the space of all possible plans and $U(\cdot)$ denotes the plan informativeness of a set of query plans.\/}
\end{definition}

\begin{definition}\label{iAPQ}
	{\em Given a query plan set $\Pi$ a user has seen so far and the user's preference $r\in\{0,1\}$ of the current plan, the \textbf{iterative informative plan selection (\textsc{i-tips}) problem} aims to select an alternative plan $\pi_{OPT}$ such that the new plan informativeness of the plans in $\Pi\cup\{\pi_{OPT}\}$ is maximized, $$\pi_{OPT}=\argmax_\pi U_{r}(\Pi\cup\{\pi\}) \mbox{ s.t. } \pi\in\Pi^* \mbox{ and } \pi\notin\Pi$$ where $U_{r}(\cdot)$ denotes the plan informativeness adjusted according to $r$.\/}
\end{definition}
Here $U_r(\cdot)$ employs the same \textit{MaxMin-style} plan informativeness as $U(\cdot)$, but its reference point is updated in each iteration based on the user’s feedback $r$. The update mechanism is described in detail in Section~\ref{sec:pi} and formalized in Algorithm~\ref{alg: update_u}.
In the next section, we shall elaborate on the notion of \textit{plan informativeness} $U(\cdot)$ to capture informative plans.
\eat{Both of the problems are conceptually (but not from the solution aspect) correlated with the existing facility placement (FP) problem. Given a group of existing locations of a specific kind of facility, the goal of facility placement is to find a location from series of candidate places to establish a new facility such that the objective function (\eg minimum service distance) is optimum~\cite{place_cui_2018}. If we map the selected plan set in the iAPQ problem to the existing locations, the unselected plan set to the candidate locations, and the utility score to the objective function, then the iAPQ problem is a classic location selection problem. The bAPQ problem is a bit special where we need to select multiple locations. Unfortunately, our objective function is not as simple as the FP problem. As we have seen in the example (AP1$+$AP3 vs. AP1$+$AP2) shown in Fig.~\ref{fig:alter_plans}, distance is not the only factor that has to be taken into account in query plan learning, we shall elaborate the factors and present our objective function in the next section. Besides, a different objective leads to completely different solution path, which will be demonstrated in the next section.}
\begin{example}
Reconsider Example~\ref{eg1}. Suppose Lena has viewed the \textsc{qep} in Fig.~\ref{fig:alter_plans}(a). She may now wish to view two additional \textsc{aqp}s (\ie $k=2$) that are informative. The \textsc{b-tips} problem is designed to address it by selecting these \textsc{aqp}s from the entire plan space so that the plan informativeness $U(\cdot)$  is maximized. On the other hand, suppose another learner is unsure about how many alternative plans should be viewed or is dissatisfied with the results returned by \textsc{b-tips}. In this case, the \textsc{i-tips} problem enables the learner to iteratively select ``\textsc{aqp}-at-a-time'' and provide a rating \wrt its informativeness until satisfied. The plan informativeness is maximized by considering the rating from the learner.
\EndOfExample\end{example}
Given a plan informativeness measure $U(\cdot)$ and a plan space $\Pi^\ast$, we can find an alternative plan for the \textsc{i-tips} problem in $O(|\Pi^\ast|)$ time. However, the plan space increases exponentially with increasing number of joins and it is NP-hard to find the best join order~\cite{optimal_ibaraki_1984}. This demands efficient solutions especially when there is a large number of joined tables. 
\begin{theorem}\label{np-bapq}
	\textit{Given a plan informativeness measure $U(\cdot)$, the \textsc{b-tips} problem is NP-hard.}
\end{theorem}

\section{Quantifying Plan Informativeness} \label{sec:aqps}  
In this section, we describe the quantification of \textit{plan informativeness} $U(\cdot)$\eat{introduced in the preceding section} by exploiting the insights from learner-centered characterization of \textsc{aqp}s in Section~\ref{sec:feedback}. 

\subsection{Lessons from Characterization of AQPs}
In Section~\ref{sec:feedback}, feedback from volunteers reveals that the selection of informative plans is influenced by the following: 

\begin{itemize}
\item[(a)] The \textsf{LogiPln}, \textsf{PhyOpr}, and \textsf{Cost} differences of an alternative plan \wrt the \textsc{qep} (and other alternative plans) can be exploited to differentiate between informative and uninformative alternative plans\eat{ that are consistent with learners' pedagogical interests}. Hence, we need to quantify these \textit{differences} for informative plan selection. 
\item[(b)] 
In database education, volunteers typically prefer not to encounter alternative plans that closely resemble the \textsc{qep} or any plans they have already seen--specifically, plans that are similar across all three dimensions: \textsf{LogiPln}, \textsf{PhyOpr}, and \textsf{Cost}. Therefore, the \textit{relevance} of a new alternative plan for a learner should be assessed based on the set of query plans they have viewed so far.
\end{itemize}

Notably, our survey revealed that learners found plans with minor \textsf{LogiPln} or \textsf{PhyOpr} differences but substantial cost differences (\eg $AP_{\uppercase\expandafter{\romannumeral2}}$) particularly informative, which motivates our approach to measuring relevance.
We elaborate on the quantification approach to capture these two factors and then utilize them to quantify \textit{plan informativeness}.

\begin{figure}[t]
    \centering
    \includegraphics[width=0.7\linewidth]{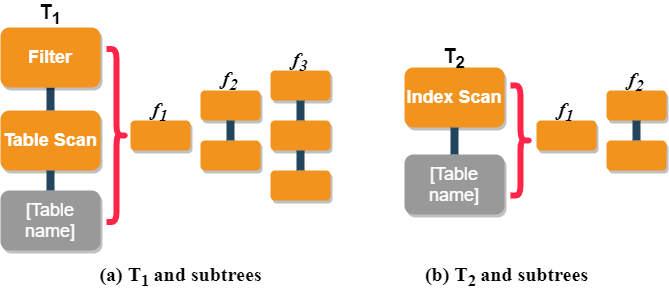}
    \caption{Plan Tree examples.}
    \label{fig:subtree}
\end{figure}

\subsection{Plan Differences}\label{ssec:metric}
Broadly, $U(\cdot)$ should facilitate identifying specific \textsc{aqp}s other than the \textsc{qep} such that by putting them together a learner can reinforce her understanding of the alternative plan choices made by a query optimizer for her query (\ie informative). This involves effectively quantifying the differences between different plans \wrt their \textsf{LogiPln}, \textsf{PhyOpr}, and \textsf{Cost}. As $U(\cdot)$ is defined over a set of query plans, we shall first discuss how to quantify these \textit{differences}.

\noindent\textbf{LogiPln Difference.} To evaluate the \textsf{LogiPln} difference between a pair of physical operator trees, we need to employ a measure to compute structural differences between trees. The choice of this measure is orthogonal to our framework and any distance measure that can compute the structural differences between trees efficiently can be adopted. By default, we use the \textit{subtree kernel}\eat{\footnote{\scriptsize We have conducted experiments by employing an alternative LogiPln difference measure in our solution, which is reported in Section~\ref{sec:experiment}.}}, which has been widely adopted for tree-structured data~\cite{fast_sm_2003} to measure the structural difference. Formally, a kernel function~\cite{kernel_jo_2004} is a function measuring the similarity of any pair of objects $\{x, x'\}$ in the input domain $\mathcal{X}$. It is written as $\kappa(x, x') = \left<\phi(x), \phi(x') \right>$, in which $\phi$ is a mapping from $\mathcal{X}$ to a feature space $\mathcal{F}$. The basic idea of subtree kernel is to express a kernel on a discrete object by a sum of kernels of their constituent parts. The features of the subtree kernel are \textit{proper subtrees} of the input tree $T$. A proper subtree $f_i$ comprises node $n_i$ along with all of its descendants. Two proper subtrees are identical if and only if they have the same tree structure. For example, consider $T_1$ and $T_2$ in Fig.~\ref{fig:subtree}.  $T_1$ has three different proper subtrees, while $T_2$ only has two subtrees. They share two of them, namely, $f_1$ and $f_2$. Note that we do not consider the content of nodes here. Formally, subtree kernel is defined as follows.

\begin{definition} 
{\em Given two trees $T_a$ and $T_b$, the \textbf{subtree kernel} is:
$\kappa_S\left(T_a, T_b\right) = \sum_{n_a\in N(T_a)} \sum_{n_b\in N(T_b)} \Delta(n_a, n_b)$
where $\Delta(n_a, n_b)=\sum_{i=1}^{|\mathcal{F}|}I_i(n_a)I_i(n_b)$, $\mathcal{F}$ is the feature space that contains all possible proper subtrees, and $I_i(n)$ is an indicator function which determines whether the proper subtree $f_i$ is rooted at node $n$ and $N(T)$ denotes the node set of $T$.\/}
\end{definition}
Since the subtree kernel is a similarity measure, we need to convert it to a distance. Moreover, each dimension may have a different scale. Hence, it is necessary to normalize each dimension following $\kappa\left(T_a, T_b\right)= \kappa_S(T_a, T_b)/ \sqrt{\kappa_S(T_a, T_a)\cdot \kappa_S(T_b, T_b)}$\eat{, and then define the LogiPln distance as follows:}. 
\begin{definition} 
{\em Given two query plans $\pi_a, \pi_b$ with corresponding tree structures $T_a, T_b$, based on the normalized kernel, we can define the \textbf{LogiPln distance} between the plans as 
$s\_dist\left(\pi_a, \pi_b\right) = \sqrt{1-\kappa\left(T_a, T_b\right)}$.}
\end{definition}
\begin{lemma}\label{lem:metric}
\textit{The LogiPln distance is a metric.}
\end{lemma}

 
It is worth noting that the classical approach of computing subtree kernel using suffix tree will lead to quadratic time cost~\cite{fast_sm_2003}.
\eat{Traditionally, the suffix tree is used to calculate the \textit{string kernel} and \textit{subtree kernel}. For \textit{string kernel}, if the weight of substrings can be computed in constant time, the kernel function can be computed in linear time~\cite{fast_sm_2003}. In comparison, for \textit{subtree kernel}, not all substrings reflect valid subtrees (if it is a subtree, the weight is 1, otherwise it is 0), so we cannot compute the weight in constant time. Therefore, the computation of the suffix tree-based subtree kernel will lead to a quadratic time cost} We advocate that this will introduce a large computational cost since we need to build a suffix tree for each alternative plan at runtime. To address this problem, we first serialize a tree bottom-up, and then use a hash table to record the number of different subtrees. When calculating the subtree kernel, we can directly employ a hash table to find the same subtree. Although the worst-case time complexity is still quadratic, as we shall see in Section~\ref{sec:experiment}, it is significantly faster than the classical suffix tree-based approach in practice.
We also note that hash-based distance measures (\eg the Picasso visualizer~\cite{picasso_jayant_2010}) primarily rely on plan hashing to group similar plans, whereas our distances are explicitly defined over \textsf{LogiPln}, \textsf{PhyOpr} and \textsf{Cost}. This allows us to reason about learner-perceived differences and to combine these dimensions in a unified \textit{MaxMin} objective.

\eat{When using suffix tree to calculate the \textit{subtree kernel} of two trees, we only need to construct the suffix tree of one tree, so the suffix tree is especially suitable for the search of similar tree structures. For example, when finding the \textsc{xml} documents which is similar to a given tree structure, we only need construct the suffix tree of the given structure. But in our scenario, we need to calculate the structural distance between any two plans, and each plan tree needs an initialization of the suffix tree, so it is not suitable for our scenario.}

\noindent\textbf{PhyOpr Difference.} In order to further compare the content of the nodes, we adopt the edit distance to quantify \textsf{PhyOpr} differences. Specifically, we first use the preorder traversal to convert a physical operator tree to a string, and based on it we apply the edit distance. Note that the alphabet is the collection of all node contents and the content of each tree node is regarded as the basic unit of comparison. For instance, in Fig.~\ref{fig:subtree}, the strings of $T_1$ and $T_2$ are [``\texttt{Filter}'', ``\texttt{Table Scan}[Table name]''] and [``\texttt{Index Scan}[Table name]''], respectively. The edit distance is 2 as they require at least two inserts/replaces to become identical to each other.

\eat{Since we require that the \textit{PhyOpr distance} is a metric, we further normalize the edit distance~\cite{normalized_yu_2007}. Let $\Sigma$ be the alphabet, $\lambda \notin \Sigma$ is the null string, and $\gamma$ be a nonnegative real number that represents the weight of a transformation. Given two strings $X, Y$, the normalized edit distance is
$$ edit_{norm}(X, Y) = \frac{2\cdot edit(X, Y)}{\alpha\cdot (|X|+|Y|)+edit(X, Y)} $$
where $\alpha=\max\{\gamma(a\rightarrow \lambda), \gamma(\lambda \rightarrow b), a,b\in \sum\}$ and $|X|$ is the length of $X$. This has been proven to be a metric~\cite{normalized_yu_2007}. In our work, $\gamma=1$, so the \textit{PhyOpr distance} can be defined as:}
By normalizing the edit distance following~\cite{normalized_yu_2007},\eat{ who also proved the corresponding result to be a metric,} we define the \textit{PhyOpr distance} as follows.
\begin{definition} \label{def:PhyOpr}
{\em Given two query plans $\pi_a, \pi_b$ and the corresponding strings $X_a, X_b$, the \textbf{PhyOpr distance} between the plans is defined as:
    $$c\_dist(\pi_a, \pi_b) = \frac{2\cdot edit(X_a, X_b)}{|X_a|+|X_b|+edit(X_a, X_b)}.$$}
\end{definition}
Notably, join order difference is a special case where the \textsf{PhyOpr} is the same while the order is different. Clearly, Defn.~\ref{def:PhyOpr} is able to differentiate different join orders, \eg \textit{AP1} and \textit{AP2} in Fig.~\ref{fig:alter_plans}(b) and \ref{fig:alter_plans}(c).

\noindent\textbf{\textit{Remark}.} There exist several different ways to define the distance between \textsf{LogiPln} and \textsf{PhyOpr}. For example, we could extract a finite-length feature vector for each plan, and then map it to a feature space to calculate the similarity via dot product. We shall investigate this approach in Section~\ref{sec:experiment}.

\noindent\textbf{Cost Difference.} Finally, because the (time) cost is a real number, we can adopt $L1$ distance and apply Min-Max normalization accordingly.\eat{ That is, for plan $\pi_{a}$, the cost can be normalized as $\frac{cost_a-cost_{min}}{cost_{max}-cost_{min}}$.} 

We are now ready to present the definition of \textit{distance}.

\begin{definition}\label{def:distance}
 {\em Given two query plans $\pi_i, \pi_j$, the \textbf{distance} between them is
    \begin{align*}
    dist(\pi_i, \pi_j)= \alpha\cdot s\_dist(\pi_i, \pi_j) + \beta\cdot c\_dist(\pi_i, \pi_j) + 
     (1-\alpha-\beta)\cdot cost\_dist(\pi_i, \pi_j)
    \end{align*} 
    where $0\le \alpha, \beta\le 1$ are user-defined weights based on the importance of each component.\/}
\end{definition}

\subsection{Relevance of an Alternative Plan}\label{ssec:relevance}
The aforementioned distance measures can prevent the selection of an \textsc{aqp} that is similar to plans that have been already viewed. However, these measures individually cannot distill informative plans from uninformative ones. Therefore, we propose a measure called \textit{relevance} (denoted as $rel(\cdot)$), to evaluate the ``value'' of each plan, so that informative plans receive higher \textit{relevance} scores. Recall the plan categories introduced in Section~\ref{sec:feedback}. We convert the differences in each category into a binary form (Fig.~\ref{fig:reference}(a)) where 1 indicates large. We can also represent the informative value of these types (based on learner feedback) in binary form (1 means informative). We can then employ a multivariate function fitting to find an appropriate \textit{relevance function}  such that alternative plans that are close to the informative ones receive high scores. With only eight labeled plan types, higher-order models are underdetermined and less interpretable; standard ML regressors (\eg SVR) do not improve agreement with learner feedback compared to the cubic form. Since eight existing samples are insufficient to fit a polynomial function with an order greater than 3~\cite{jam2013introduction}, it should appear in the following form: $\rho_1\times s+\rho_2\times c + \rho_3 \times cost + \rho_4\times s\times c + \rho_5\times s\times cost + \rho_6 \times c \times cost + \rho_7\times s\times c\times cost $ where
$\rho_1,\ldots,\rho_7$ are the parameters to be fit. By fitting this polynomial, we can define the \textit{relevance} measure as follows.

\eat{The\textit{ relevance} measure evaluates the degree of deviation between the topological difference and cost difference between an alternative plan and the \textsc{qep} or a set of already viewed plans. Intuitively, the higher the relevance score, the greater the deviation between them. This enables us to select informative alternative plans that are different from plans viewed by a learner so far.  Formally, it is defined as follows.}

\begin{definition}
\label{def:relevance}
           {\em  Given the normalized LogiPln distance $s_i$, PhyOpr distance $c_i$ and cost distance $cost_i$ between an alternative plan $\pi_i$ and the \textsc{qep} $\pi^*$, the \textbf{relevance score} of $\pi_i$ is defined as} $rel(\pi_i)=s_i+c_i+cost_i-s_i\cdot c_i - 
           2\cdot s_i \cdot cost_i -2\cdot c_i \cdot cost_i + 2\cdot s_i \cdot c_i \cdot cost_i$. 
\end{definition}
The effect of $rel(\cdot)$ is shown in Fig.~\ref{fig:reference}(b). The red vertices represent the plan types that learners consider the most informative, which can effectively be captured by $rel(\cdot)$. In addition, plans close to the informative ones receive high scores accordingly. In fact, each plan $\pi$ can be referred to as a 3-dimensional vector $\vec{\pi}$ in the space shown in Fig.~\ref{fig:reference}(b). Hence, $rel(\cdot)$ transforms a vector into a scalar value. We shall use $rel(\pi)$ instead of $rel(\vec{\pi})$.  

\noindent\textbf{\textit{Remark}.} Note that the distance and relevance scores are different. The former evaluates the differences between plans, while the latter is used to assign a higher score to \textsc{aqp}s that users consider as informative. Any lower-order polynomial cannot satisfy this requirement. More importantly, these measures are generic and orthogonal to the distribution of plan categories and the characterization approach of informative alternative plans in specific applications.
\eat{The ``value'' column in Fig.~\ref{fig:reference}(a) denotes the relevance score derived from learner feedback.}

\begin{figure}[t]
	\centering
	\subfloat[Categories of \textsc{aqp}s]{\label{subfig:dp} \includegraphics[width=0.35\columnwidth]{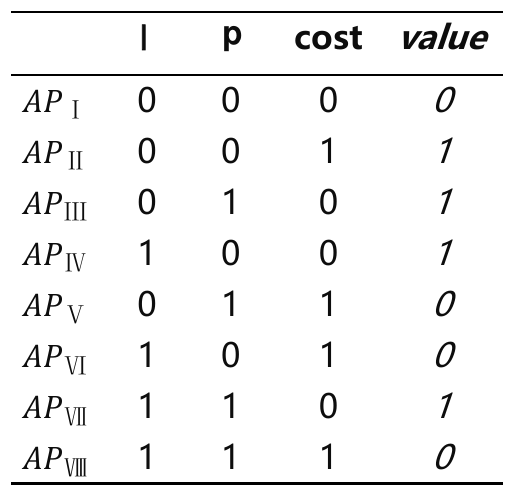}}
	\subfloat[Effect of the fitted function]{\label{subfig:fitted} \includegraphics[width=0.46\columnwidth]{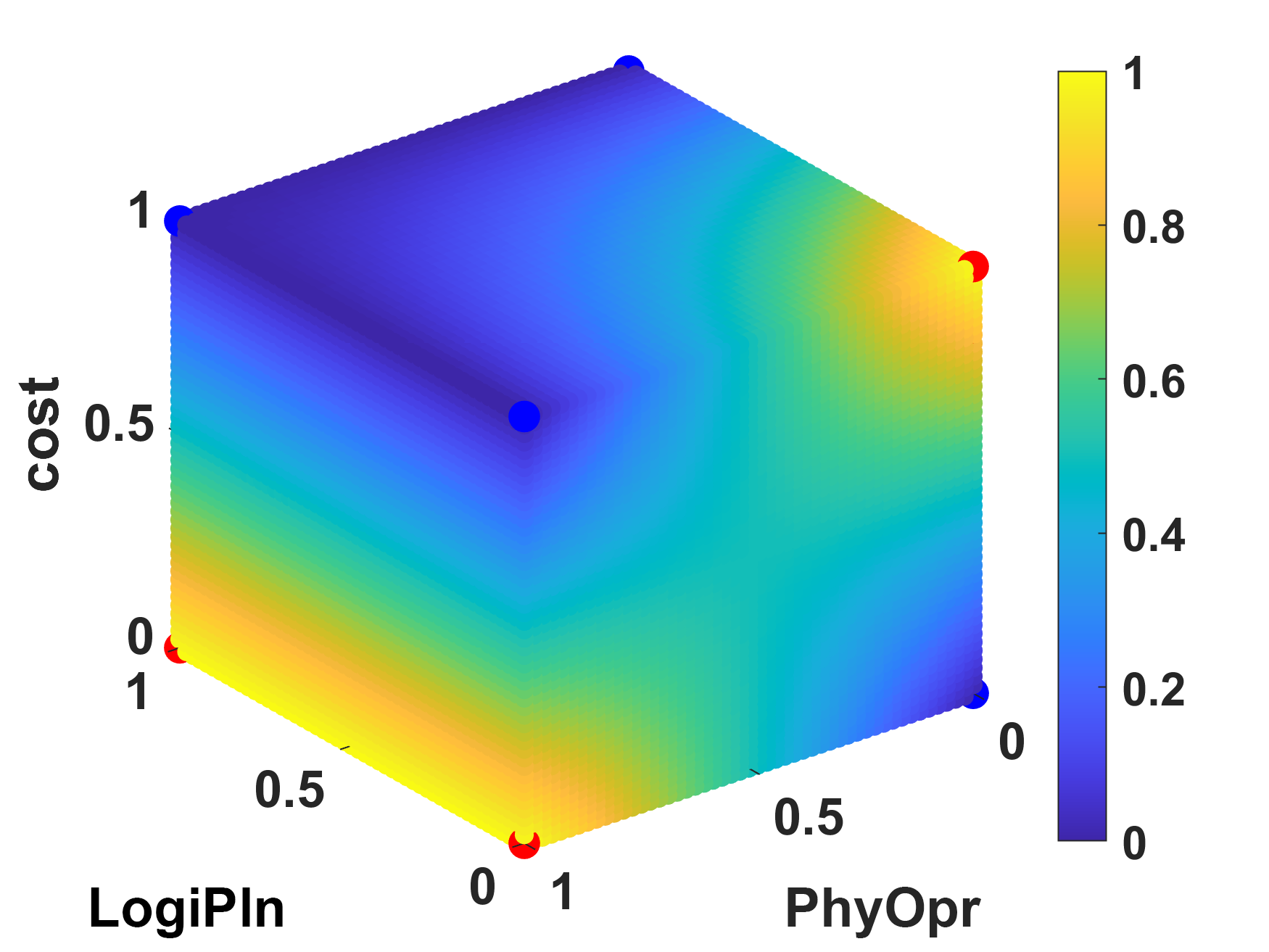}}
        \caption{Relevance (in (a), the ``value'' column denotes the relevance score derived from learner feedback).}
  \label{fig:reference}
\end{figure}

\subsection{Plan Informativeness} \label{sec:pi}
\eat{In the preceding subsections, we introduce the distance and the relevance measures between a pair of query plans.} Intuitively, through the \textit{plan informativeness} measure, we aim to maximize the minimum relevance and distance of the selected query plan set. To this end, we can exploit measures that have been used in diversifying search results, such as MaxMin, MaxSum, DisC, \etc\cite{search_drosou_2010,disc_ma_2012}. Based on~\cite{poikilo_dr_2013}, we utilize the MaxMin model in this work. We adopt the MaxMin criterion because it directly safeguards the \emph{worst-case} informativeness among the selected plans, matching the pedagogical goal of avoiding any ``uninformative'' outlier. It also provides a simple, interpretable objective with a known approximation guarantee. By contrast, objectives such as MaxSum can be driven by a small number of highly informative plans, while still admitting low-informative ones into the selection.

Suppose the plan set is $\Pi$. Then the objective is to maximize $(1-\lambda)\min \limits_{\pi_i\in \Pi} rel(\pi_i) + \lambda \min \limits_{\pi_i, \pi_j \in \Pi} dist(\pi_i, \pi_j)$, where $\lambda > 0$  is a parameter specifying the trade-off between relevance and distance.
According to~\cite{axiomatic_go_2009}, this bi-criteria objective problem can be transformed to finding the distance defined as follows.
\begin{definition}\label{def:new_dist}
{\em Given two query plans $\pi_i, \pi_j$, the \textbf{refined distance} function is 
 $$ Dist(\pi_i, \pi_j) = \frac{(1-\lambda)}{2}(rel(\pi_i)+rel(\pi_j)) + \lambda dist(\pi_i, \pi_j).$$\/}
\end{definition}
\begin{lemma}\label{disttri}
 \textit{The distance function defined in Defn.~\ref{def:new_dist} satisfies the triangle inequality.}
\end{lemma}
Consequently, we define the \textit{plan informativeness} $U(\cdot)$ as follows.
\begin{definition}\label{planinfor}
{\em Given a query plan set $\Pi$, the \textbf{plan informativeness} is defined as
$U(\Pi) = \min{Dist(\pi_i, \pi_j)}$ where $\pi_i\in \Pi, \pi_j \in \Pi.$\/}
\end{definition}

\noindent\textbf{Incorporating user preferences.} Consider $U_r$ in Defn.~\ref{iAPQ}. Notably, \textsc{i-tips} allows a user to explore \textit{AP} iteratively based on what has been revealed so far. We take into account user preferences in the form of $r\in\{1,0\}$ (\ie prefer or not) during the exploration. For example, a user may want to find only those plans with different join orders but not others. This scenario can be easily addressed by \textsc{i-tips}, which allows one to mark whether a newly generated plan is preferred. Consider Fig.~\autoref{subfig:fitted} where each \textit{AP} is represented as a 3-dimensional vector $\vec{\pi}$ \wrt the \textsc{qep} ($\vec{\pi^*}=\langle 0,0,0\rangle$). According to Defn.~\ref{def:relevance}, a preferred \textit{AP} (as well as its neighbors) should receive higher relevance scores (warmer in the heatmap). Intuitively, finding the category of the preferred \textit{AP} (\ie $AP_{II}\sim AP_{IV}, AP_{VII}$), say $AP_{II}$ (according to Fig.~\ref{subfig:dp}, we can denote $\vec{AP_{II}}=\langle 0,0,1\rangle$), and moving the next plan towards (\resp away from) it if $r=1$ (\resp $r=0$), \eg $\vec{\pi}=\vec{\pi}+\epsilon(\vec{AP_{II}}-\vec{\pi})$ would address this ($\epsilon=0.05$ by default). However, this strategy can affect only a single \textit{AP} in the current iteration but not later. Therefore, instead of moving a single \textit{AP}, we choose to update the origin of the 3-dimensional system, \ie $\vec{\pi^*}=\vec{\pi^*}-\epsilon(\vec{AP_{pre}}-\vec{\pi})$, where $\pi$ refers to the latest plan seen so far, and $AP_{pre}$ refers to the category to which $\pi$ belongs. To find $AP_{pre}$ for a plan $\pi$, we only need to perform a nearest neighbor search over the set of $\{\vec{AP_{II}}, \vec{AP_{III}}, \vec{AP_{IV}}, \vec{AP_{VII}}\}$ (each can be encoded according to Fig.~\ref{subfig:dp}) with respect to $\vec{\pi}$. Algorithm~\ref{alg: update_u} formalizes this update of $U_r$.

\setlength{\textfloatsep}{2pt}
\begin{algorithm}[t]
\scriptsize
\caption{Algorithm $update\_U$} %
\label{alg: update_u}

 \KwIn{\textsc{qep} $\pi^*$, user's rating for current \textsc{aqp} $r$}
 \KwOut{$U_{r}(\cdot)$}
 
 Let $\pi^{(i-1)}$ represents the current \textsc{aqp}.\;
 $AP_{pre}$ = $argmin\ dist(AP, \pi^{(i-1)})$ subject to $AP\in \{AP_{II}, AP_{III}, AP_{IV}, AP_{VII}\}$\;
 Let $\Delta_\pi=\epsilon(\vec{AP_{pre}}-\vec{\pi^{(i-1)}})$\;
 update the origin (\textsc{qep}) : $\vec{\pi^*}=(r==1)?(\vec{\pi^*}-\Delta_\pi):(\vec{\pi^*}+\Delta_\pi)$\;
\end{algorithm}

\section{\textsc{tips} Algorithms}\label{ssec:plansel}
As discussed in Section~\ref{secintro}, we cannot integrate the selection of informative plans into the optimizer’s enumeration step because these plans are chosen based on the \textsc{qep} and \textsc{aqp}s observed by a learner. Moreover, query optimizers search and filter by cost, whereas we are not merely seeking the lowest-cost plans. Thus, new algorithms are needed for the \textsc{tips} problem. We propose \textsc{i-tips} and \textsc{b-tips} for the two problem variants in Section~\ref{sec:problem_def}.

We first present two pruning strategies--\textit{group forest-based pruning} (\textsc{gfp}) and \textit{importance-based plan sampling} (\textsc{ips})--to enable efficient exploration of the candidate plan space and support queries with many joins. We then describe the \textsc{i-tips} and \textsc{b-tips} algorithms, which leverage these strategies to select informative plans.
It is important to note that the pruning strategies are essential for maintaining acceptable response times even for queries with a small number of joins (\eg 3–4) that exhibit join-order effects. Such queries may generate thousands of alternative plans due to join-order permutations, access-path choices, and join algorithms~\cite{DBLP:conf/icde/GraefeM93,DBLP:books/daglib/0011128}. As reported in Section~\ref{sec:ustudy}, \textsc{gfp} keeps the per-\textsc{aqp} wait time around one second.

We begin by introducing our framework to retrieve a collection of candidate \textsc{aqp}s $\Pi^\ast$ from the underlying \textsc{rdbms}. These candidate plans are passed on as input to our proposed algorithms.
\subsection{Candidate Plan Set Retrieval} \label{sec:framework} 
\eat{\noindent\textbf{Candidate Plan Set Retriever.}} To obtain the plan space for a given \textsc{sql} query (\ie $\Pi^*$), we adopt \textsc{orca}~\cite{orca_mo_2014}, which is a modular top-down query optimizer based on the cascades optimization framework~\cite{cascade_go_1995}. \eat{We adopt it for the following reasons. First, it is a stand-alone optimizer and is independent of the \textsc{rdbms} a learner interacts with. Consequently, it facilitates portability of \textsc{arena} across different \textsc{rdbms}. Second, it interacts with the \textsc{rdbms} through a standard \textsc{xml} interface such that the effort to support a new \textsc{rdbms} is minimized.} We choose to adopt \textsc{orca} because it is a stand-alone optimizer that is portable across different \textsc{rdbms} systems, thanks to its interaction through a standard \textsc{xml} interface. The only task one needs to undertake is to rewrite the parser that transforms a new format of query plan  (\eg \textsc{xml} format in \textit{SQL Server} and \textit{PostgreSQL}) into the standard \textsc{xml} interface of \textsc{orca}. Consequently, this makes \textsc{tips} compatible with many \textsc{rdbms} as a parser can be implemented without modifying the internals of the target \textsc{rdbms}. Specifically, we have implemented this interface on top of \textit{PostgreSQL} and \textit{Greenplum}. 

It is worth noting that the choice of candidate plan set retrieval technique is orthogonal to our work. Other frameworks with similar functionality can also be adopted.

\noindent\textbf{\textit{Remark}.} There are alternative ways to retrieve candidate plans--such as using query hints, altering data statistics. However, such strategies obtain a very limited number of candidates; among them, there are even fewer informative plans. We may also obtain the plans by modifying some physical operators or join order of the \textsc{qep}. 
This offers no advantage over our approach, restricts candidate plan diversity, and makes it challenging to identify modifications that yield informative plans.\eat{ Even if we know what to modify, since the method of selecting alternative plans is fixed, the information obtained is also fixed and limited. Therefore, these alternatives are not suitable for our problem.}

\subsection{Group Forest-based Pruning (GFP)}
\label{ssec:filter}

Because the plan space can be very large, traversing all \textsc{aqp}s to find informative plans can hurt efficiency. We therefore need a way to quickly discard uninformative \textsc{aqp}s. As shown in Section~\ref{sec:feedback}, $AP_{\uppercase\expandafter{\romannumeral6}}$ and $AP_{\uppercase\expandafter{\romannumeral8}}$ are uninformative, and both have large \textsf{LogiPln} differences from the \textsc{qep}. A simple approach is to compare each \textsc{aqp} with the \textsc{qep} and discard those with large differences, but this still requires scanning the entire plan space. To address this, we propose a \textit{group forest-based pruning} (\textsc{gfp}) strategy, which both reduces the number of \textsc{aqp}s and avoids accessing the entire plan space.

Intuitively, the idea is as follows. Reconsider Example~\ref{eg1}. \textsc{gfp} first groups \textsc{aqp}s according to their \textit{memo-derived group trees}. If a group tree differs significantly from the \textsc{qep} in terms of \textsf{LogiPln} or \textsf{Cost} (exceeding $\tau_d$ or $\tau_c$), all plans in that group--typically $AP_{\uppercase\expandafter{\romannumeral6}}$/$AP_{\uppercase\expandafter{\romannumeral8}}$-like--are pruned, while more informative categories such as $AP_{\uppercase\expandafter{\romannumeral2}}$ and $AP_{\uppercase\expandafter{\romannumeral4}}$ are retained. We elaborate on this strategy now.\eat{Reconsider Example~\ref{eg1}: GFP first groups AQPs by their memo-derived group trees. If a group tree is far from the QEP in LogiPln and cost (exceeding $\tau_d,\tau_c$), all plans in that group (typically $AP_{\uppercase\expandafter{\romannumeral6}}/AP_{\uppercase\expandafter{\romannumeral8}}$-like) are pruned, while informative categories such as $AP_{\uppercase\expandafter{\romannumeral2}}$ and $AP_{\uppercase\expandafter{\romannumeral4}}$ remain. We elaborate on this strategy now.} 

\begin{figure}[t]
  \centering
  \subfloat[MEMO Structure.]{\label{fig:memo}
  \includegraphics[width=0.45\columnwidth]{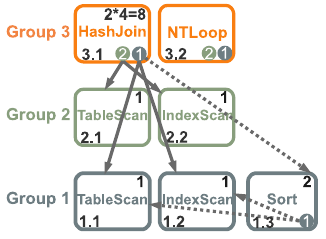}
  }
  \subfloat[Group Forest.]{\label{fig:groupTree}
  \includegraphics[width=0.3\columnwidth]{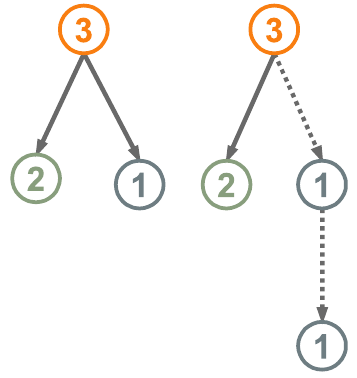}
  }
  \caption{Example of MEMO and Group Forest.}
  \label{fig:memoAndGt}
\end{figure}

\eat{We can leverage the method in~\cite{counting_fl_2000} to count and traverse the entire plan space by assigning to each plan a unique \textit{id}, starting from 1.  Furthermore,} Cascade optimizers~\cite{cascade_go_1995} utilize a structure \textsc{memo} (\ie a forest of operator trees) to track all possible plans. A \textsc{memo} consists of a series of \textit{groups}, each of which contains a number of \textit{expressions}. An \textit{expression} is used to represent a logical/physical operator. Fig.~\ref{fig:memo} shows an example of \textsc{memo}, where each row represents a group and the square represents an \textit{expression} in it. Each expression is associated with an \textit{id} in the form of \textsf{`GroupID.ExpressionID'}, shown at the lower left corner (\eg 3.1, 2.1). The lower right corner of an expression shows the id(s) of the child group(s) to which it points. Different expressions in the same group are different operators that can achieve the same functionality. For example, both 3.1 (\texttt{HashJoin}) and 3.2 (\texttt{NTLoop}) implement the join operation.  Suppose \textit{Group 3} is the \textit{root} group. By traversing down from the expression in the root group, we can get a plan. For example, the id sequence \textit{(3.1, 2.1, 1.1)} represents the plan \texttt{HashJoin[TableScan, TableScan]}. If we change one of the expressions, a different plan can be achieved. For example, suppose the child for \bcirclednumber{1} of expression 3.1 from 1.1 to 1.3 (as shown by the dotted line in Fig.~\ref{fig:memo}), we get a different plan \texttt{HashJoin[TableScan, Sort[TableScan]]} or \texttt{HashJoin[TableScan, Sort[indexScan]]}. The upper right corner of each expression records the number of alternatives with this expression as the root node. For example, the number of expression 1.3 is 2, which includes \textit{1.3$\rightarrow$1.1} and \textit{1.3$\rightarrow$1.2}. 

Observe that the number of alternatives of expression 3.1 in the root node is 8. That is, there are 8 different plans rooted at expression 3.1. However, there are only 2 tree structures in all these 8 plans, as shown in Fig.~\ref{fig:groupTree}. The number in each node represents the id of the group. Since these trees represent connections between different groups, we collectively refer to them as \textit{group forest}. Each \textit{group tree} in a group forest represents multiple plans with the same structure, and if the \textsf{LogiPln} distance of a group tree and the \textsc{qep} is larger than a given \textit{threshold} $\tau_d$, \ie $s\_dist(\pi,\pi^*)\ge\tau_d$, all the corresponding ids can be filtered out. Observe that the number of group trees is significantly smaller than the number of \textsc{aqp}s. Therefore, uninformative plans can be quickly eliminated accordingly. To construct the \textit{group forest}, we only need to traverse \textsc{memo} bottom-up, which is extremely fast, \ie time cost is linear to the height of the constructed tree.\eat{ Empirical results shows that even if the plan space reaches millions, it can still be constructed in 2s.}

Lastly, observe that if we filter based on $s\_dist$ only, informative plans of types $AP_{\uppercase\expandafter{\romannumeral4}}$ and $AP_{\uppercase\expandafter{\romannumeral7}}$ \eat{(Table~\ref{tbl:eightAP}) }can be abandoned. The difference between these two types of plans and $AP_{\uppercase\expandafter{\romannumeral6}}$ and $AP_{\uppercase\expandafter{\romannumeral8}}$ is that the cost is smaller. Hence, we filter using both the \textsf{LogiPln} and \textsf{Cost} differences, \ie $s\_dist(\pi,\pi^*)\ge\tau_d, cost\_dist(\pi,\pi^*) \ge \tau_c$ where $\tau_c$ is a pre-defined \textit{cost threshold}. 

\begin{algorithm}[t]
\scriptsize
\caption{Algorithm \textsc{ips}}
\label{alg:aos}

 \KwIn{\textsc{qep} $\pi^{\ast}$,  collection of all \textsc{aqp}s $\Pi^{\ast}$}
 \KwOut{sampling of $\Pi^{\ast}$}
    $\Pi_1 \gets \emptyset$\;
    \eat{\tcc{In practice, instead of traversing the entire space, we obtain the plans with the same LogiPln as QEP through the data structure inside the optimizer}}
    \For{$\pi_i$ in $\Pi^{\ast}$}
    {
           \If{$s\_dist(\pi_i, \pi^*)=0$} {$\Pi_1 = \Pi_1 \cup \{\pi_i\}$}
    }
    
    $\Pi_2 = UniformSampling(\Pi^{\ast})$\;
    return $\Pi_1 \cup \Pi_2$\;
\end{algorithm}

\subsection{Importance-based Plan Sampling (IPS)}

Enumerating the full plan space is prohibitively expensive for queries with many joins (\ie more than 10 relations). Although sampling can reduce this cost, simply deploying random sampling risks discarding many informative plans. We propose an \textit{importance-based plan sampling} (\textsc{ips}) strategy that exploits the insights gained from feedback from volunteers in Section~\ref{sec:feedback} to select plans for subsequent processing. The overall goal is to preserve the informative plans as much as possible during the sampling.  However, retaining all informative plans will inevitably demand traversal of every \textsc{aqp} to compute plan informativeness. Hence, instead of preserving all informative plans, we choose to keep the most important ones when the number of joined relations in a query is large. We exploit the feedback in Section~\ref{sec:feedback} to sample plans. Observe that $AP_{\uppercase\expandafter{\romannumeral2}}$ plans receive a significant number of votes among all informative plans. Hence, we inject such plans as much as possible in our sample. Also,  $AP_{\uppercase\expandafter{\romannumeral2}}$ plans have small LogiPln differences with the \textsc{qep}. Hence, given a \textsc{qep}, we can easily determine whether an \textsc{aqp} belongs to $AP_{\uppercase\expandafter{\romannumeral2}}$ or not by finding all alternative plans with the same LogiPln as \textsc{qep}. These plans are then combined with uniformly random sampled results. For instance, reconsider Example~\ref{eg1}. \textsc{ips} preserves \textsc{aqp}s that share the same \textsf{LogiPln} as the \textsc{qep} (capturing $AP_{\uppercase\expandafter{\romannumeral2}}$-like plans with large cost differences) and augments them with uniform samples to ensure that informative but rare categories are retained. The pseudocode of the \textsc{ips} strategy is shown in Algorithm~\ref{alg:aos}.

\begin{algorithm}[t]
\scriptsize
\caption{Algorithm \textsc{i-tips}}
\label{alg: iaqps}

 \KwIn{\textsc{qep} $\pi^*$, collection of all \textsc{aqp}s $\Pi^*$, user's rating $r$, \textsc{ips} threshold $\tau_\ell$, \textsc{gfp} threshold $\tau_g$}
 \KwOut{an alternative informative query plan $\pi^{(i)}$ during iteration $i$}
 
 \If{the number of join tables > $\tau_\ell$}{$\Pi^* =$ \textsc{ips}($\pi^*$, $\Pi^*$)\; }
 
 \If{ $|\Pi^*| > \tau_g$}{$\Pi^* =$\textsc{gfp}($\pi^*$, $\Pi^*$)  \tcp{Group forest-based pruning}}
 $update\_U$($\pi^*$, $r$) \tcp{Algorithm~\ref{alg: update_u}}
 \For{$\pi_j \in \Pi^{\ast} -\Pi$}
    {
    $d_j \gets Dist(\vec{\pi_j}, \vec{\pi^*})$
    
    }
    
  Find the plan $\pi^{(i)} \in \Pi^{\ast}-\Pi$ with the maximum $d$\;
  $\Pi=\Pi\bigcup\{\pi^{(i)}\}$\;
  return $\pi^{(i)}$
\end{algorithm}

\subsection{\textsc{i-tips} Algorithm}
Algorithm~\ref{alg: iaqps} outlines the procedure to address the \textsc{i-tips} problem defined in Defn.~\ref{iAPQ}. If the number of joined relations is greater than the predefined \textit{\textsc{ips} threshold} $\tau_\ell$, we invoke \textsc{ips} (Line 2). In addition, if the number of \textsc{aqp}s is greater than the predefined \textit{\textsc{gfp} threshold} $\tau_g$, it invokes the GFP strategy (Line 5), to filter out some uninformative plans early. 

Notably, \textsc{i-tips} allows a user to explore \textit{AP} iteratively based on what has been revealed so far. We take into account the user preference in the form of rating $r$ during exploration. Algorithm \ref{alg: update_u} shows the process of updating the plan informativeness, and the procedure strictly follows what we have discussed after Defn.~\ref{planinfor}.
Reconsider Example~\ref{eg1}. After Lena marks $AQP_2$ as preferred ($r=1$) plan, \textsc{i-tips} shifts the reference point toward its category, guiding subsequent suggestion towards nearby informative categories instead of repeating near-duplicate plans.



\begin{algorithm}[t]
	\scriptsize
	\caption{Algorithm \textsc{b-tips-heap}}
	\label{alg: greedy}
	
	\KwIn{\textsc{qep} $\pi^*$, collection of all \textsc{aqp}s $\Pi^*$, budget $k$,  \textsc{ips} threshold $\tau_\ell$, \textsc{gfp} threshold $\tau_g$}
	\KwOut{top-k informative query plans $\Pi$}
	
	Lines 1-6 in Algorithm~\ref{alg: iaqps}\;
	
	set $Heap$ as a max-heap  \tcp{record the minimum distance}
	$\Pi \gets \{\pi^*\}$ \\
	
	\ForEach{$\pi_j \in \Pi^*-\{\pi^*\}$}{
		$d \gets Dist(\pi^*, \pi_j)$\;
		$Heap.push(tuple(d, \pi_j))$\;
	}

	\While{$\left|\Pi\right| \neq k+1$}	{
		set $tmpRecord$ is $\emptyset$\;
		set $optTuple$ is None\;
		
		\While{$Heap$ not empty}{
			$tmpTuple \gets Heap.pop()$\;
			\If{$tmpTuple < optTuple$}{
				break;
			}
			\If{$tmpTuple$ is not latest}{
				$tmpTuple \gets Update(tmpTuple)$\;
			}
			\If{$tmpTuple > optTuple$}{
				$optTuple \gets tmpTuple$\;
			}
			$tempRecord.push(tempTuple)$\;
		}
		$\Pi \gets \Pi \cup \{optTuple.\pi\}$\;
		\tcc{plans that are removed from the heap but not in use are put back into the heap}
		\ForEach{$t$ in $tmpRecord-\{optTuple\}$}{
			$Heap.push(t)$ \;
		}
	}

	return $\Pi\setminus\{\pi^*\}$
\end{algorithm}
\subsection{\textsc{b-tips} Algorithm}
\eat{As discussed in Section~\ref{sec:problem_def}, the \textsc{b-tips} problem is NP-hard. Hence, the practical way to address it is to find a polynomial solution, either heuristic or approximate.} We propose two flavors of the algorithm to address \textsc{b-tips}, namely, \textsc{b-tips-basic} and \textsc{b-tips-heap}. 

\noindent\textbf{\textsc{B-TIPS-basic} algorithm.} We can exploit Algorithm~\ref{alg: iaqps} to address the \textsc{b-tips} problem. We first add the \textsc{qep} $\pi^*$ to the final set $\Pi$. Then iteratively invoke it until there are $k+1$ plans in $\Pi$. Each time we just need to add the result of Algorithm~\ref{alg: iaqps} to $\Pi$. This algorithm has a time complexity of $O\left(k|\Pi^{\ast}|\right)$.

\noindent\textbf{\textsc{B-TIPS-heap} algorithm.} 
Since the minimum distance between each unselected plan in $\Pi^* \setminus \Pi$ and the selected set $\Pi$ decreases monotonically as $|\Pi|$ increases, we can use a max-heap to filter plans with very small distances to avoid unnecessary computation. Algorithm~\ref{alg: greedy} outlines the procedure. We first apply the two pruning strategies (Line 1). In Line 6, we insert into the heap a tuple containing the ID of an alternative plan $\pi$ and its distance, defined as the minimum distance between $\pi$ and the plans in $\Pi$. Each time we add a plan to $\Pi$, we do not recompute all distances; instead, we only verify whether each stored distance is still valid (Line 17) and update it if necessary (Line 18). At each iteration, we select the alternative plan with the largest minimum distance. For instance, in Example~\ref{eg1}, when $k=2$, \textsc{b-tips} prefers sets such as $\{AQP_1, AQP_2\}$ or $\{AQP_2, AQP_3\}$, which collectively maximize the minimum relevance-distance trade-off, instead of selecting two highly similar plans. Observe that the worst-case time complexity is also $O\left(k|\Pi^{\ast}|\right)$. However, it is significantly more efficient than \textsc{b-tips-basic} in practice (detailed in Section~\ref{sec:experiment}).


For a \textit{MaxMin} problem, a polynomial-time approximation algorithm provides a performance guarantee of $\rho \ge 1$, if for each instance of the problem, the minimum distance of the result set produced by the algorithm is at least $1/\rho$ of the optimal set minimum distance. We will also refer to such an algorithm as a $\rho$-approximation algorithm.

\begin{theorem}
	\label{the:approximate}
\textit{\textsc{b-tips} is a 2-approximation algorithm if the distance satisfies the triangle inequality (Lemma~\ref{disttri}).}
\end{theorem}

\noindent\textbf{\textit{Remark}.} Our algorithms are generic and not limited to database education. First, they are orthogonal to any candidate plan set retrieval strategy. Second, the \textsc{gfp} and \textsc{ips} strategies—though illustrated with education-focused plan categories—can be easily adapted to other applications (\eg database administration) with different informative and uninformative plan categories derived from a human-centered characterization of \textsc{aqp}s. They are also easily adaptable to non-human-centered characterizations of \textsc{aqp}s by modifying the if-statement in Line 3 of Algorithm~\ref{alg:aos}.

\begin{figure*}[t]
  \centering
  \subfloat[tpch\_18]{\label{fig:infor_18}
  \includegraphics[width=0.2\linewidth, height=2.1cm]{revision_material/Exp1/effect/tpch/18/effect\_tpch\_18.pdf}}
  \subfloat[tpch\_17]{\label{fig:infor_17}
  \includegraphics[width=0.2\linewidth, height=2.1cm]{revision_material/Exp1/effect/tpch/17/effect\_tpch\_17.pdf}}
    \subfloat[tpch\_11]{\label{fig:infor_11}\includegraphics[width=0.2\linewidth, height=2.1cm]{revision_material/Exp1/effect/tpch/11/effect\_tpch\_11.pdf}} 
\subfloat[tpch\_15]{\label{fig:infor_15}\includegraphics[width=0.2\linewidth, height=2.1cm]{revision_material/Exp1/effect/tpch/15/effect\_tpch\_15.pdf}} 
\subfloat[tpch\_22]{\label{fig:infor_22}\includegraphics[width=0.2\linewidth, height=2.1cm]{revision_material/Exp1/effect/tpch/22/effect\_tpch\_22.pdf}} \\ 
  \subfloat[tpch\_18]{\label{fig:rtime_18}
  \includegraphics[width=0.2\linewidth, height=2.1cm]{revision_material/Exp1/efficiency/tpch/18/efficiency\_tpch\_18.pdf}
  }
  \subfloat[tpch\_17]{\label{fig:rtime_17}
  \includegraphics[width=0.2\linewidth, height=2.1cm]{revision_material/Exp1/efficiency/tpch/17/efficiency\_tpch\_17.pdf}
  }
\subfloat[tpch\_11]{\label{fig:rtime_11}\includegraphics[width=0.2\linewidth, height=2.1cm]{revision_material/Exp1/efficiency/tpch/11/efficiency\_tpch\_11.pdf}} 
\subfloat[tpch\_15]{\label{fig:rtime_15}\includegraphics[width=0.2\linewidth, height=2.1cm]{revision_material/Exp1/efficiency/tpch/15/efficiency\_tpch\_15.pdf}}
\subfloat[tpch\_22]{\label{fig:rtime_22}\includegraphics[width=0.2\linewidth, height=2.1cm]{revision_material/Exp1/efficiency/tpch/22/efficiency\_tpch\_22.pdf}}
\caption{(a-e) Plan informativeness; (f-j) Running time (for the 5 \textsc{tpc-h} queries in the first column of Table~\ref{tab:imdb_space}).}
	\label{fig:informative}
\end{figure*}

\begin{figure}[t]
  \centering
  \subfloat[tpcds\_13]{\label{fig:infor_tpcds_13}
  \includegraphics[width=0.24\linewidth]{revision_material/Exp1/effect/tpcds/13/effect\_tpcds\_13.pdf}}
\subfloat[tpcds\_47]{\label{fig:infor_tpcds_47}\includegraphics[width=0.24\linewidth]{revision_material/Exp1/effect/tpcds/47/effect\_tpcds\_47.pdf}}
  \subfloat[tpcds\_13]{\label{fig:rtime_tpcds_13}
  \includegraphics[width=0.24\linewidth]{revision_material/Exp1/efficiency/tpcds/13/efficiency\_tpcds\_13.pdf}
  }
\subfloat[tpcds\_47]{\label{fig:rtime_tpcds_47}\includegraphics[width=0.24\linewidth]{revision_material/Exp1/efficiency/tpcds/47/efficiency\_tpcds\_47.pdf}}
  \caption{(a) Plan informativeness; (b) Running time for two TPC-DS queries (13, 47) in Table~\ref{tab:imdb_space}.} %
	\label{fig:informative_tpcds}
\end{figure}

\begin{figure}[t]
  \centering
\subfloat[Initialization]{\label{diff_tree_kernel_a}\includegraphics[width=0.44\columnwidth]{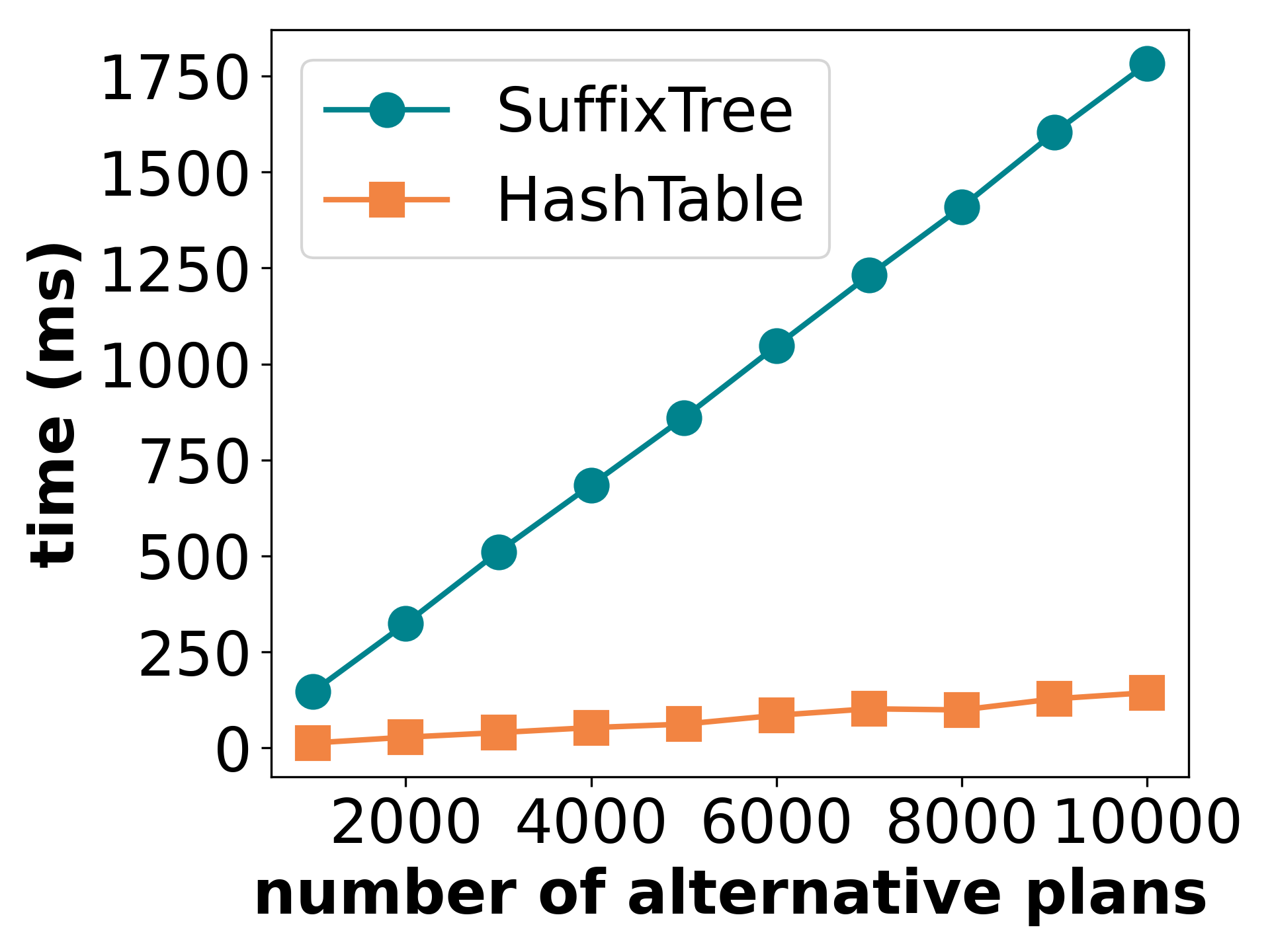}}
\subfloat[Calculate subtree kernel]{\label{diff_tree_kernel_b}\includegraphics[width=0.44\columnwidth]{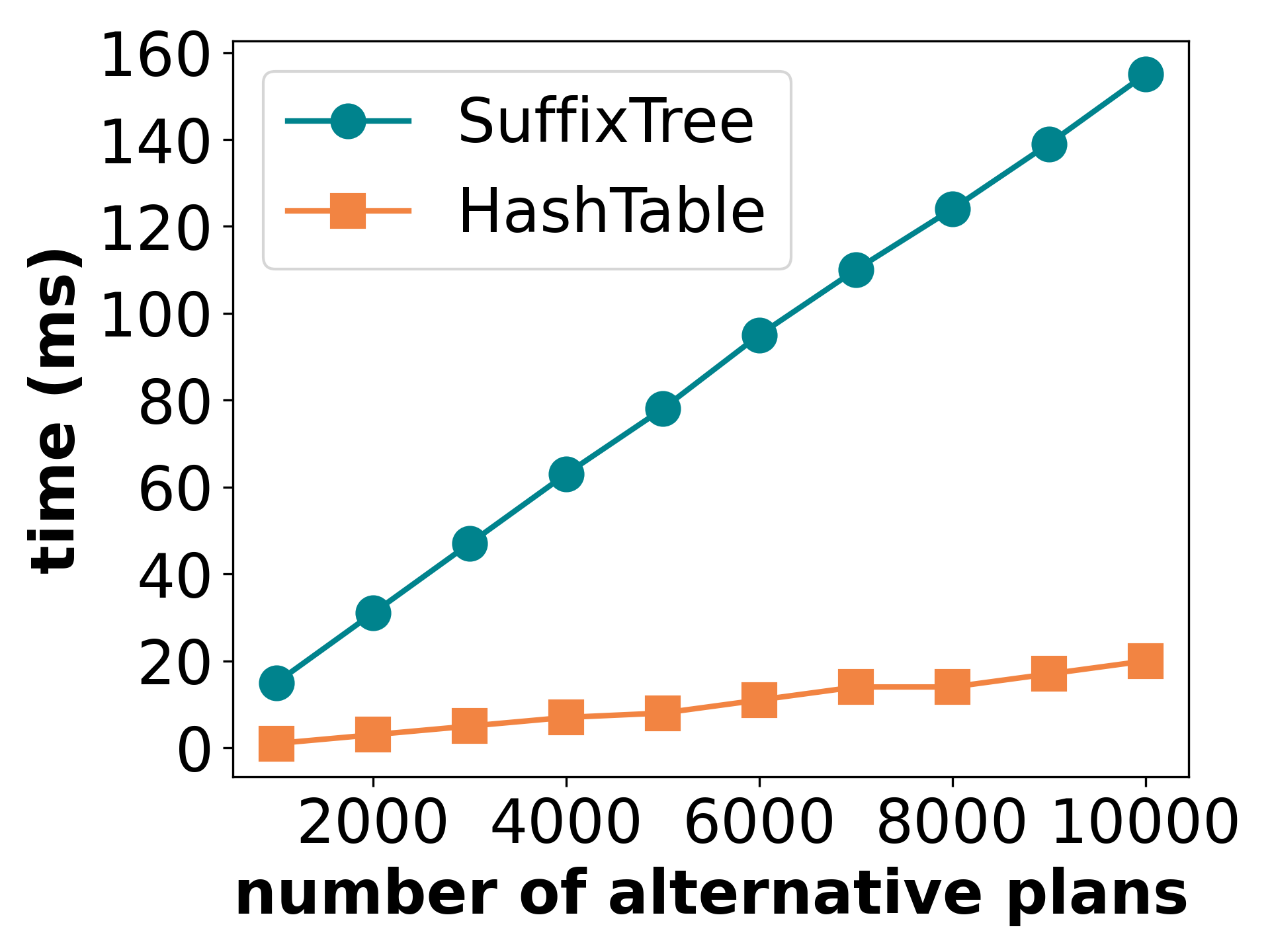}}
	\caption{Subtree kernel.}
	\label{fig:diff_tree_kernel}
\end{figure}

\section{Performance Study}\label{sec:experiment} 
In this section, we investigate the performance of the \textsc{tips} algorithms\footnote{The source code can be found in \url{https://github.com/ZiHao256/TIPS}}. In the next section, we report their usefulness and applicability in database education.
\subsection{Experimental Setup}
\noindent\textbf{Datasets.}  We use two datasets. The first is the Internet Movie Data Base (IMDb) dataset~\cite{imdb,good_viktor_2015}, \eat{ The data is freely available at \url{ftp://ftp.fu-berlin.de/pub/misc/movies/database/} for non-commercial use.} \eat{, which consumes 3.6  GB of space when exported to \textsc{csv} files} containing the two largest tables, \textsf{cast\_info} and \textsf{movie\_info} have 36M and 15M rows, respectively. The second one is the TPC-H dataset~\cite{tpch}. We use the TPC-H v3 and 1 GB data generated using \textit{dbgen}~\cite{tpch}.

Since the number of candidate plans affects the performance of \textsc{tips}, which is mainly affected by the \textsc{sql} statement, we choose different \textsc{sql} queries to vary the plan size. For IMDb,  ~\cite{good_viktor_2015} provides 113 \textsc{sql} queries. We compute the plan-space size $|\Pi^*|$ for all 113 queries, bin the queries into eight equal-width groups based on $|\Pi^*|$ (\eg the first group includes queries with plan counts in the range [1500, 2500]), and select the median query in each bin (manually swap it out for a nearby query if its \textsf{LogiPln} is too similar to one that’s already been selected). For the TPC-H dataset, there are 22 standard \textsc{sql} statements. Because of the uneven distribution of the plan space of these queries, it is difficult to find \textsc{sql} statements that are evenly distributed as in IMDb dataset. We select 5 \textsc{sql} queries from them for our experiments. To cover more diverse schemas and query patterns, we additionally include five TPC-DS queries (13, 26, 34, 47, 48) as a complementary decision-support workload. These selected queries and their plan space sizes are reported in Table~\ref{tab:imdb_space}. 

\begin{table}[t]
  \centering 
  \small
  \caption{Datasets and plan-space sizes.}
  \resizebox{\columnwidth}{!}{%
  \begin{tabular}{lclclclc}
    \toprule
    SQL& $|\Pi^*|$ & SQL& $|\Pi^*|$ & SQL& $|\Pi^*|$ & SQL& $|\Pi^*|$ \\
    \midrule
    tpch\_18.sql & 2070 & tpcds\_48.sql & 956 & imdb\_6b.sql & 2059 & imdb\_1d.sql & 7065\\
    tpch\_17.sql & 2864 & tpcds\_34.sql & 2357 & imdb\_2a.sql & 3211 & imdb\_4b.sql & 7887\\
    tpch\_11.sql & 11880 & tpcds\_26.sql & 6177 & imdb\_5c.sql & 3957 & imdb\_1a.sql & 9007\\
    tpch\_15.sql & 13800 & tpcds\_13.sql & 11309 & imdb\_2b.sql & 4975 & imdb\_6e.sql & 10200\\
    tpch\_22.sql & 51870 & tpcds\_47.sql & 33769 &  &  &  & \\
    \bottomrule
  \end{tabular}}
  \label{tab:imdb_space}
\end{table}

\noindent\textbf{Baselines.} We are unaware of any existing informative plan selection technique. Hence, we are confined to comparing the \textsc{tips} algorithms with the following baselines. (a) \textsc{random}: We implemented a random selection algorithm. It is executed repeatedly $n$ times, and each time a result is randomly generated. The best result is returned after $n$ iterations. In our experiment, we set $n=30$. (b) \textsc{cost}: We implement a \textit{cost-based approach} that returns the top-$k$ plans with the least cost in $\Pi^\ast$ besides the \textsc{qep}. (c) \textsc{embedding}:  We also implemented an embedding-based approach, where we embed each \textsc{aqp} into a vector by feeding its estimated time cost and tree structure into an autoencoder following~\cite{AroraLM17} and then rank the \textsc{aqp}s according to their Euclidean distance \wrt the \textsc{qep}. Afterwards, the plan informativeness of the top-$k$ can be computed using Defn.~\ref{planinfor}. (d) \textsc{category-rep}: We partition \textsc{aqp}s into eight plan categories by thresholding \textsf{LogiPln}, \textsf{PhyOpr}, and \textsf{Cost} distances into small vs. large, using the same thresholds as \textsc{gfp}, and select one representative per category (the plan with the highest relevance within that category). If $k<8$, we take the top-$k$ representatives by relevance; if $k\ge 8$, we return all representatives and fill the remaining slots by relevance. This baseline ensures category coverage but does not optimize the global \textit{MaxMin} objective addressed by \textsc{tips}.

All algorithms are implemented in C++ on a 3.2GHz 4-core machine with 16GB RAM. We denote \textsc{b-tips-basic} and \textsc{b-tips-heap}  as \textsc{b-tips-b} and \textsc{b-tips-h}, respectively. Unless otherwise specified, we set $\alpha=0.33$, $\beta=0.33$, $\lambda=0.5$, $\tau_\ell = 10$, and $\tau_g = 50000$. 

\begin{figure*}[t]
  \centering
  \subfloat[the weight of relevance is 0]{\label{subfig:relevance_0}
  \includegraphics[width=.32\linewidth]{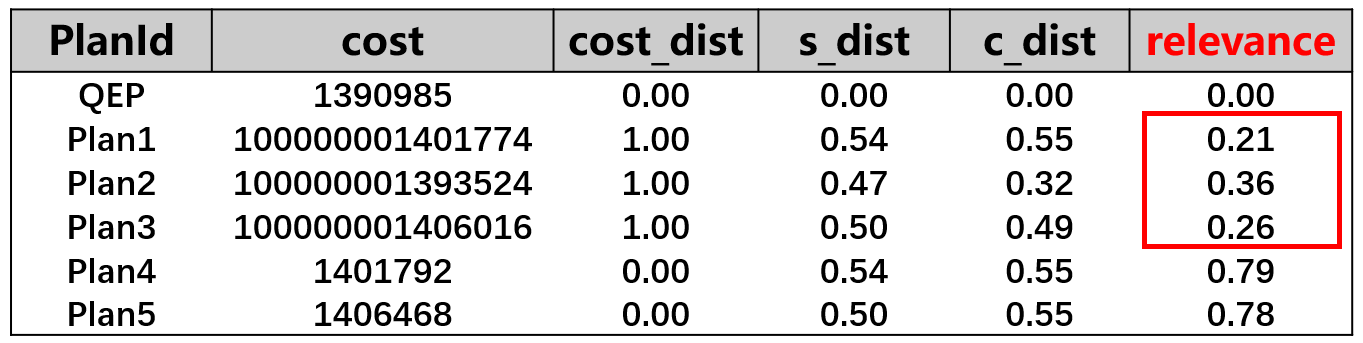}
  }
  \subfloat[the weight of relevance is 0.5]{\label{subfig:relevance_5}
  \includegraphics[width=.32\linewidth]{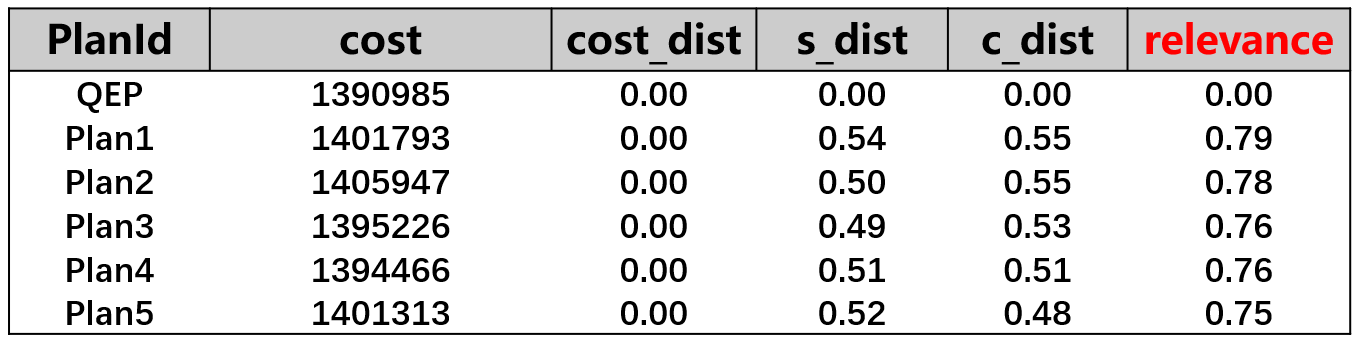}
  }
  \subfloat[the weight of relevance is 1]{\label{subfig:relevance_1}
  \includegraphics[width=.32\linewidth]{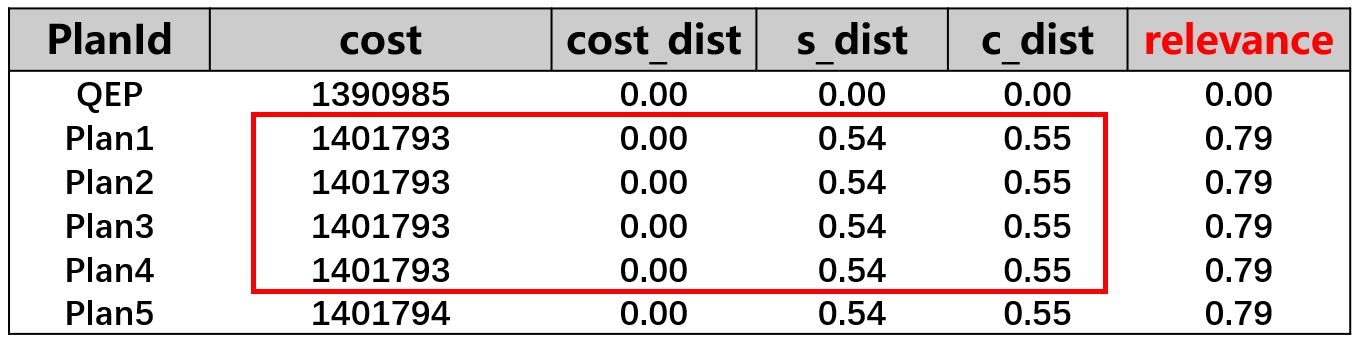}
  }
    \caption{Effect of relevance.}
  \label{fig:relevance}
\end{figure*}

\begin{figure*}[t]
	\centering
	\subfloat[Effect of pruning]{\includegraphics[width=0.25\columnwidth, height=2.5cm]{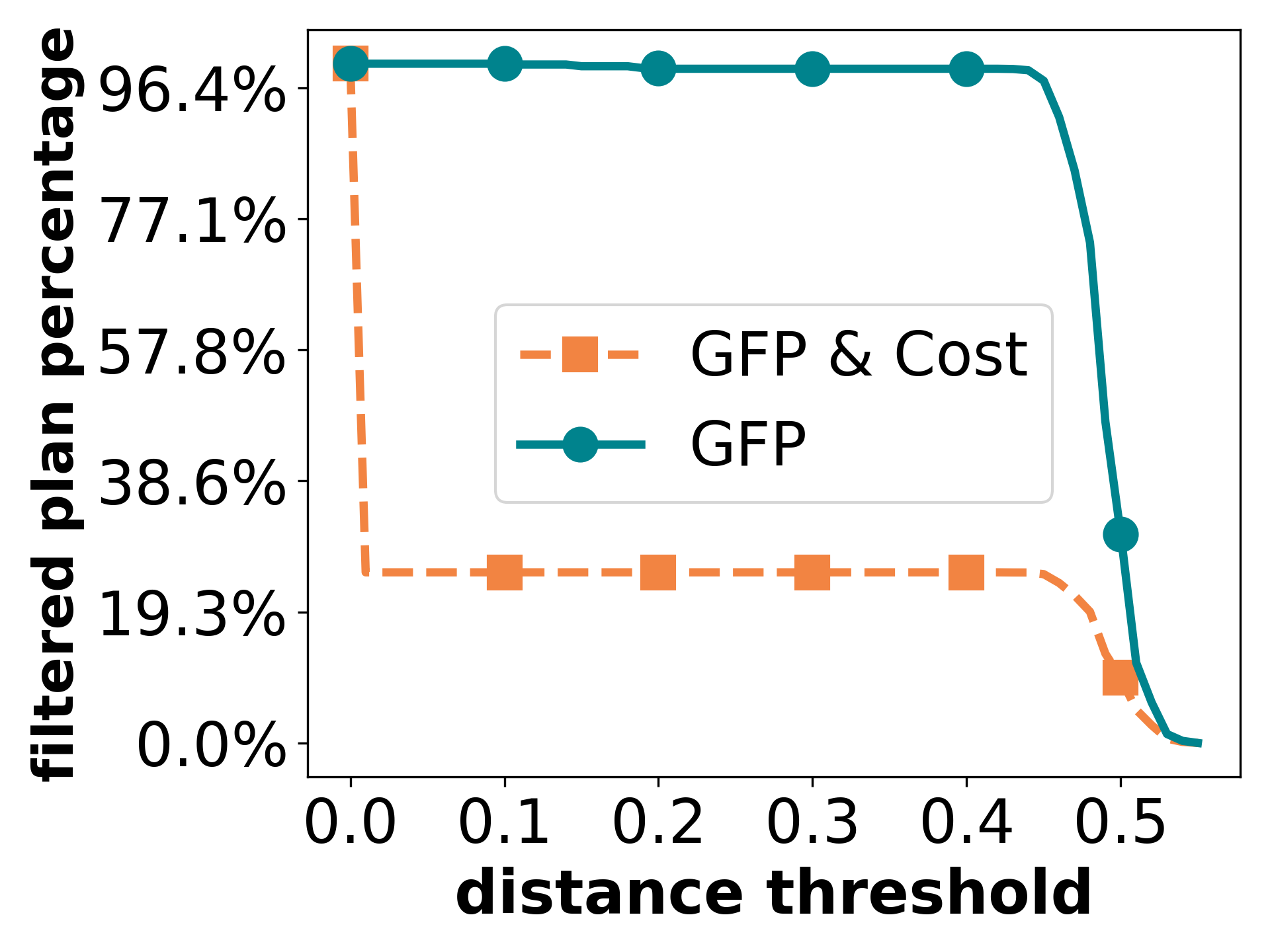}}
	\subfloat[Plan informativeness]{\includegraphics[width=0.25\columnwidth, height=2.5cm]{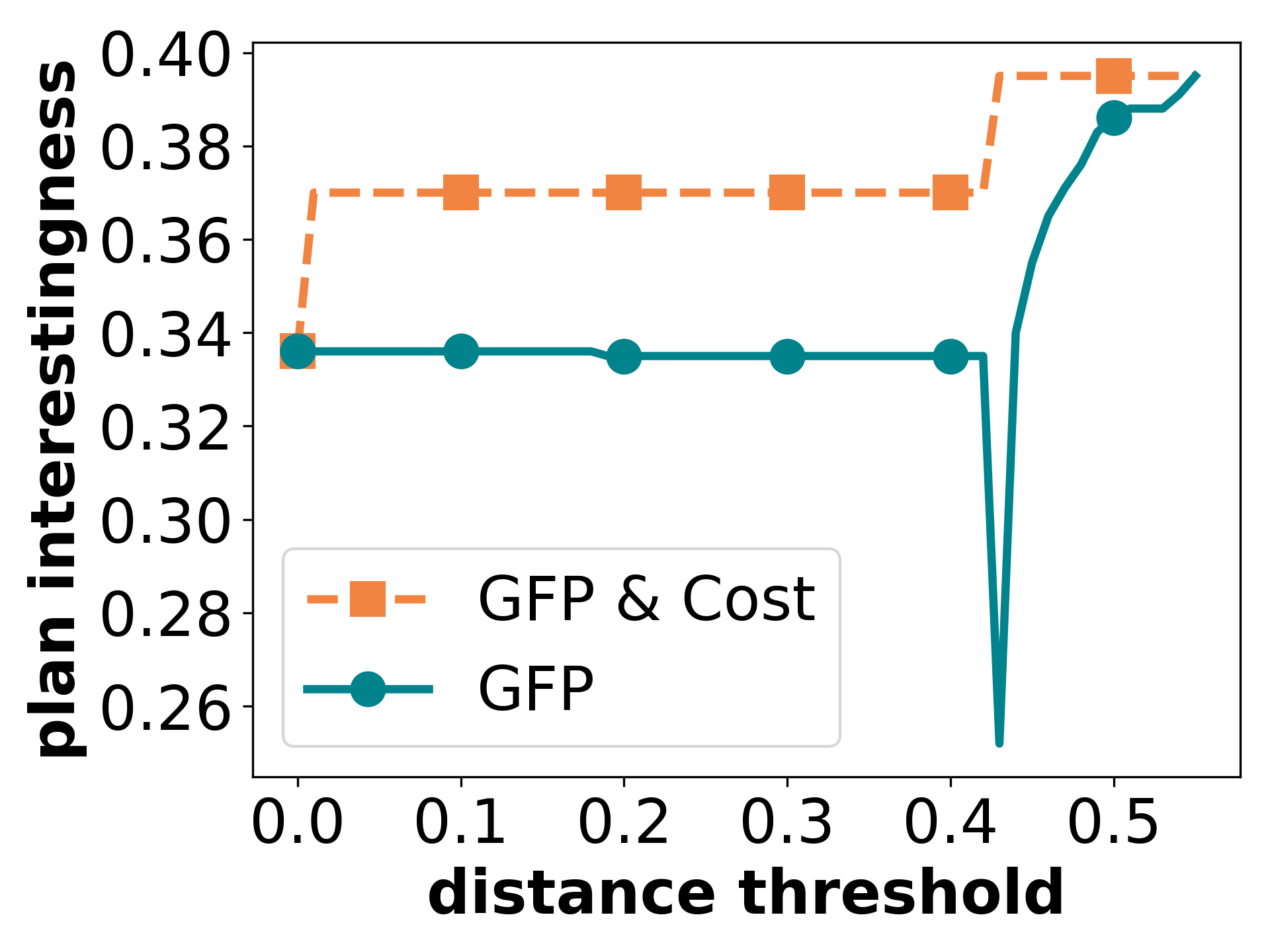}}
	\subfloat[Efficiency]{\includegraphics[width=0.25\columnwidth, height=2.5cm]{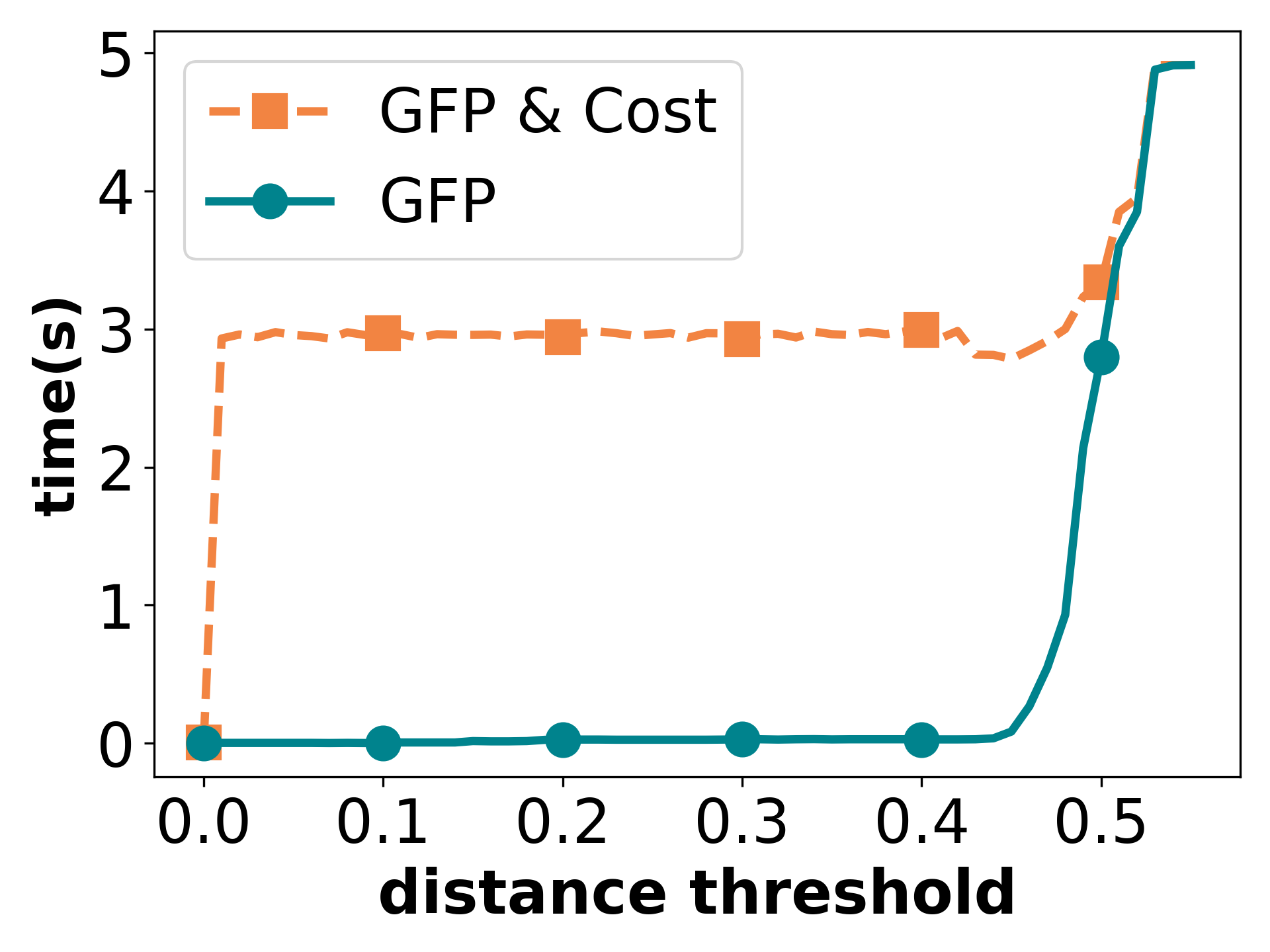}}
	\subfloat[Sample 1000/5000 plans]{\label{subfig:aos_1000_5000}\includegraphics[width=0.25\columnwidth, height=2.5cm]{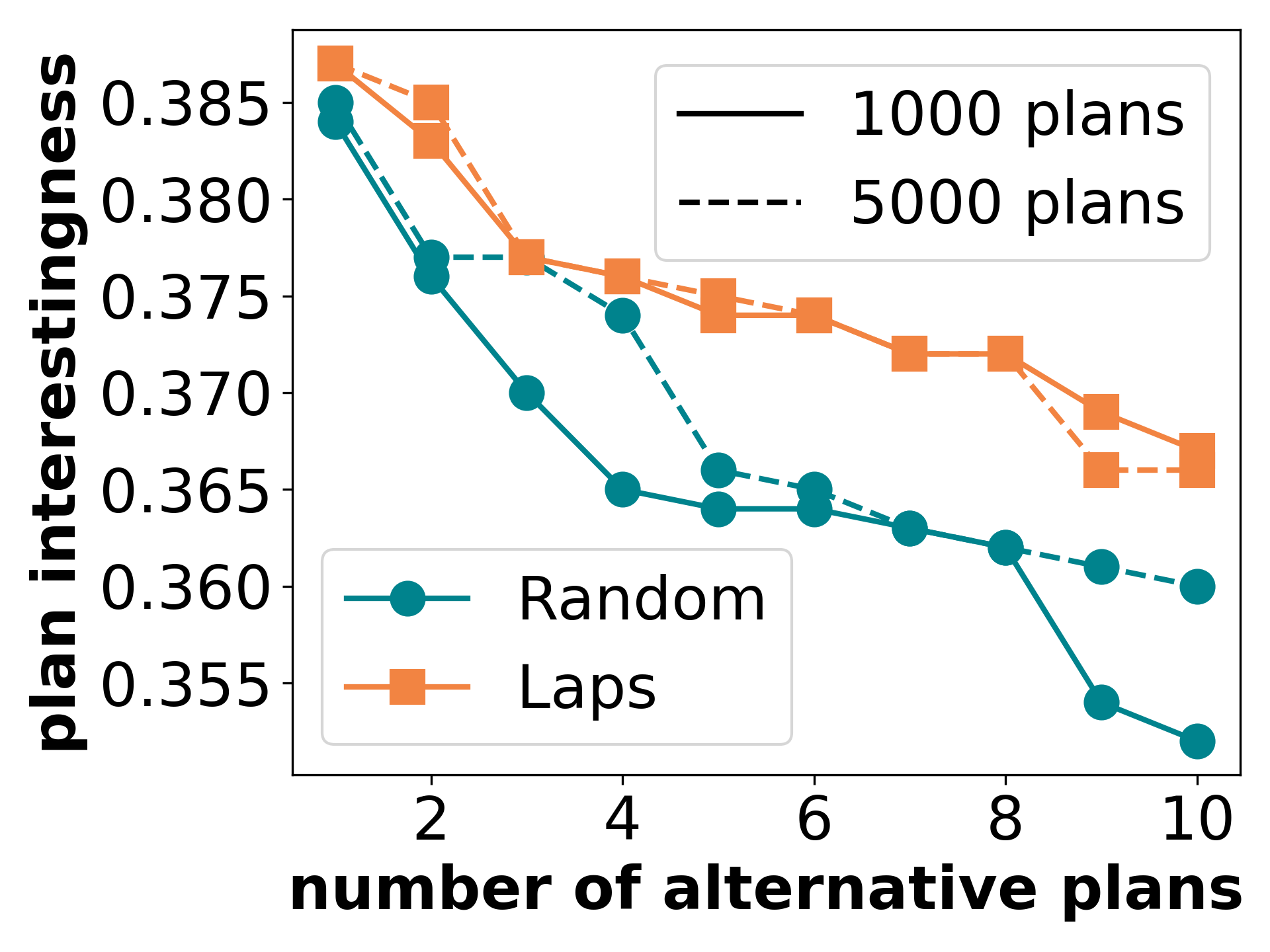}}
        \caption{Pruning strategies: (a--c) effect of \textsc{gfp} (filtered plan percentage, plan informativeness, and efficiency); (d) effect of \textsc{ips}.} %
	\label{fig:gtFilter}
\end{figure*}

\subsection{Experimental Results}\label{ssec:expres}

\noindent\textbf{Exp 1: Efficiency and effectiveness.} We first report the efficiency and effectiveness of different algorithms. For each query in Table~\ref{tab:imdb_space}, we use different algorithms to select $k$ plans and report the plan informativeness $U$ and runtime, shown in Fig.~\ref{fig:informative}(a-e) and ~\ref{fig:informative}(f-j), respectively.\eat{The results are plotted in Figures~\ref{fig:diff_tree_kernel}(a) and ~\ref{fig:diff_tree_kernel}(b). Since the results on different queries are qualitatively similar, we only show results for one query.} Figures \ref{fig:informative} and \ref{fig:informative_tpcds} report plan informativeness and runtime for five TPC-H queries and two representative TPC-DS queries, respectively, with similar trends observed for the remaining TPC-DS queries. We make the following observations. First, the plan informativeness obtained by \textsc{random}, \textsc{cost} and \textsc{embedding} are significantly inferior to that of the \textsc{b-tips} algorithms. \textsc{category-rep} improves coverage across plan categories but still underperforms \textsc{b-tips} because it does not optimize the global \textit{MaxMin} objective. Moreover, the effectiveness of \textsc{embedding} decreases with respect to $k$ unless the volume of the plan space is large enough, \ie $|\Pi^*|>10000$. Second, the running time of \textsc{b-tips-h} is closer to \textsc{cost} and \textsc{random}, significantly faster than \textsc{b-tips-b} and \textsc{embedding}, and stable with increasing $k$ due to the heap-based pruning strategy.

\noindent\textbf{Exp 2: Comparison of strategies for subtree kernel.} We compare the efficiency of two subtree kernel computation methods during initialization and execution. The results are reported in Figures~\ref{fig:diff_tree_kernel}. Observe that our hash table-based strategy is consistently more efficient.

\noindent\textbf{Exp 3: Impact of relevance.} In this set of experiments, we evaluate the relevance measure (Defn.~\ref{def:relevance}). We set $\lambda$ to 0, 0.5, and 1, and use \textsc{b-tips} to select 5 \textsc{aqp}s; the results appear in Fig.~\ref{fig:relevance}. When $\lambda=0$ (\resp $\lambda=1$), $U$ depends only on relevance (\resp distance). Thus, in Fig.~\ref{fig:relevance}\subref{subfig:relevance_0}, the selected plans differ greatly from \textsc{qep} but have low relevance, while in Fig.~\ref{fig:relevance}\subref{subfig:relevance_1}, the plans have high relevance but are closer to \textsc{qep}. By the learner-centric characterization of \textsc{aqp}s in Section~\ref{sec:feedback}, both situations are undesirable. When $\lambda=0.5$, $U$ combines distance and relevance. As shown in Fig.~\ref{fig:relevance}\subref{subfig:relevance_5}, the two values are balanced: neither is particularly small, the selected plans are dissimilar yet highly relevant, and the result is more reasonable and informative.
We conduct a leave-one-out ablation over the eight labeled plan types, comparing the fixed cubic relevance form to standard regressors trained on the same polynomial features. The cubic form achieves the best agreement with learner feedback (MAE 0.225, RMSE 0.237, Spearman $\rho=0.873$), whereas linear regression and SVR baselines yield higher errors (MAE 0.396--1.048) and lower or negative rank agreement ($\rho\in[-0.86,-0.19]$).

\noindent\textbf{Exp 4: Impact of \textsc{GFP} strategy.} We report the performance of the \textsc{gfp}-based pruning strategy, which underpins the system’s practical latency and whose impact is evaluated here. Fig.~\ref{fig:gtFilter}(a)-(c) plot the results on \textit{tpch\_22.sql}, which has a large plan space. We compare the pruning power, plan informativeness, and efficiency of two variants of \textsc{gfp}, one pruned based only on $s\_dist(\cdot)$ and the other with both $s\_dist(\cdot)$ and $cost\_dist(\cdot)$ (referred to as \textsc{gfp} and \textit{Cost}, respectively). We vary the distance threshold $\tau_d$ and set the cost threshold $\tau_c=\tau_d$. Since the distribution of the plans is not uniform, the curves do not change uniformly. Pruning based only on $s\_dist(\cdot)$ removes more plans but degrades plan informativeness. 
Compared to no \textsc{gfp} (plan informativeness 0.395, time 5s), using \textsc{gfp} with \textit{Cost} (or \textsc{gfp} alone) saves over 40\% (or 95\%) of runtime while sacrificing less than 6\% (or 14\%) of plan informativeness, justifying the use of both $s\_dist(\cdot)$ and $cost\_dist(\cdot)$ in \textsc{gfp} and \textit{Cost}.

Fig.~\ref{fig:gfp} reports the performance of \textsc{gfp} strategy for different $k$ values. It consistently filters more aggressively but achieves lower informativeness than \textit{\textsc{gfp}+Cost}, and the efficiency-informativeness trade-off remains stable as $k$ increases. \eat{Our system implements both of these methods, allowing users to choose according to their needs.}
Figs.~\ref{fig:apx_gfp_k30} and~\ref{fig:apx_gfp_k50} provide supplementary results for $k=30$ and $k=50$, respectively. Observe that they exhibit trends qualitatively similar to those shown in Fig.~\ref{fig:gfp}.

\noindent\textbf{Exp 5: Impact of the \textsc{IPS} strategy.} We evaluate the \textsc{ips} strategy from Section~\ref{ssec:plansel} using \textit{13c.sql} in~\cite{good_viktor_2015}, which has 11 joins. We create two datasets: (a) a baseline set of $n$ plans $(n=1000, 2000, \ldots, 5000)$  randomly sampled from all available plans; and (b) an \textsc{ips}-enhanced set constructed by using \textsc{ips} to collect all alternative plans with the same \textsf{LogiPln} as \textsc{qep}, add them to $n$ randomly selected alternative plans, and then randomly removing plans to match the baseline’s size. We then run \textsc{b-tips} on both datasets and compare plan informativeness. Fig.~\ref{fig:gtFilter}(d) shows that \textsc{ips} typically samples plans with higher informativeness.

\eat{we only change the join order of any two tables to generate a series of alternative plans, add them to the $n$ randomly selected alternative plans. Notably, for a fair comparison, we randomly delete the same number of plans from (b) to ensure that the two sets of data, \ie (a) and (b), have the same size of plan space. After that, we execute \textsc{b-aqps} on the two sets of data and compare the educational utility scores of the  results. Fig.~\ref{fig:gtFilter}(c,d) reports the results.}

\eat{It can be seen from Fig.~\ref{fig:gtFilter}(c,d) that when the number of alternative plans is small, the effect of \textsc{aos} strategy is obvious. But the effect of this method is gradually no longer obvious as the number of alternative plans increases. This is because \textsc{aos} only generates plans that are within a single commutation of joined tables compared to that of the \textsc{qep}. The number of plans \textsc{aos} can generate is limited. When the number of alternative plans gradually increases, its effect can easily be overshadowed by rest of the plans. In practice, this is fine as according to our user study, the number of alternative plans returned to the learners should be small  (within 5) to avoid fatigue and confusion. Obviously, the \textsc{aos} strategy shows significant benefits when less than 10 alternative plans are selected. }

\begin{figure*}[t]
	\centering
	\subfloat[Effect of pruning]{\label{fig:gfp_k_filter}\includegraphics[width=0.32\linewidth, height=3cm]{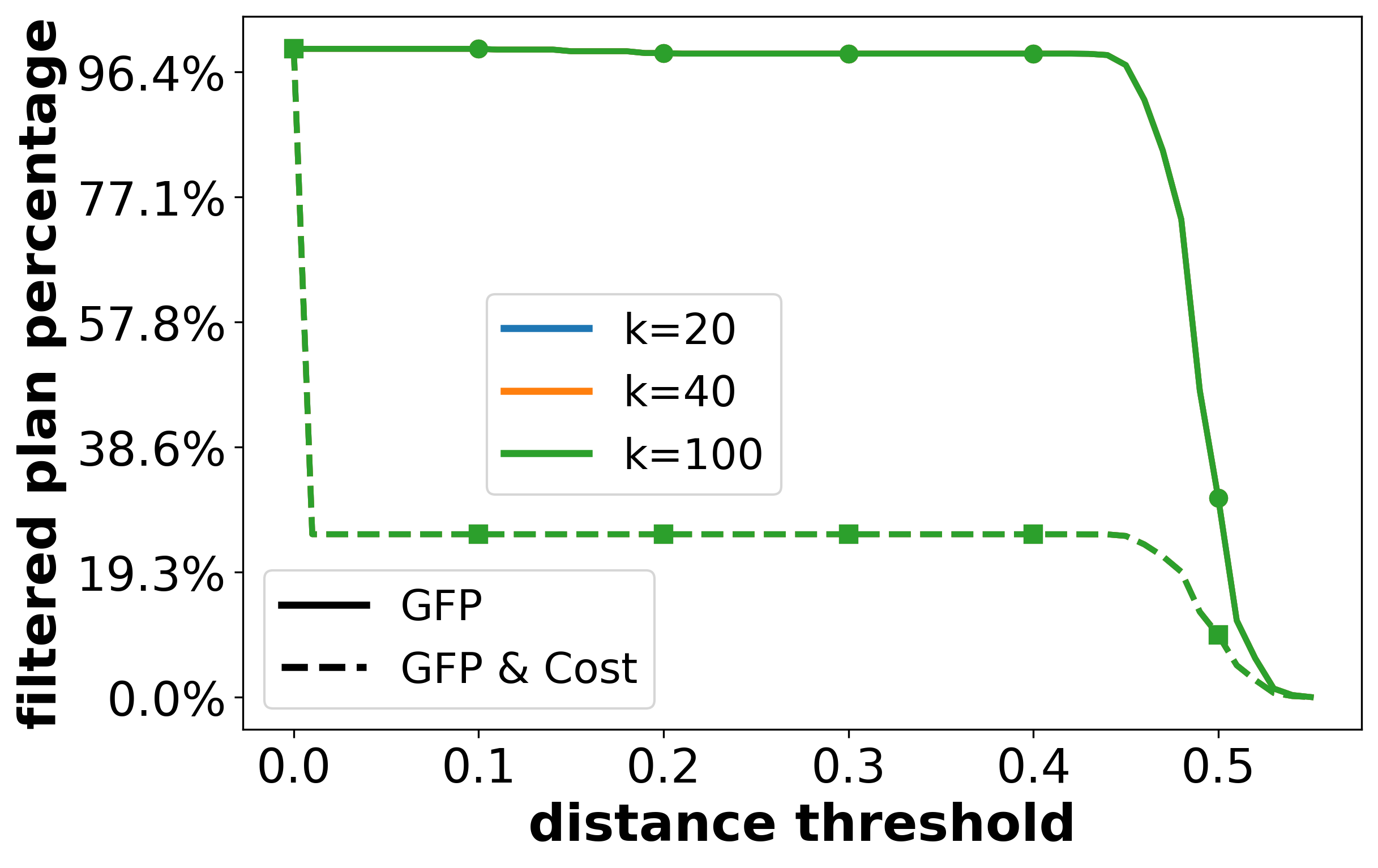}}
	\subfloat[Plan informativeness]{\label{fig:gfp_k_infor}\includegraphics[width=0.32\linewidth, height=3cm]{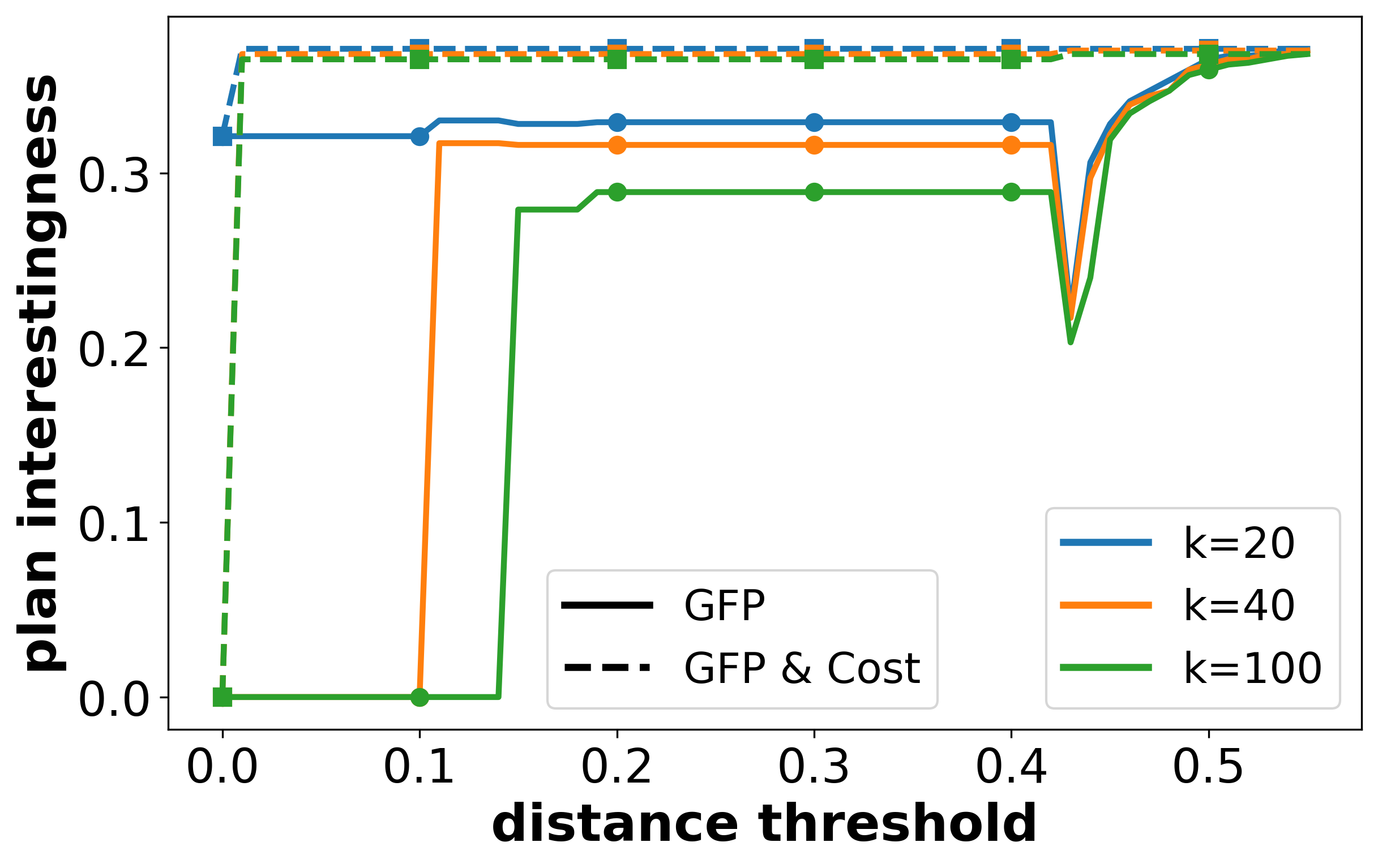}}
	\subfloat[Efficiency]{\label{fig:gfp_k_time}\includegraphics[width=0.32\linewidth, height=3cm]{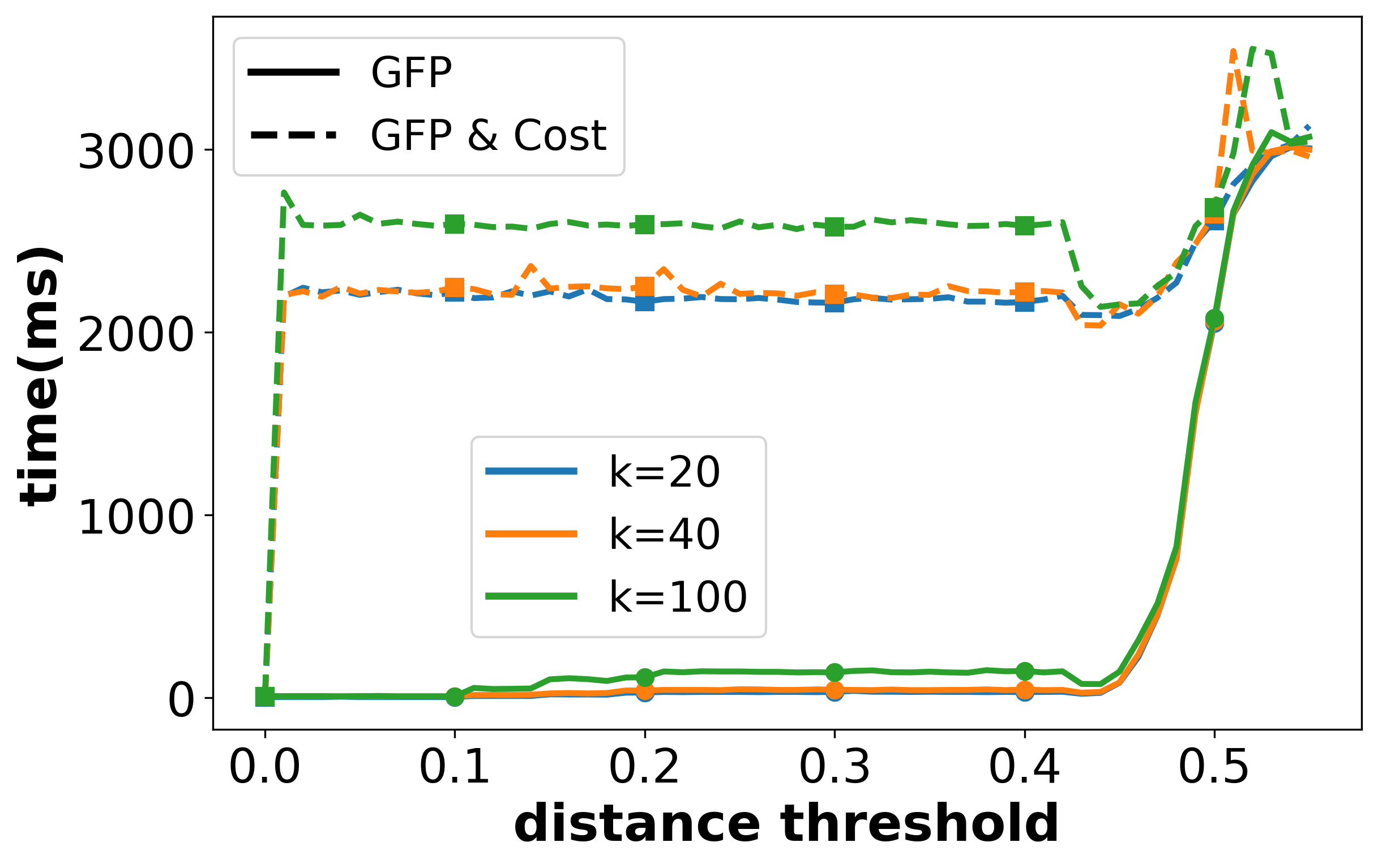}}
  \caption{GFP strategy for different $k$ values (\textit{tpch\_22.sql}): (a) filtered plan percentage, (b) plan informativeness, and (c) runtime.} %
	\label{fig:gfp}
\end{figure*}

\begin{figure*}[t]
	\centering
	\subfloat[Plan informativeness]{\includegraphics[width=0.32\linewidth]{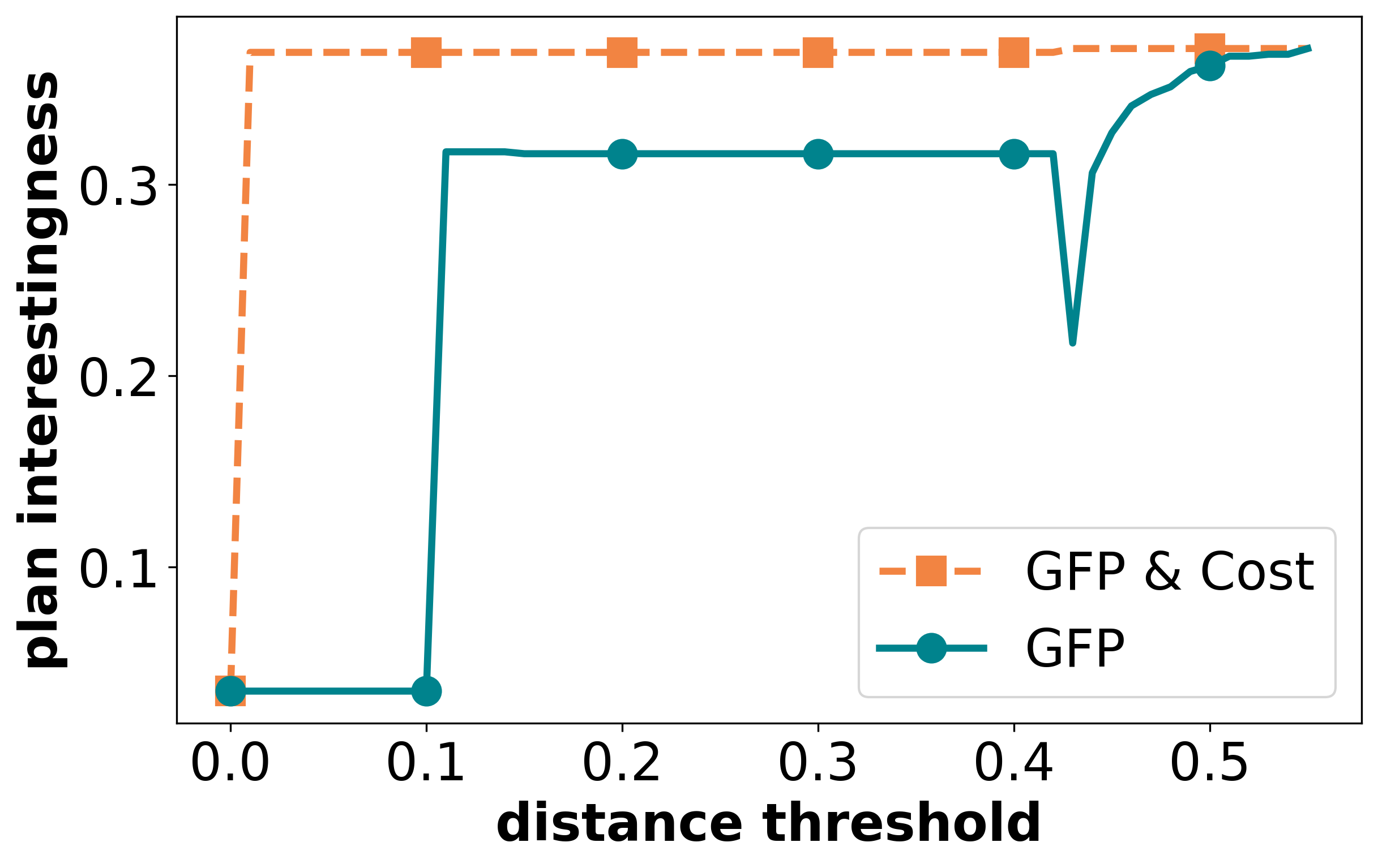}}
	\subfloat[Runtime]{\includegraphics[width=0.32\linewidth]{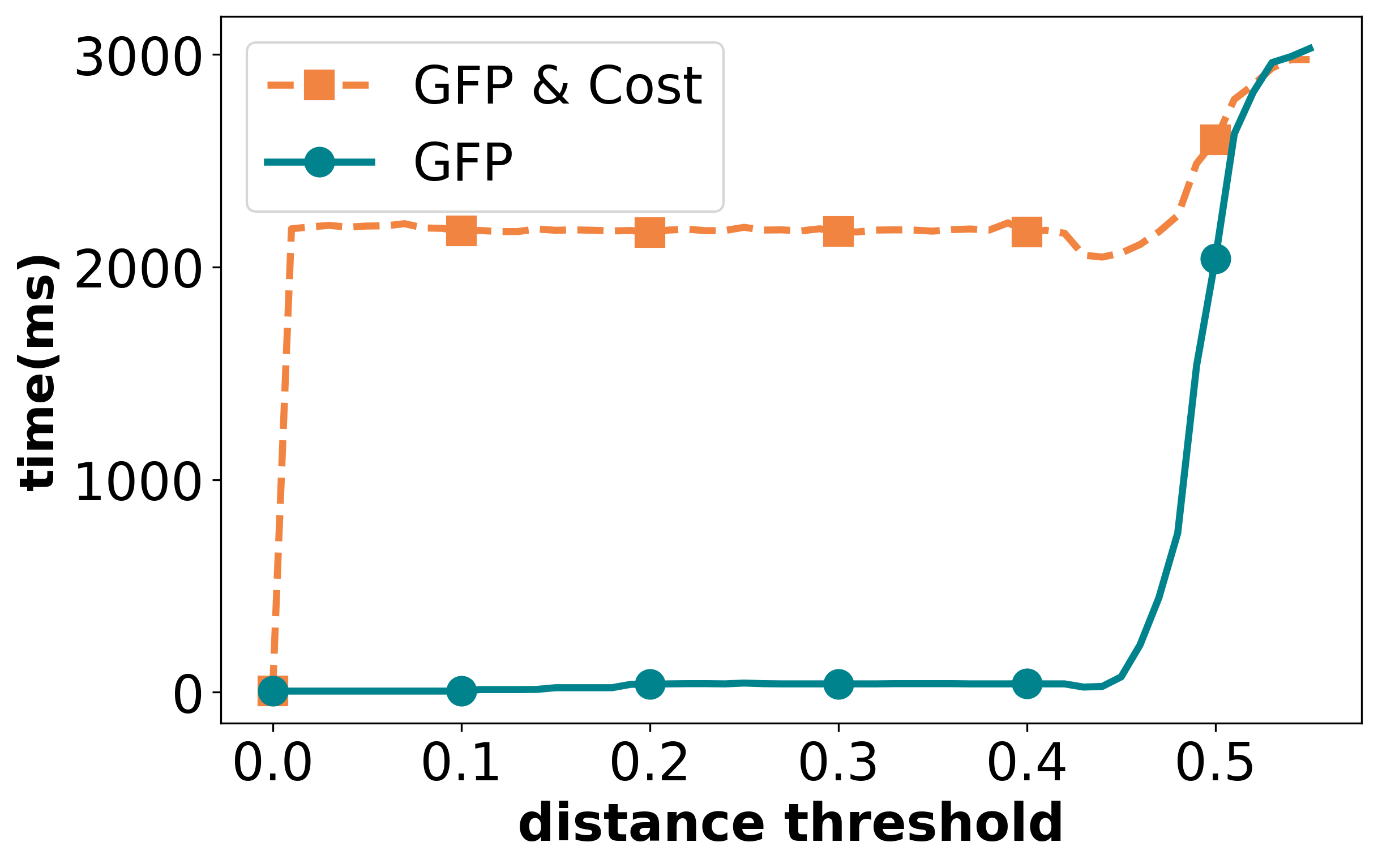}}
	\subfloat[Filtered plan percentage]{\includegraphics[width=0.32\linewidth]{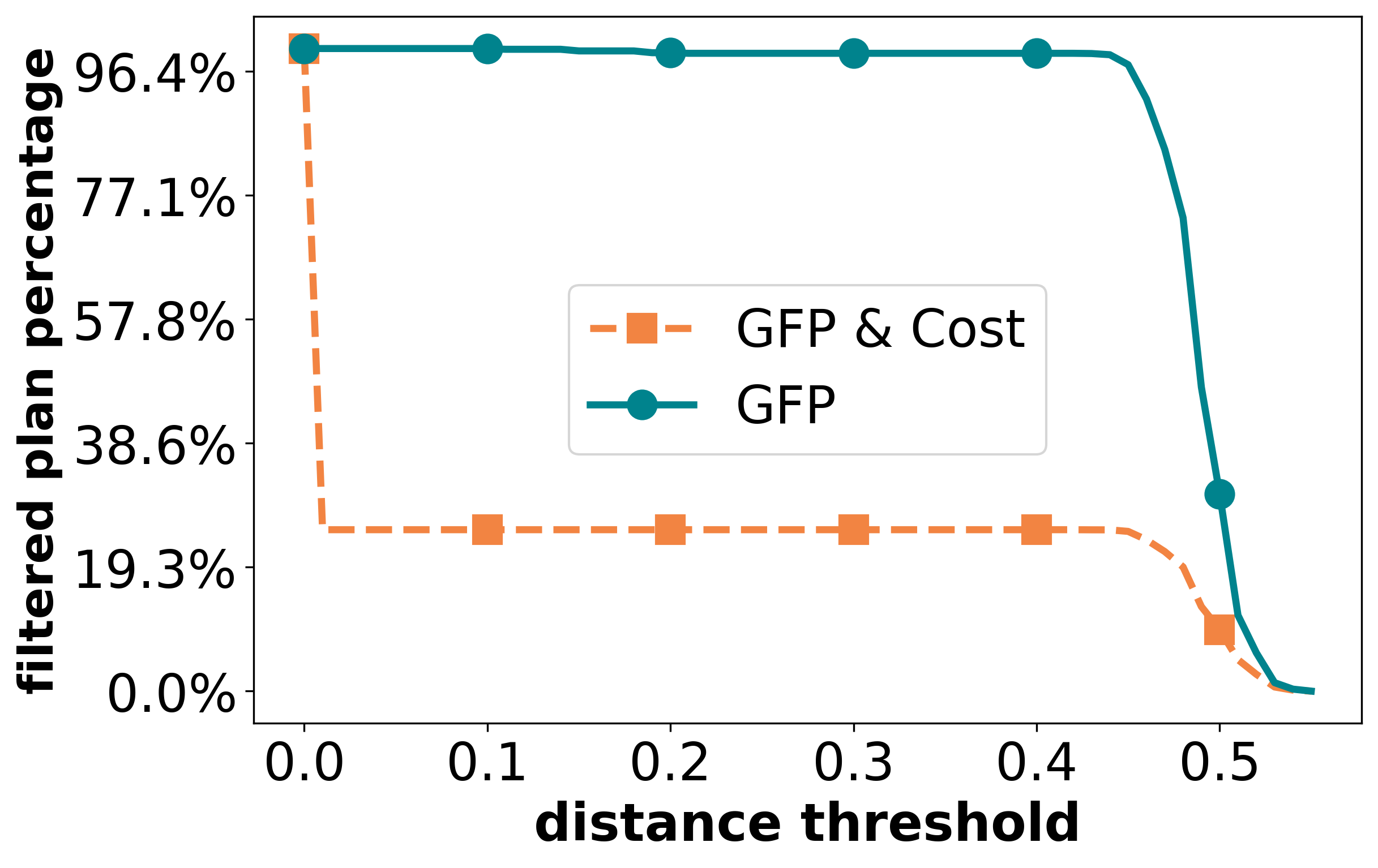}}
	\vspace{-2ex}\caption{\textsc{gfp} strategy for \textit{tpch\_22.sql} with $k=30$.}
	\label{fig:apx_gfp_k30}
	\vspace{-3ex}
\end{figure*}

\begin{figure*}[t]
	\centering
	\subfloat[Plan informativeness]{\includegraphics[width=0.32\linewidth]{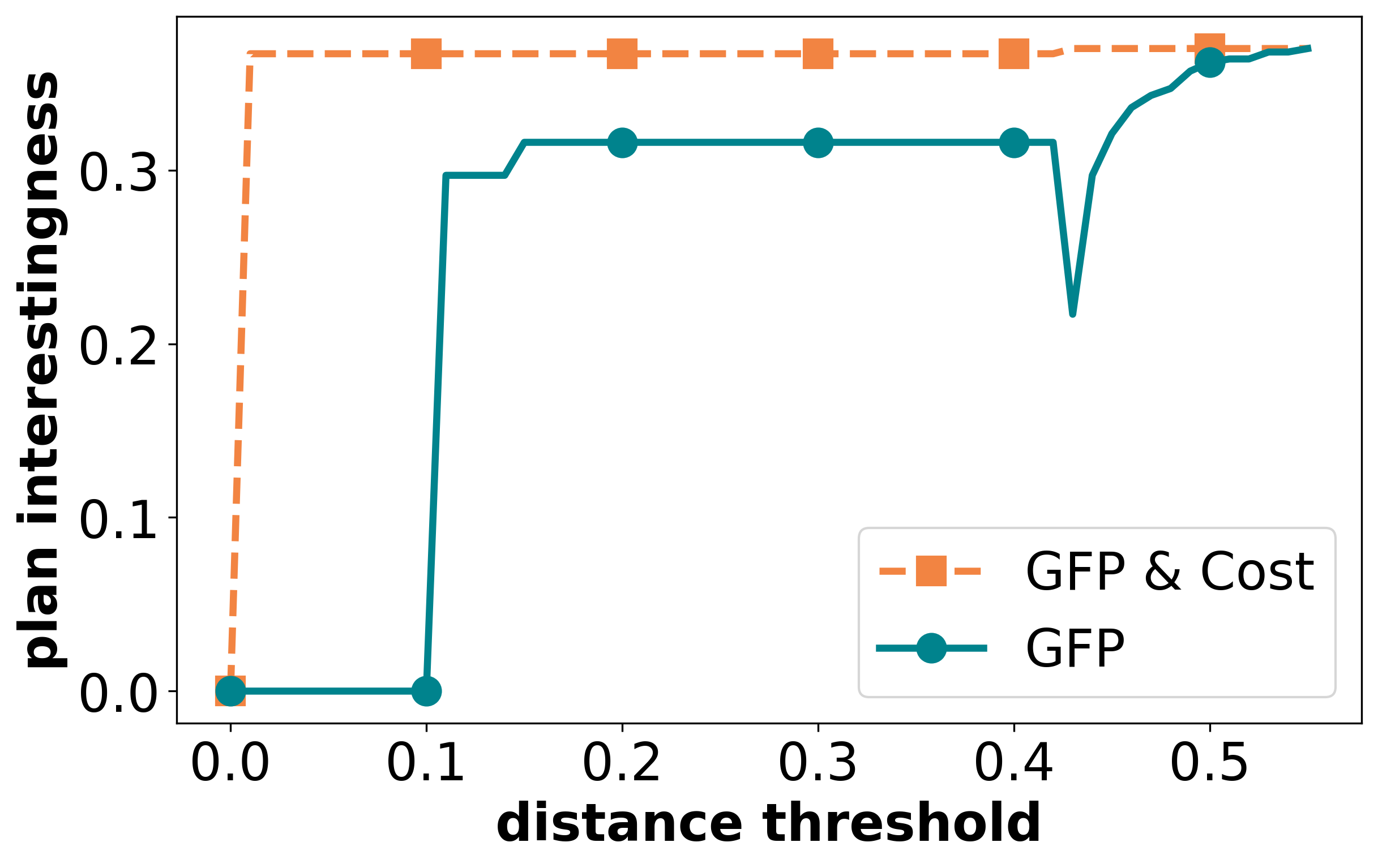}}
	\subfloat[Runtime]{\includegraphics[width=0.32\linewidth]{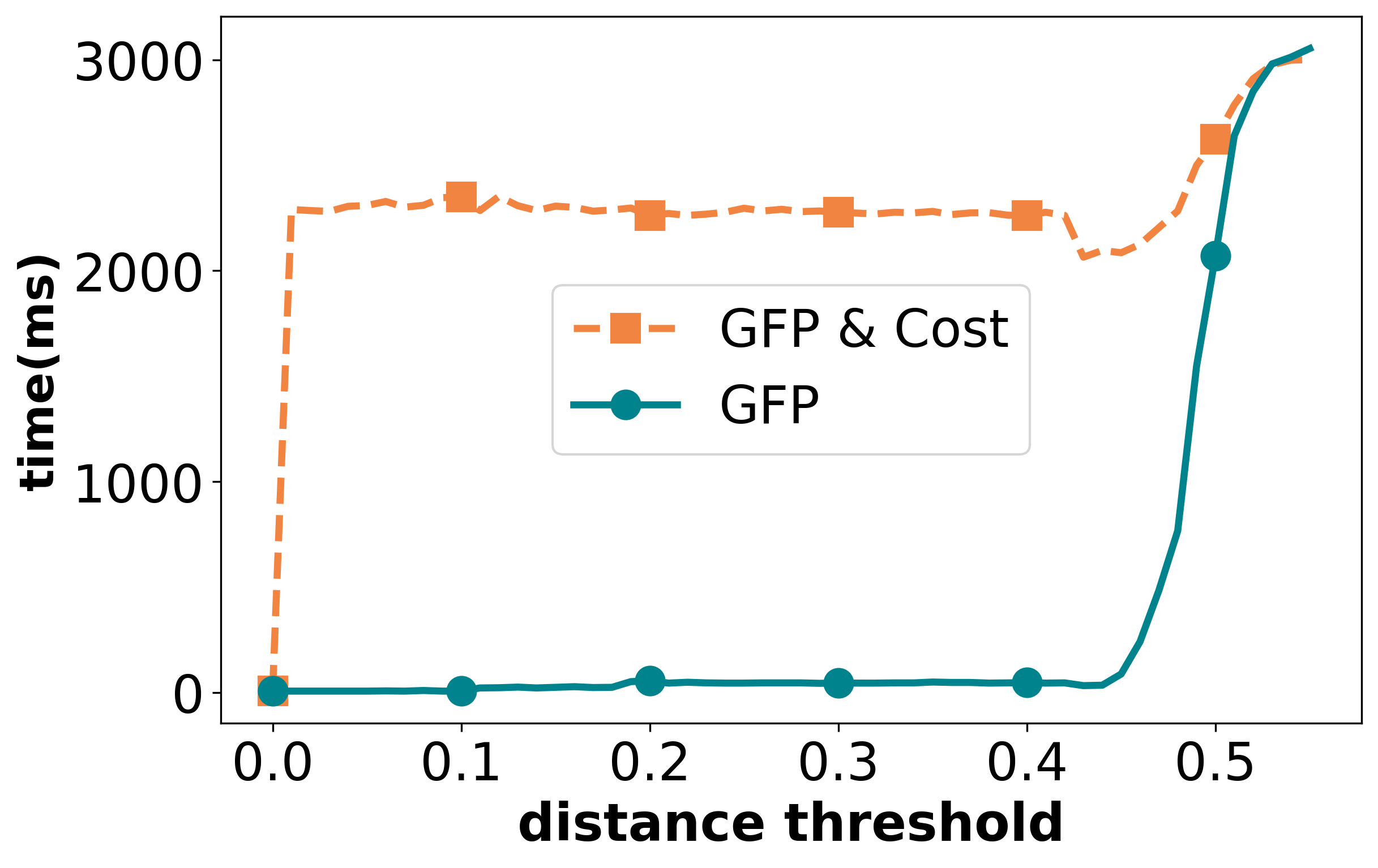}}
	\subfloat[Filtered plan percentage]{\includegraphics[width=0.32\linewidth]{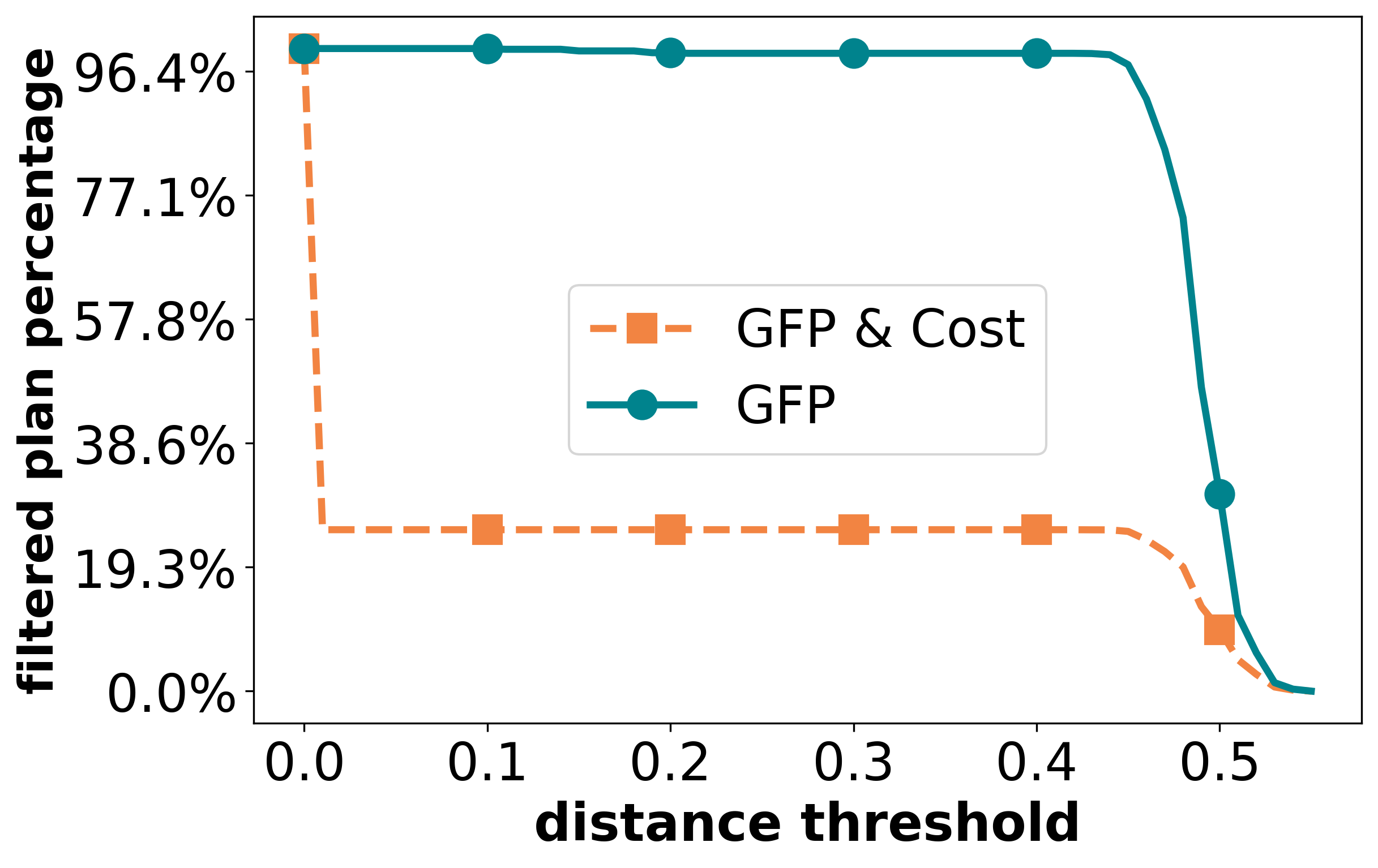}}
	\vspace{-2ex}\caption{\textsc{gfp} strategy for \textit{tpch\_22.sql} with $k=50$.}
	\label{fig:apx_gfp_k50}
	\vspace{-1ex}
\end{figure*}

\begin{figure*}[t]
    \centering
    \subfloat[Q1]{\label{subfig:q1}\includegraphics[width=0.25\linewidth, height=2cm]{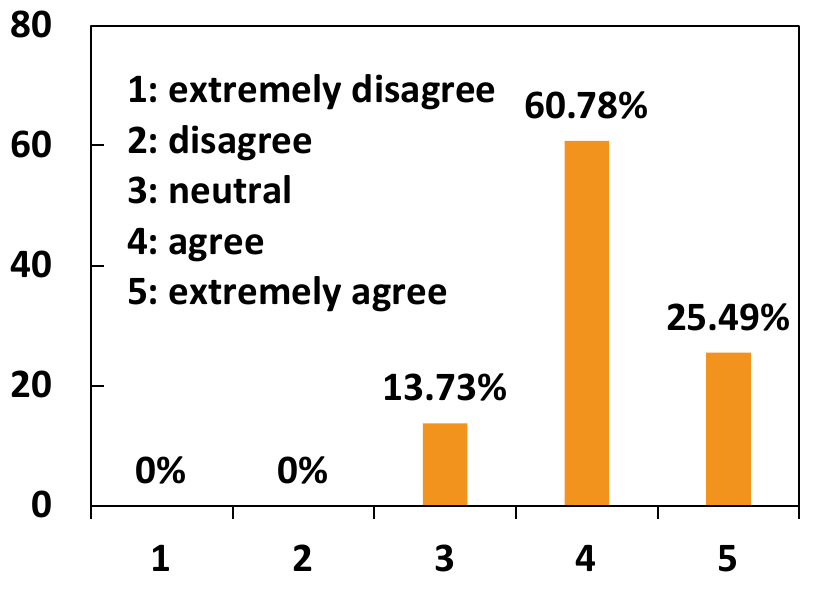}}
    \subfloat[Q2]{\label{subfig:q8}\includegraphics[width=0.25\linewidth, height=2cm]{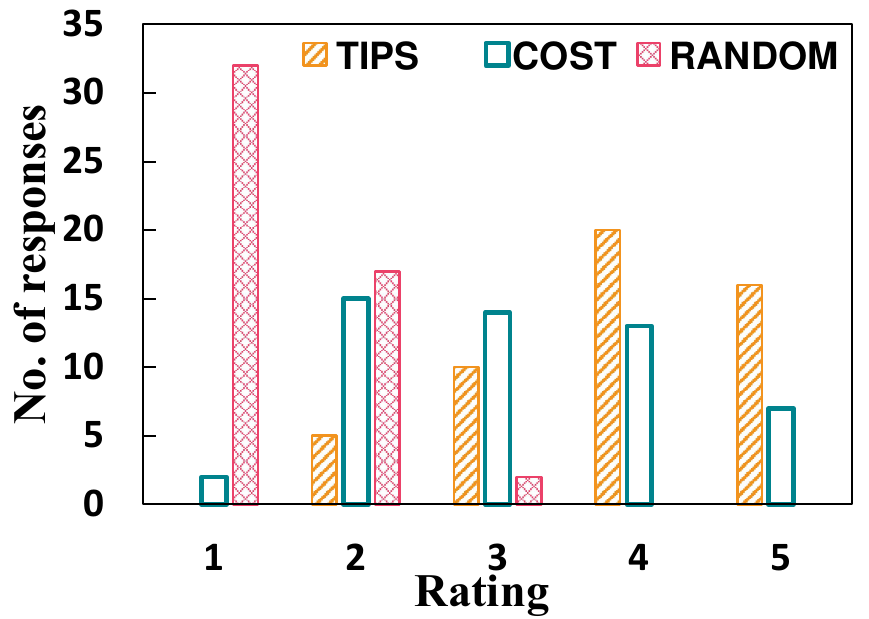}}
    \subfloat[Q4]{\label{subfig:q3}\includegraphics[width=0.25\linewidth, height=2cm]{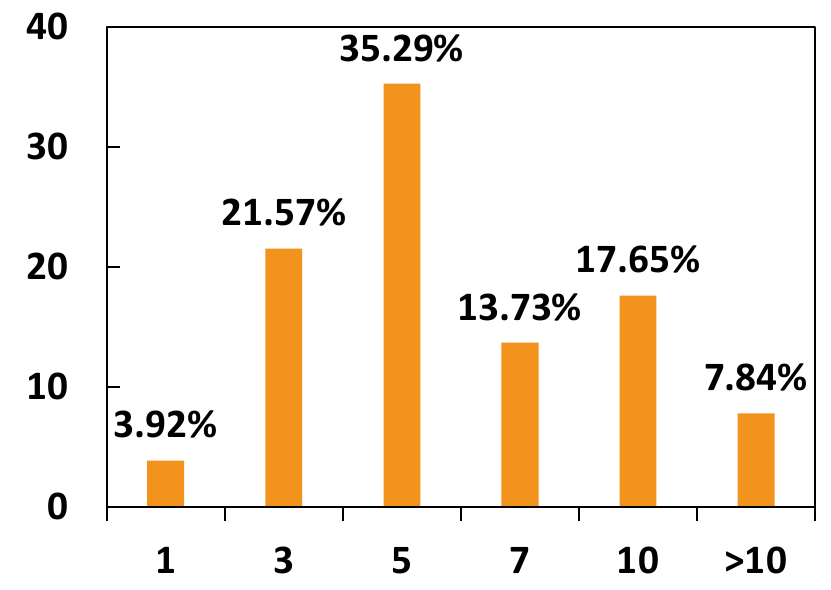}} 
    \subfloat[Q5]{\label{subfig:q7}\includegraphics[width=0.25\linewidth, height=2cm]{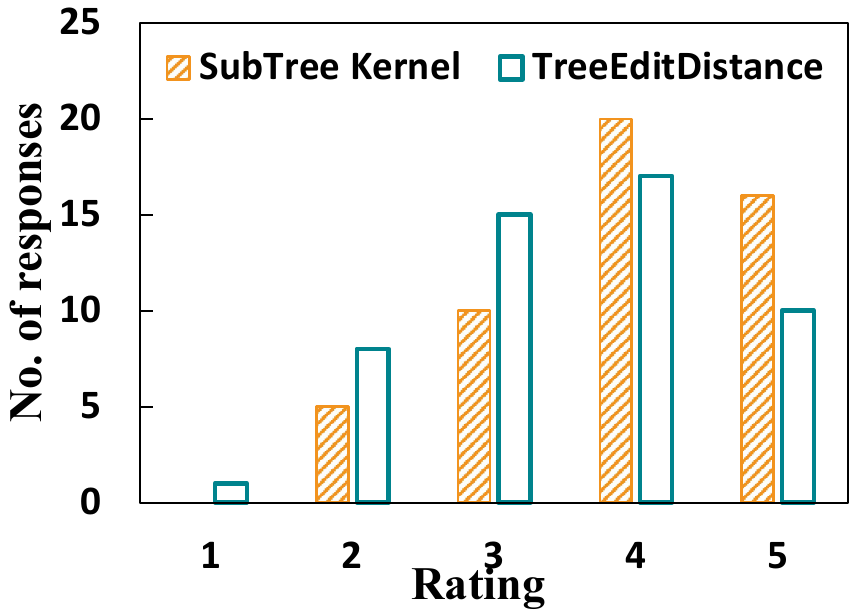}}
    \\
    \subfloat[Q6]{\label{fig:tips_llm_q2}
  \includegraphics[width=0.19\linewidth, height=2cm]{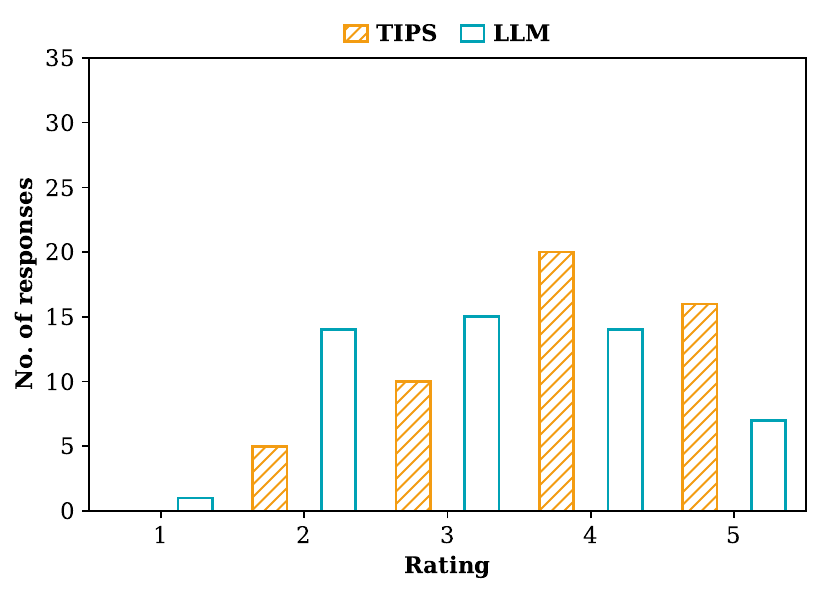}
  } 
    \subfloat[\textsc{b-tips}]{\label{subfig:bapq_diff_dist}
  \includegraphics[width=0.19\linewidth, height=2cm]{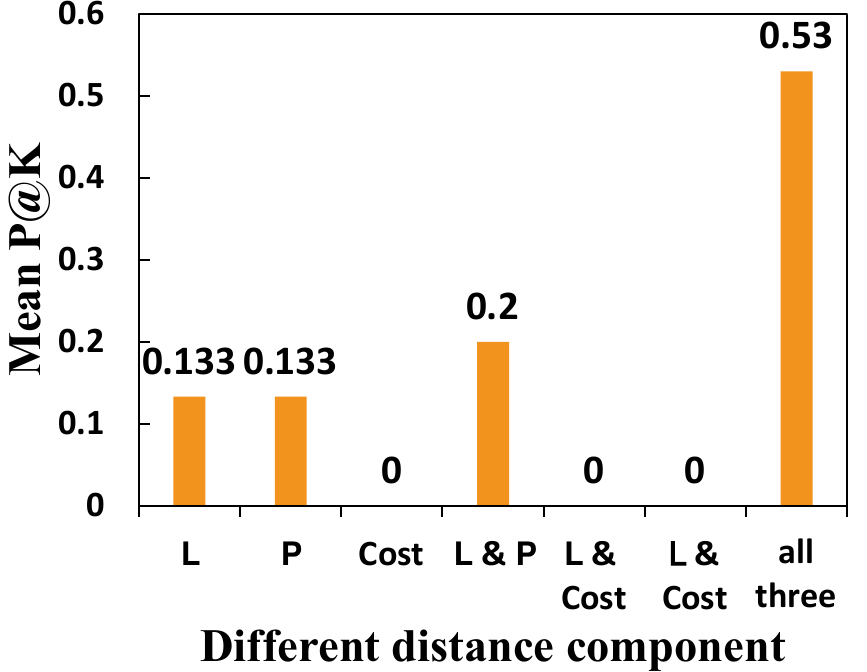}
  }
    \subfloat[\textsc{i-tips}]{\label{subfig:iapq_diff_dist}
  \includegraphics[width=0.19\linewidth, height=2cm]{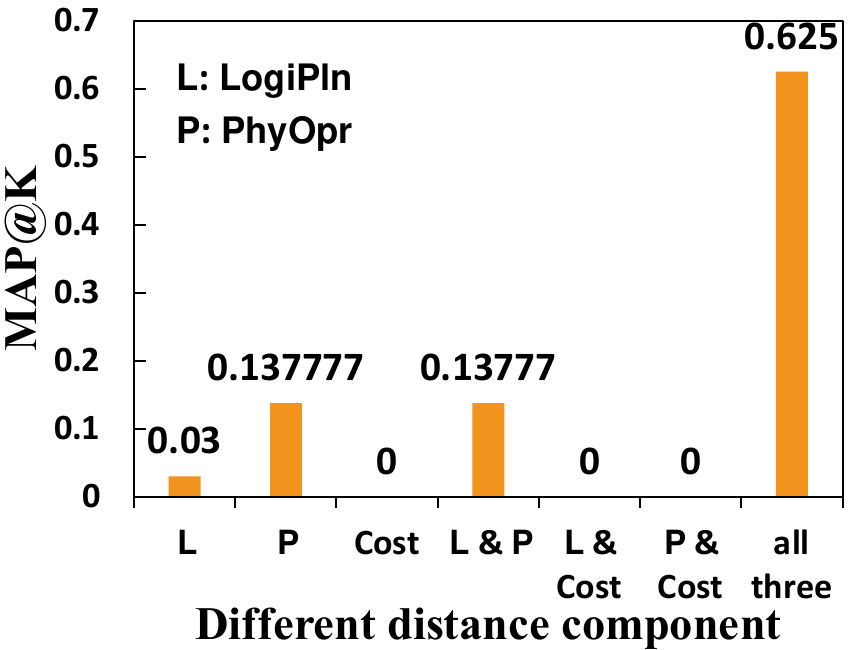}
  } 
    \subfloat[Score distribution]{\label{fig:quizrep1}\includegraphics[width=0.19\linewidth, height=2cm]{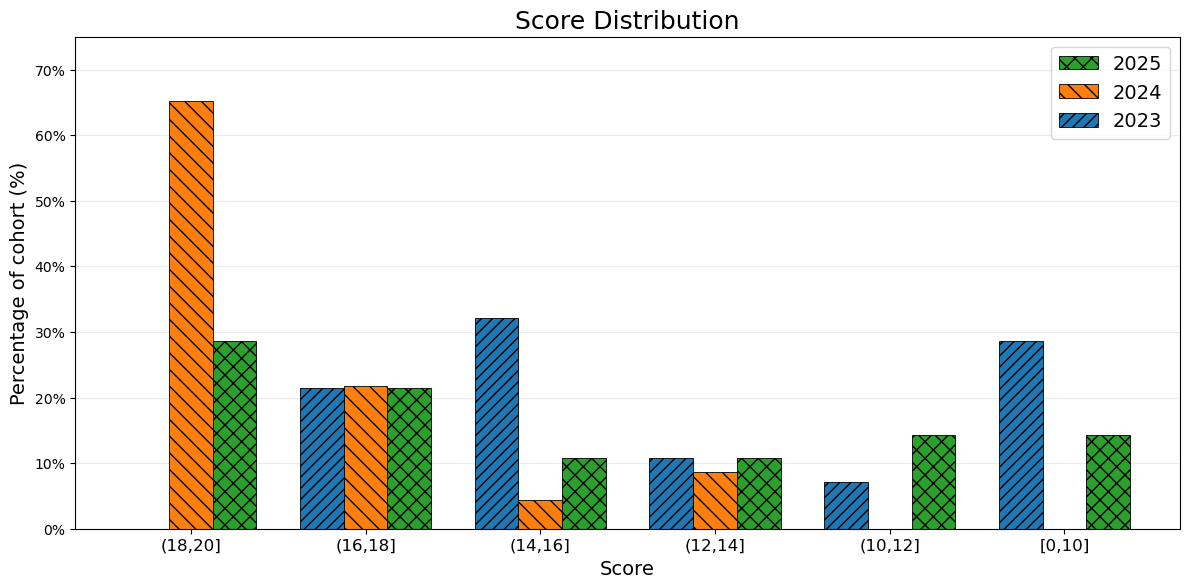}} %
    \subfloat[Score vs. relative \#logs]{\label{fig:quizrep2}\includegraphics[width=0.19\linewidth, height=2cm]{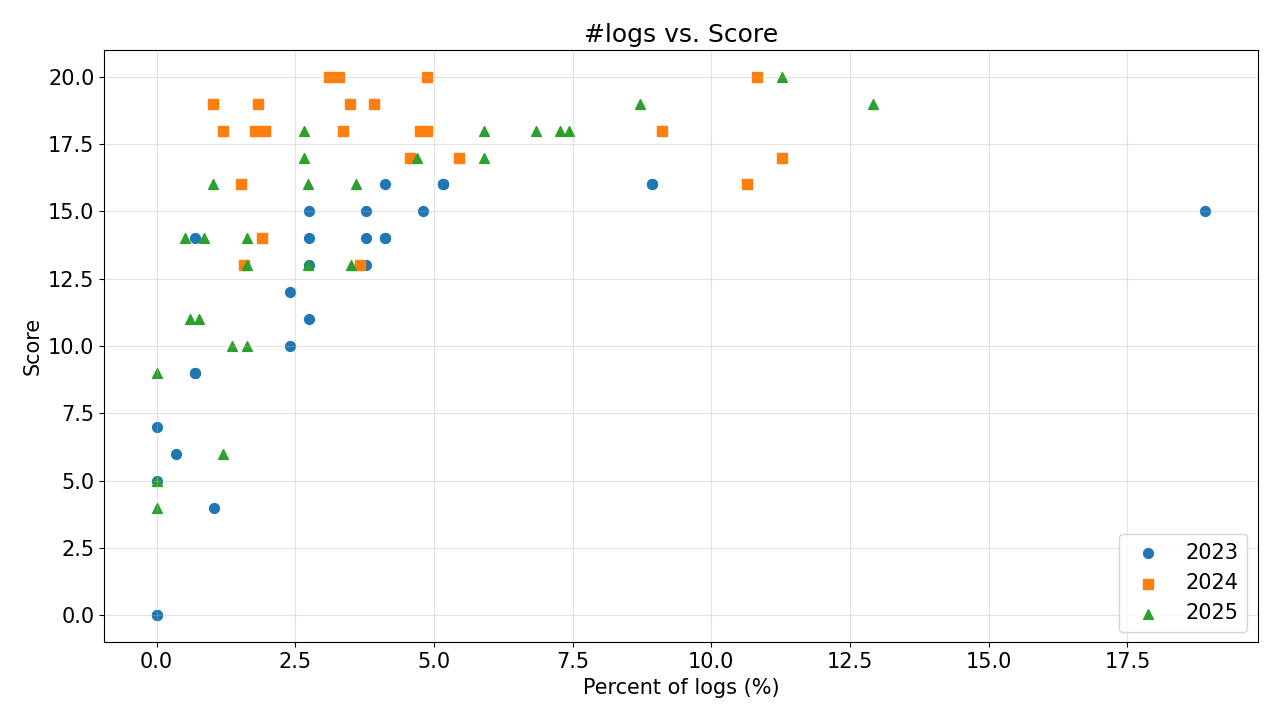}} %
    \caption{User study and academic outcomes.}
    \label{fig:Q1Q3}
\end{figure*}

\section{\textsc{tips} for Database Education}\label{sec:appl}  
We report the usefulness of the \textsc{tips} algorithms (collectively referred to as \textsc{tips}) for supplementing learning of different plan choices (\ie \textsc{aqp}s). To this end, we use the GUI of \textsc{ARENA}~\cite{arena}, which provides a user-friendly interface for retrieving and comparing the \textsc{qep} and \textsc{aqp}s. Specifically, we conducted a user study and an academic assessment to demonstrate the usefulness of the \textsc{tips} algorithms. The various parameters are set to the default (best) values as mentioned in Section~\ref{sec:experiment}.
\subsection{User Study}\label{sec:ustudy}
We conducted a user study among \textsc{cs} undergraduate/postgraduate students who are enrolled in the database course at our university. 51 unpaid volunteers participated in the study. Note that these volunteers are different from those in the survey in Section~\ref{sec:feedback}. After consenting to have their feedback recorded and analyzed, we presented a brief tutorial of the \textsc{ARENA} GUI~\cite{arena}, describing how learners can use the interface to inspect the \textsc{qep} and alternative plans. Then they were allowed to explore it for an unstructured time.\eat{ Through the \textsc{gui}, the participants can specify the choice of problem (\textsc{b-tips} vs \textsc{i-tips}), configure parameters (\eg $\alpha,\beta,\lambda$), and switch between subtree kernel and tree edit distance\eat{ for computing LogiPln differences}.} 
With \textsc{gfp} enabled, the observed waiting time for each \textsc{aqp} in the study did not exceed approximately 1 second.
Figure~\ref{fig:tips_gui} shows the \textsc{tips-gui} interface used in the study.

\begin{figure*}[t]
  \centering
  \includegraphics[width=0.81\linewidth]{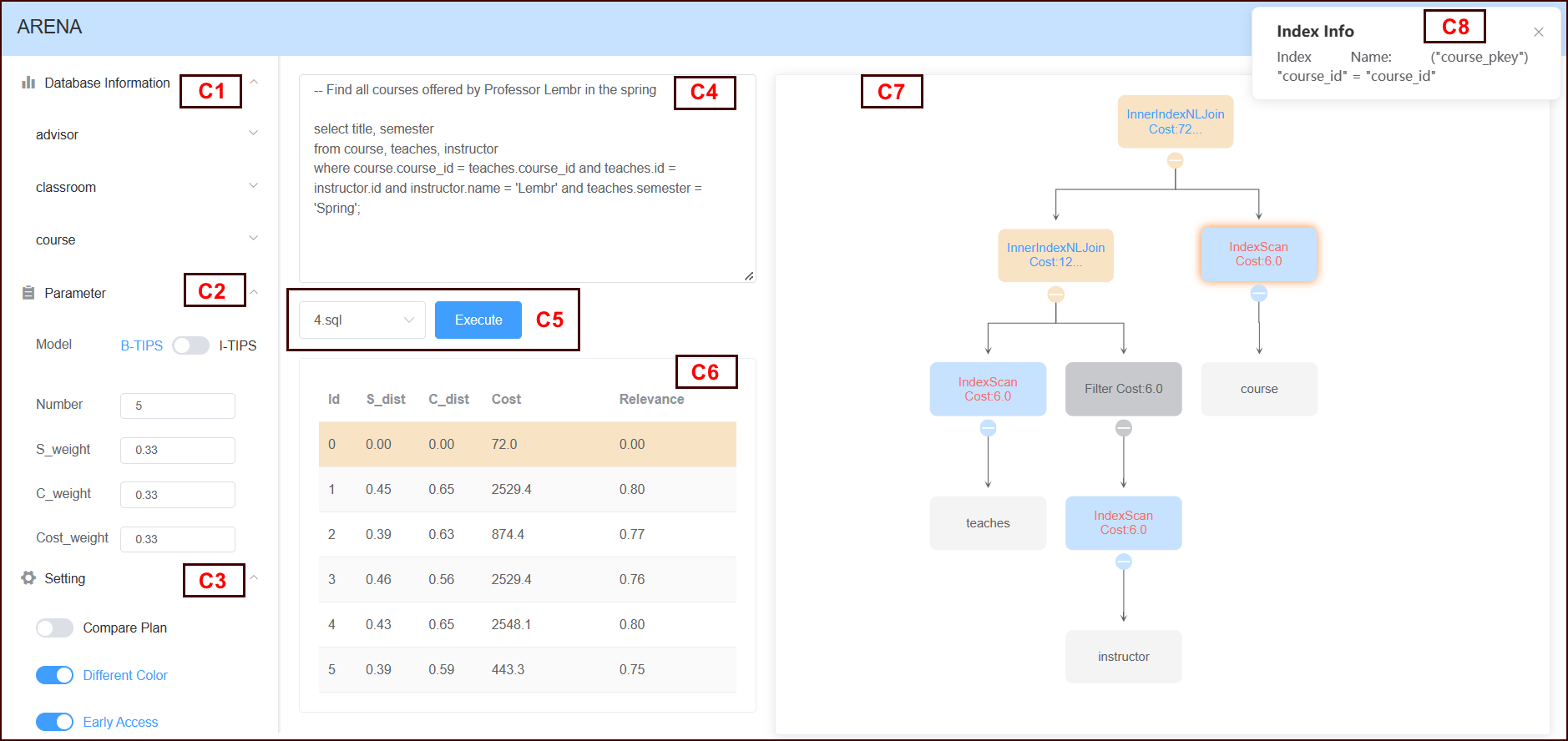}
  \vspace{-3ex}\caption{Screenshot of \textsc{tips-gui} interface used in the user study.}
  \label{fig:tips_gui}
\end{figure*}

\noindent\textbf{US 1: Survey.} We presented 30 predefined IMDb queries to all volunteers, who then viewed the returned \textsc{aqp}s for feedback. For our study, we generated \textsc{aqp}s using \textsc{tips}, \textsc{cost}, and \textsc{random}; all plans can be visualized in \textsc{tips-gui}. Participants could click any \textsc{aqp} to see it plotted side-by-side with the \textsc{qep}, showing the tree structure, each node’s physical operator, and estimated cost. Volunteers could take as much time as they wished to explore the alternative plans. They then completed a survey with several questions.\eat{ In \textit{Q1} and \textit{Q2}, participants indicated whether they found the method useful for improving their understanding in the database systems course.} Where applicable, each subject gave a rating on a 1–5 Likert scale, with higher values indicating stronger agreement with the affirmative statement. We emphasized that survey responses would not affect their grades. Below, we report the key results.

\textit{Q1: Is \textsc{tips} useful in supporting their learning?} Fig.~\ref{subfig:q1} plots the results of the corresponding responses.  Observe that 44 participants agree that it is useful for their learning.\eat{ None thinks that there is no positive impact.}

\textit{Q2: How well does \textsc{tips}/\textsc{random}/\textsc{cost} help in understanding the alternative plan space?} To mitigate bias, the results of \textsc{tips}, \textsc{cost}, and \textsc{random} are presented in random order, and the participants were not informed of the approach used to select the \textsc{aqp}s. Fig.~\ref{subfig:q8} reports the results. In general, \textsc{tips} receives the best scores (\ie 36 out of 51 gave a score over 3).

\textit{Q3: Preference for \textsc{i-tips} or \textsc{b-tips}.} Most participants preferred \textsc{b-tips} (38 vs. 13), indicating that both \textsc{i-tips} and \textsc{b-tips} are valuable since learners differ in how they prefer to explore \textsc{aqp}s.

\textit{Q4: How many \textsc{aqp}s is sufficient to understand different plan choices?} Fig.~\ref{subfig:q3} reports the responses. Most participants felt that 10 or fewer \textsc{aqp}s were sufficient, noting that additional \textsc{aqp}s provided diminishing gains in understanding, which supports the diminishing-return nature of the approximate solution.

\eat{Next, we report the impact of different parameters on the understanding of \textsc{qep} selection for a given \textsc{sql} query.} 

\eat{\textit{Q5: Importance of LogiPln, PhyOpr, and Cost differences.}  Fig.~\ref{subfig:q5} reports the responses. Setting $\alpha$ (\resp $\beta$) as 0 indicates that the LogiPln (\resp PhyOpr) difference is not taken into account in \textsc{tips}. Similarly, by setting $\alpha=\beta=0.5$, \textsc{tips} ignores the cost difference between plans. We observe that $\alpha=\beta=1/3$ receives the best response. 36 participants give a score of 4 and above, while each of the remaining settings have less than 10 responses with scores 4-5. This highlights the importance of LogiPln, PhyOpr, and Cost differences.

\textit{Q6: Importance of distance and relevance measures.} $\lambda$ allows us to balance the impact of distance and relevance in selecting alternative plans. Fig.~\ref{subfig:q6} reports the results. Clearly, by setting $\lambda=0.5$, the results of \textsc{tips} is significantly better than either disabling distance or relevance.  }

\textit{Q5: Tree Edit Distance vs. Subtree Kernel.} Fig.~\ref{subfig:q7} reports the results. In general,  subtree kernel-based structural difference computation strategy receives superior scores (\ie 36 out of 51 give a score over 3) to tree edit distance-based strategies (\ie 27 out 51). This justifies the choice of the former as the default LogiPln difference computation technique.

\textit{Q6: \textsc{tips} vs.\ LLM-based Learning.}
Given the popularity of Large Language Models (\textsc{llm}) in education, we compare \textsc{tips} with \textsc{llm}-based learning (ChatGPT, GPT-4o). Unlike \textsc{llm} explanations that operate over text descriptions, \textsc{tips} exposes concrete optimizer plans and costs, enabling direct comparison of \textsf{LogiPln}/\textsf{PhyOpr}/\textsf{Cost} trade-offs.
We follow the \textit{Q2} protocol. After each query, participants rate how well the provided information helps them understand the alternative plan space. Fig.~\ref{fig:tips_llm_q2} summarizes the ratings. 
\textsc{tips} outperformed \textsc{llm}: 36 of 51 \textsc{tips} ratings exceeded 3 compared to 20 of 51 for \textsc{llm}, with a higher average rating (3.92 vs. 3.16) and more ratings concentrated at 4-5 rather than 1-2.


\noindent\textbf{US 2: Impact of different distance components.} We explore the impact of different distance components on \textsc{aqp} preferences. We consider various combinations of distance components ($x$ axis in  Fig.~\ref{fig:Q1Q3}(f) and~\ref{fig:Q1Q3}(g)). For each combination, the distance components are equal in weight. These combinations were hidden from the students to mitigate any bias. We choose 3 \textsc{sql} queries and ask participants to mark and order the top-5 \textsc{aqp}s (from 50 plans returned by our algorithms in each setting) that they most want to see for each query. Note that we fix $k=50$ here for readability; the system supports larger $k$ and we observe similar trends. We then compare the user-marked list of plans in each case with the results returned by the \textsc{tips} algorithms.\eat{ Since \textsc{b-tips} and \textsc{i-tips} are two different modes of exploration, we have adopted different evaluation indicators for these two cases.} For \textsc{b-tips}, we determine whether the returned result contains the alternative plans that users expect to see and the number of such plans. We do not consider the order in which these plans appear. Therefore, we use \textit{Mean P@k (k=5)}~\cite{mcf2010metric} to measure similarity. \textsc{i-tips}, on the other hand,  returns \textsc{aqp}s iteratively. The order of appearance is important. Therefore, for \textsc{i-tips}, we use the \textit{MAP@K (k=5)}~\cite{mcf2010metric}.\eat{ \textit{P@k} represents the precision for top-k elements in the answer. \textit{Mean P@k} represents the mean of the \textit{P@k} for all the queries and in our experiment, it represents the mean of 3 \textsc{sql} queries. \textit{AP@k} represents the mean of \textit{P@i for i=1,...,k}. Finally, \textit{MAP@k} represents the mean of \textit{AP@k} for all the queries.}

\eat{The experimental results are reported in Fig.~\ref{fig:Q1Q3}(f) and~\ref{fig:Q1Q3}(g).} Observe that regardless of the exploration mode, the effect when the three distances (\ie \textsf{LogiPln}, \textsf{PhyOpr}, and \textsf{Cost}) are considered simultaneously is superior to other combinations. Interestingly, when the weight of cost distance is large, \textit{Mean P@k} and \textit{MAP@k} are both 0. That is, only considering the cost of plans as a factor for selecting \textsc{aqp}s results in potentially suboptimal collection as far as learners are concerned.  

\noindent\textbf{US 3: Result differences between \textsc{i-tips} and \textsc{b-tips}.} Lastly, we investigate the differences in the results returned by our proposed algorithms. We collected 80 groups of results from 20 volunteers; each volunteer was requested to use both \textsc{i-tips} and \textsc{b-tips} and view 10 \textsc{aqp}s. For each \textsc{aqp} returned in \textsc{i-tips}, volunteers can mark whether it is useful to them (\ie rating $r$) and the \textsc{aqp}s returned subsequently take this feedback into account. \eat{During the study, after 10 \textsc{aqp}s returned, the users are asked to label those ones helpful in her study experience.} We observe that, on average, 56.07\% of the returned \textsc{aqp}s are different in both solutions.
\subsection{Academic Outcomes}\label{sec:outcome}

We conducted a three-year longitudinal study across three cohorts of an undergraduate computer science degree program. The \textit{\textsc{tips}-based pedagogical model} was introduced to each of the cohorts in \textit{February 2023}, \textit{February 2024} and \textit{February 2025} by the end of March, respectively, with courses ending in mid-June and a final paper-based exam held later that month. Each exam included a 20-mark question assessing students’ understanding of alternative plans, including (a) other \textsc{aqp}s that can also answer the query; (b) whether a plan is better than another; and (c) what plans the optimizer will never choose. The question set can be found in the Appendix. We report \emph{relative} (within-cohort normalized) counts: within each cohort, score-bin counts are divided by cohort size so that bins sum to 100\%, and log-score plots normalize each student's logs by the cohort’s total logs. Fig.~\ref{fig:Q1Q3} presents these results.

\eat{We have been conducting a long-term empirical test on the academic outcomes within the last 3 years, over the classes in 3 grades of the same program, major in computer science. Our model was introduced to each of the cohorts in \textit{February 2023}, \textit{February 2024} and \textit{February 2025} by the end of March, respectively. The course culminated in mid-June for all three cohorts. After that, a final paper-based exam is conducted by the end of June. In the exam of each year, there is a question, with 20 marks, asking about the following: (a) other \textsc{aqp}s that can also answer the query; (b) whether a plan is better than another; and (c) what plans the optimizer will never select. The question set can be found in the Appendix. 
We report \emph{relative} (within-cohort normalized) counts: within each cohort, score-bin counts are divided by cohort size so that bins sum to 100\%, and log-score plots normalize each student's logs by the cohort’s total logs. Fig.~\ref{fig:Q1Q3} presents these results.}

Each cohort consisted of about 30 undergraduate students. Over the three-year study, we varied how the \textsc{tips}-based pedagogical model was integrated into teaching: in 2023, \textsc{tips} was introduced without required practice; in 2024, students were encouraged to practice extensively on the \textsc{tips-gui}, with usage contributing to their lab homework grade (the final grade consists of final exam score, code-based project score, and lab homework score); and in 2025, practice was optional and ungraded to assess voluntary adoption. To ensure unbiased evaluation, different teaching assistants--unaware of the \textsc{tips}--graded the final exams for each cohort, enabling us to address the following research questions:

\eat{During the 3-year empirical study, we have been trying different approaches to inject our model into the pedagogical process. In 2023, we only introduced the model but did not ask for complementary practice over the given \textsc{tips-gui}; in 2024, we asked the students to practice on the \textsc{tips-gui} as much as possible and informed them that the practical logs would account for their lab homework score (the final grade consists of final exam score, code-based project score, and lab homework score); in 2025, different from 2024, we adopted a more flexible policy, \ie practicing on \textsc{tips-gui} is not complementary and did not directly contribute to their scores, in order to test whether the students find \textsc{tips-gui} helpful and spontaneous practice on that. We employed three different teaching assistants, one for each cohort, where none of the assistants is aware of \textsc{tips-gui}, to mark the final exam sheets, in order to answer the following research questions:}

\textbf{RQ1: Do \textsc{tips} help to improve academic outcomes?} Students in 2024 performed best, with over half scoring above 18, compared to fewer than one third in 2025 and none in 2023, indicating superior overall performance for the 2024 cohort, followed by those in 2025.
\eat{More than half of the students in the 2024 cohort scored higher than 18; less than one third of the students in the 2025 cohort got a score higher than 18; while no one in the 2023 cohort scored above 18. Comparing all three cohorts, it is evident that the students in 2024 obtained the best scores, followed by those in 2025.} We conducted a one-way ANOVA to determine whether the differences between the score distributions of the three cohorts are significant. The test resulted in a p-value of 0.000008, which indicates that the different pedagogical model of \textsc{tips} significantly affects the scores obtained by students across the three cohorts.

\textbf{RQ2: Is there any correlation between the user logs and their outcomes?} As illustrated in \Cref{fig:quizrep2}, quiz scores generally increase with the percentage of cohort logs, exhibiting a pattern of diminishing returns. The improvement is particularly pronounced for the 2023 cohort, where scores rise rapidly with increasing log percentages, likely due to their relatively limited number of logs compared to the other cohorts. However, once the proportion of logs reaches approximately 5--10\%, further increases yield little additional gain in scores. Because the 2024 cohort were encouraged to use \textsc{tips}, their average number of logs is correspondingly higher than those of the 2023 and 2025 cohorts. Given the positive relationship between the number of logs and performance, this also accounts for why their score distribution in \Cref{fig:quizrep1} is the strongest among the three cohorts.

\eat{As shown in \Cref{fig:quizrep2}, in general, as the percentage of cohort logs increases, the scores also increase with a diminishing return pattern. In particular, the scores improve rapidly \wrt as the percentage of logs increases for students of 2023, who have a generally limited number of logs compared with the other two cohorts. Besides, after the percentage of logs reaches 5--10\%, further increases in \#logs will not lead to significant improvements in the scores obtained. As practicing on \textsc{tips-gui} is complementary for the 2024 cohort, the average \#logs is correspondingly larger than that of the 2023 and 2025 cohorts. As \#logs shows a positive effect on the scores, this also explains why the score distribution in \Cref{fig:quizrep1} is the best among the three cohorts.}

\noindent\textbf{\textit{Remark}}. We analyzed the login and query logs of \textsc{tips} and found that database course students have used it frequently over the past two years. Their queries often involve complex nested subqueries and multi-table joins (more than three tables), indicating that students actively rely on our framework to better understand plan choices for complex queries.

\subsection{Beyond Database Education}
Finally, we briefly describe how our framework can be used beyond database education. Specifically, we highlight two practical use cases where it benefits DBAs and junior database engineers.

\textbf{Slow-query diagnosis.} In traditional tuning, DBAs often spend considerable effort examining minor variants of the same inefficient plan. By suppressing \textsc{aqp}s that are structurally or operator-wise similar to the \textsc{qep}, \textsc{tips} enables a DBA to explore queries that are slow by identifying \textsc{aqp}s from a different ``quadrant'' of the search space (\eg plans with a substantially different join order but comparable estimated cost). If such a structurally distant \textsc{aqp} performs better in practice, it indicates that the optimizer's cost model may be biased toward a particular plan shape (\eg left-deep trees) that is suboptimal for the current data distribution. Moreover, by exposing plans that differ significantly in their physical operators, \textsc{tips} can also help DBAs determine whether a performance bottleneck stems from a specific implementation (\eg a \textsf{Nested Loop Join} causing I/O thrashing) rather than the query logic itself. Because our framework presents the DBA with a ``candidate list'' of distinct physical strategies to test, it can potentially accelerate the process of identifying an effective solution.

\textbf{Onboarding junior engineers.} \textsc{tips} can serve as a valuable tool for onboarding junior database engineers. By filtering out \textsc{aqp}s that closely resemble the \textsc{qep}, it immediately exposes juniors to high-contrast examples, allowing them to examine why the optimizer favors a \textsf{Hash Join}-based pipeline over a structurally different \textsf{Sort-Merge}-based alternative. Additionally, by reviewing $k$ plans that vary across both \textsf{LogiPln} and \textsf{PhyOpr}, junior DBAs can develop an understanding of how different workload patterns (\eg OLTP vs. OLAP) translate into distinct plan shapes.

\section{Related Work}\label{sec:rel}  
Our work is closely related to the problem of diversified top-\(k\) query processing, which seeks to compute the \(k\) most relevant results for a user query while explicitly incorporating diversity into the ranking objective. This problem has been extensively investigated across a broad range of settings, including diversified keyword search in databases~\cite{broad_zhao_2011}, textual document collections~\cite{efficient_2011_albert}, and graphs~\cite{keyword_2008_konstantin}; diversified top-\(k\) pattern matching~\cite{diversified_wenfei_2013}; diversified clique enumeration~\cite{diversified_long_2015}; and diversified structural search in graphs~\cite{top_xin_2013}, among others~\cite{sap_rui_2017,top_bin_2018}.  
Furthermore, \cite{search_drosou_2010} provided a comprehensive survey of methodologies for diversifying query results, Ting and Jin~\cite{on_ting_2013} study the computational complexity of diversification problems, and Gollapudi and Sharma~\cite{axiomatic_go_2009} employ an axiomatic framework to analyze diversification systems. Nevertheless, these techniques are not directly applicable to our setting, as they do not address the informativeness of query execution plans within a relational database management system (\textsc{rdbms}).

Databases and Structured Query Language (\textsc{sql}) are core topics in Software Engineering and Computer Science curricula~\cite{taipalusSQLEducationSystematic2020}. Because students struggle with \textsc{sql}~\cite{MAF21}, many tools have been developed to help them understand complex statements~\cite{KV+12,HM+22,MRY19,MF21,LZ+20,DG11} and learn to write correct, efficient queries~\cite{Hu2024QrHintAH,wangFalsePositivesDeceptive2024,miedemaExpertPerspectivesStudent2022,taipalusErrorsComplicationsSQL2018,brusilovskyLearningSQLProgramming2010}. Other work focuses on visualizing \textsc{sql} queries~\cite{Gatterbauer2024ACT, Gat22,LZ+20,MF21} and query plans~\cite{SDB15}. Far less research targets technologies for learning relational query processing~\cite{Tian2024SQLucidGN, towards_weiguo_2021,neuron_liu_2018,picasso_jayant_2010,BL22,SL25SIGCSE}. \textsc{neuron}~\cite{neuron_liu_2018} and \textsc{lantern}~\cite{towards_weiguo_2021} generate natural language descriptions of \textsc{qep}s. \textsc{mocha}~\cite{mocha} supports learner-friendly interaction and visualization of how alternative physical operators affect a chosen \textsc{qep}, but assumes learners already know which operators they want to explore. Picasso~\cite{picasso_jayant_2010} visualizes different \textsc{qep}s and their costs across the selectivity space. Our work complements these systems by enabling exploration of informative \textsc{aqp}s.

\section{Conclusions and Future Work}\label{sec:conclusion}  
This paper presents the informative plan selection (\textsc{tips}) problem, one piece of a larger problem of technology-enabled learning of relational query processing. We propose efficient approximate algorithms (\textsc{b-tips} and \textsc{i-tips}) that maximize the plan informativeness of the selected plans. These algorithms are integrated into our tool called \textsc{tips} for database education. Experimental studies, feedback from real-world learners, and analysis of academic performance demonstrate the effectiveness of our solutions. As part of future work, we wish to facilitate learning of various other components related to relational query processing such as the query plan selection algorithm, query rewriting, and cost estimation. As mentioned above, the \textsc{tips} algorithms can be also adapted for database administration, which is an interesting direction for future work.

\bibliographystyle{ACM-Reference-Format}
\bibliography{main}


\begin{thebibliography}{53}


\ifx \showCODEN    \undefined \def \showCODEN     #1{\unskip}     \fi
\ifx \showISBNx    \undefined \def \showISBNx     #1{\unskip}     \fi
\ifx \showISBNxiii \undefined \def \showISBNxiii  #1{\unskip}     \fi
\ifx \showISSN     \undefined \def \showISSN      #1{\unskip}     \fi
\ifx \showLCCN     \undefined \def \showLCCN      #1{\unskip}     \fi
\ifx \shownote     \undefined \def \shownote      #1{#1}          \fi
\ifx \showarticletitle \undefined \def \showarticletitle #1{#1}   \fi
\ifx \showURL      \undefined \def \showURL       {\relax}        \fi
\providecommand\bibfield[2]{#2}
\providecommand\bibinfo[2]{#2}
\providecommand\natexlab[1]{#1}
\providecommand\showeprint[2][]{arXiv:#2}

\bibitem[imd(2024)]%
        {imdb}
 \bibinfo{year}{2024}\natexlab{}.
\newblock \bibinfo{title}{{Internet Movie Data Base (IMDb)}}.
\newblock \bibinfo{howpublished}{\url{https://www.imdb.com/}}.
\newblock


\bibitem[Angel and Koudas(2011)]%
        {efficient_2011_albert}
\bibfield{author}{\bibinfo{person}{Albert Angel} {and} \bibinfo{person}{Nick Koudas}.} \bibinfo{year}{2011}\natexlab{}.
\newblock \showarticletitle{Efficient diversity-aware search}. In \bibinfo{booktitle}{\emph{Proceedings of the {ACM} {SIGMOD} International Conference on Management of Data, {SIGMOD} 2011, Athens, Greece, June 12-16, 2011}}. \bibinfo{publisher}{{ACM}}, \bibinfo{pages}{781--792}.
\newblock


\bibitem[Arora et~al\mbox{.}(2017)]%
        {AroraLM17}
\bibfield{author}{\bibinfo{person}{Sanjeev Arora}, \bibinfo{person}{Yingyu Liang}, {and} \bibinfo{person}{Tengyu Ma}.} \bibinfo{year}{2017}\natexlab{}.
\newblock \showarticletitle{A simple but tough-to-beat baseline for sentence embeddings}. In \bibinfo{booktitle}{\emph{International conference on learning representations}}.
\newblock


\bibitem[Bhowmick and Li(2022)]%
        {BL22}
\bibfield{author}{\bibinfo{person}{Sourav~S Bhowmick} {and} \bibinfo{person}{Hui Li}.} \bibinfo{year}{2022}\natexlab{}.
\newblock \showarticletitle{Towards Technology-Enabled Learning of Relational Query Processing}.
\newblock \bibinfo{journal}{\emph{IEEE Data Engineering Bulletin}} \bibinfo{volume}{45}, \bibinfo{number}{3} (\bibinfo{year}{2022}).
\newblock


\bibitem[Bhowmick and Li(2025)]%
        {SL25SIGCSE}
\bibfield{author}{\bibinfo{person}{Sourav~S. Bhowmick} {and} \bibinfo{person}{Hui Li}.} \bibinfo{year}{2025}\natexlab{}.
\newblock \showarticletitle{Experience Report on Using {LANTERN} in Teaching Relational Query Processing}. In \bibinfo{booktitle}{\emph{Proceedings of the 56th {ACM} Technical Symposium on Computer Science Education V. 1, {SIGCSE} {TS} 2025, Pittsburgh, PA, USA, 26 February 2025 - 1 March 2025}}. \bibinfo{publisher}{{ACM}}, \bibinfo{pages}{123--129}.
\newblock


\bibitem[Brusilovsky et~al\mbox{.}(2010)]%
        {brusilovskyLearningSQLProgramming2010}
\bibfield{author}{\bibinfo{person}{Peter Brusilovsky}, \bibinfo{person}{Sergey Sosnovsky}, \bibinfo{person}{Michael~V. Yudelson}, \bibinfo{person}{Danielle~H. Lee}, \bibinfo{person}{Vladimir Zadorozhny}, {and} \bibinfo{person}{Xin Zhou}.} \bibinfo{year}{2010}\natexlab{}.
\newblock \showarticletitle{Learning {{SQL}} Programming with Interactive Tools: {{From}} Integration to Personalization}.
\newblock \bibinfo{journal}{\emph{ACM Trans. Comput. Educ.}} \bibinfo{volume}{9}, \bibinfo{number}{4} (\bibinfo{date}{Jan.} \bibinfo{year}{2010}), \bibinfo{pages}{19:1--19:15}.
\newblock


\bibitem[Chaudhuri(1998)]%
        {SC98}
\bibfield{author}{\bibinfo{person}{Surajit Chaudhuri}.} \bibinfo{year}{1998}\natexlab{}.
\newblock \showarticletitle{An overview of query optimization in relational systems}. In \bibinfo{booktitle}{\emph{Proceedings of the seventeenth ACM SIGACT-SIGMOD-SIGART symposium on Principles of database systems}}. \bibinfo{pages}{34--43}.
\newblock


\bibitem[Council(TPC)(2021)]%
        {tpch}
\bibfield{author}{\bibinfo{person}{Transaction Processing~Performance Council(TPC)}.} \bibinfo{year}{2021}\natexlab{}.
\newblock \bibinfo{title}{TPC-H Version 2 and Version 3}.
\newblock \bibinfo{howpublished}{\url{http://www.tpc.org/tpch/}}.
\newblock


\bibitem[Danaparamita and Gatterbauer(2011)]%
        {DG11}
\bibfield{author}{\bibinfo{person}{Jonathan Danaparamita} {and} \bibinfo{person}{Wolfgang Gatterbauer}.} \bibinfo{year}{2011}\natexlab{}.
\newblock \showarticletitle{QueryViz: helping users understand SQL queries and their patterns}. In \bibinfo{booktitle}{\emph{Proceedings of the 14th International Conference on Extending Database Technology}}. \bibinfo{pages}{558--561}.
\newblock


\bibitem[Deng and Fan(2013)]%
        {on_ting_2013}
\bibfield{author}{\bibinfo{person}{Ting Deng} {and} \bibinfo{person}{Wenfei Fan}.} \bibinfo{year}{2013}\natexlab{}.
\newblock \showarticletitle{On the Complexity of Query Result Diversification}.
\newblock \bibinfo{journal}{\emph{Proc. {VLDB} Endow.}} \bibinfo{volume}{6}, \bibinfo{number}{8} (\bibinfo{year}{2013}), \bibinfo{pages}{577--588}.
\newblock


\bibitem[Drosou and Pitoura(2010)]%
        {search_drosou_2010}
\bibfield{author}{\bibinfo{person}{Marina Drosou} {and} \bibinfo{person}{Evaggelia Pitoura}.} \bibinfo{year}{2010}\natexlab{}.
\newblock \showarticletitle{Search result diversification}.
\newblock \bibinfo{journal}{\emph{ACM SIGMOD Record}} \bibinfo{volume}{39}, \bibinfo{number}{1} (\bibinfo{year}{2010}), \bibinfo{pages}{41--47}.
\newblock


\bibitem[Drosou and Pitoura(2012)]%
        {disc_ma_2012}
\bibfield{author}{\bibinfo{person}{Marina Drosou} {and} \bibinfo{person}{Evaggelia Pitoura}.} \bibinfo{year}{2012}\natexlab{}.
\newblock \showarticletitle{Disc diversity: result diversification based on dissimilarity and coverage}.
\newblock \bibinfo{journal}{\emph{arXiv preprint arXiv:1208.3533}} (\bibinfo{year}{2012}).
\newblock


\bibitem[Drosou and Pitoura(2013)]%
        {poikilo_dr_2013}
\bibfield{author}{\bibinfo{person}{Marina Drosou} {and} \bibinfo{person}{Evaggelia Pitoura}.} \bibinfo{year}{2013}\natexlab{}.
\newblock \showarticletitle{POIKILO: A Tool for Evaluating the Results of Diversification Models and Algorithms}.
\newblock \bibinfo{journal}{\emph{Proc. VLDB Endow.}} \bibinfo{volume}{6}, \bibinfo{number}{12} (\bibinfo{date}{Aug.} \bibinfo{year}{2013}), \bibinfo{pages}{1246--1249}.
\newblock


\bibitem[Fan et~al\mbox{.}(2013)]%
        {diversified_wenfei_2013}
\bibfield{author}{\bibinfo{person}{Wenfei Fan}, \bibinfo{person}{Xin Wang}, {and} \bibinfo{person}{Yinghui Wu}.} \bibinfo{year}{2013}\natexlab{}.
\newblock \showarticletitle{Diversified Top-k Graph Pattern Matching}.
\newblock \bibinfo{journal}{\emph{Proc. {VLDB} Endow.}} \bibinfo{volume}{6}, \bibinfo{number}{13} (\bibinfo{year}{2013}), \bibinfo{pages}{1510--1521}.
\newblock


\bibitem[Gatterbauer(2024)]%
        {Gatterbauer2024ACT}
\bibfield{author}{\bibinfo{person}{Wolfgang Gatterbauer}.} \bibinfo{year}{2024}\natexlab{}.
\newblock \showarticletitle{A Comprehensive Tutorial on Over 100 Years of Diagrammatic Representations of Logical Statements and Relational Queries}.
\newblock \bibinfo{journal}{\emph{2024 IEEE 40th International Conference on Data Engineering (ICDE)}} (\bibinfo{year}{2024}), \bibinfo{pages}{5387--5392}.
\newblock
\urldef\tempurl%
\url{https://api.semanticscholar.org/CorpusID:268820151}
\showURL{%
\tempurl}


\bibitem[Gatterbauer et~al\mbox{.}(2022)]%
        {Gat22}
\bibfield{author}{\bibinfo{person}{Wolfgang Gatterbauer}, \bibinfo{person}{Cody Dunne}, \bibinfo{person}{HV Jagadish}, {and} \bibinfo{person}{Mirek Riedewald}.} \bibinfo{year}{2022}\natexlab{}.
\newblock \showarticletitle{Principles of Query Visualization}.
\newblock \bibinfo{journal}{\emph{arXiv preprint arXiv:2208.01613}} (\bibinfo{year}{2022}).
\newblock


\bibitem[Golenberg et~al\mbox{.}(2008)]%
        {keyword_2008_konstantin}
\bibfield{author}{\bibinfo{person}{Konstantin Golenberg}, \bibinfo{person}{Benny Kimelfeld}, {and} \bibinfo{person}{Yehoshua Sagiv}.} \bibinfo{year}{2008}\natexlab{}.
\newblock \showarticletitle{Keyword proximity search in complex data graphs}. In \bibinfo{booktitle}{\emph{Proceedings of the {ACM} {SIGMOD} International Conference on Management of Data, {SIGMOD} 2008, Vancouver, BC, Canada, June 10-12, 2008}}. \bibinfo{publisher}{{ACM}}, \bibinfo{pages}{927--940}.
\newblock


\bibitem[Gollapudi and Sharma(2009)]%
        {axiomatic_go_2009}
\bibfield{author}{\bibinfo{person}{Sreenivas Gollapudi} {and} \bibinfo{person}{Aneesh Sharma}.} \bibinfo{year}{2009}\natexlab{}.
\newblock \showarticletitle{An Axiomatic Approach for Result Diversification}. In \bibinfo{booktitle}{\emph{Proceedings of the 18th International Conference on World Wide Web}} (Madrid, Spain) \emph{(\bibinfo{series}{WWW '09})}. \bibinfo{publisher}{Association for Computing Machinery}, \bibinfo{pages}{381--390}.
\newblock
\showISBNx{9781605584874}


\bibitem[Graefe(1995)]%
        {cascade_go_1995}
\bibfield{author}{\bibinfo{person}{Goetz Graefe}.} \bibinfo{year}{1995}\natexlab{}.
\newblock \showarticletitle{The Cascades Framework for Query Optimization}.
\newblock \bibinfo{journal}{\emph{{IEEE} Data Eng. Bull.}} \bibinfo{volume}{18}, \bibinfo{number}{3} (\bibinfo{year}{1995}), \bibinfo{pages}{19--29}.
\newblock


\bibitem[Graefe and McKenna({[n.\,d.]})]%
        {DBLP:conf/icde/GraefeM93}
\bibfield{author}{\bibinfo{person}{Goetz Graefe} {and} \bibinfo{person}{William~J. McKenna}.} \bibinfo{year}{[n.\,d.]}\natexlab{}.
\newblock \showarticletitle{The Volcano Optimizer Generator: Extensibility and Efficient Search}. In \bibinfo{booktitle}{\emph{Proceedings of the Ninth International Conference on Data Engineering, April 19-23, 1993, Vienna, Austria}}. \bibinfo{publisher}{{IEEE} Computer Society}, \bibinfo{pages}{209--218}.
\newblock


\bibitem[Haritsa(2010)]%
        {picasso_jayant_2010}
\bibfield{author}{\bibinfo{person}{Jayant~R. Haritsa}.} \bibinfo{year}{2010}\natexlab{}.
\newblock \showarticletitle{The Picasso Database Query Optimizer Visualizer}.
\newblock \bibinfo{journal}{\emph{Proc. {VLDB} Endow.}} \bibinfo{volume}{3}, \bibinfo{number}{2} (\bibinfo{year}{2010}), \bibinfo{pages}{1517--1520}.
\newblock


\bibitem[Hu et~al\mbox{.}(2024)]%
        {Hu2024QrHintAH}
\bibfield{author}{\bibinfo{person}{Yihao Hu}, \bibinfo{person}{Amir Gilad}, \bibinfo{person}{Kristin Stephens-Martinez}, \bibinfo{person}{Sudeepa Roy}, {and} \bibinfo{person}{Jun Yang}.} \bibinfo{year}{2024}\natexlab{}.
\newblock \showarticletitle{Qr-Hint: Actionable Hints Towards Correcting Wrong SQL Queries}.
\newblock \bibinfo{journal}{\emph{ArXiv}}  \bibinfo{volume}{abs/2404.04352} (\bibinfo{year}{2024}).
\newblock
\urldef\tempurl%
\url{https://api.semanticscholar.org/CorpusID:269005402}
\showURL{%
\tempurl}


\bibitem[Hu et~al\mbox{.}(2022)]%
        {HM+22}
\bibfield{author}{\bibinfo{person}{Yihao Hu}, \bibinfo{person}{Zhengjie Miao}, \bibinfo{person}{Zhiming Leong}, \bibinfo{person}{Haechan Lim}, \bibinfo{person}{Zachary Zheng}, \bibinfo{person}{Sudeepa Roy}, \bibinfo{person}{Kristin Stephens-Martinez}, {and} \bibinfo{person}{Jun Yang}.} \bibinfo{year}{2022}\natexlab{}.
\newblock \showarticletitle{I-Rex: An interactive relational query debugger for SQL}. In \bibinfo{booktitle}{\emph{Proceedings of the 53rd ACM Technical Symposium on Computer Science Education V. 2}}. \bibinfo{pages}{1180--1180}.
\newblock


\bibitem[Huang et~al\mbox{.}(2013)]%
        {top_xin_2013}
\bibfield{author}{\bibinfo{person}{Xin Huang}, \bibinfo{person}{Hong Cheng}, \bibinfo{person}{Rong{-}Hua Li}, \bibinfo{person}{Lu Qin}, {and} \bibinfo{person}{Jeffrey~Xu Yu}.} \bibinfo{year}{2013}\natexlab{}.
\newblock \showarticletitle{Top-K Structural Diversity Search in Large Networks}.
\newblock \bibinfo{journal}{\emph{Proc. {VLDB} Endow.}} \bibinfo{volume}{6}, \bibinfo{number}{13} (\bibinfo{year}{2013}), \bibinfo{pages}{1618--1629}.
\newblock


\bibitem[Ibaraki and Kameda(1984)]%
        {optimal_ibaraki_1984}
\bibfield{author}{\bibinfo{person}{Toshihide Ibaraki} {and} \bibinfo{person}{Tiko Kameda}.} \bibinfo{year}{1984}\natexlab{}.
\newblock \showarticletitle{On the optimal nesting order for computing n-relational joins}.
\newblock \bibinfo{journal}{\emph{ACM Transactions on Database Systems (TODS)}} \bibinfo{volume}{9}, \bibinfo{number}{3} (\bibinfo{year}{1984}), \bibinfo{pages}{482--502}.
\newblock


\bibitem[James et~al\mbox{.}(2013)]%
        {jam2013introduction}
\bibfield{author}{\bibinfo{person}{Gareth James}, \bibinfo{person}{Daniela Witten}, \bibinfo{person}{Trevor Hastie}, \bibinfo{person}{Robert Tibshirani}, {et~al\mbox{.}}} \bibinfo{year}{2013}\natexlab{}.
\newblock \bibinfo{booktitle}{\emph{An introduction to statistical learning}}. Vol.~\bibinfo{volume}{112}.
\newblock \bibinfo{publisher}{Springer}.
\newblock


\bibitem[Kokkalis et~al\mbox{.}(2012)]%
        {KV+12}
\bibfield{author}{\bibinfo{person}{Andreas Kokkalis}, \bibinfo{person}{Panagiotis Vagenas}, \bibinfo{person}{Alexandros Zervakis}, \bibinfo{person}{Alkis Simitsis}, \bibinfo{person}{Georgia Koutrika}, {and} \bibinfo{person}{Yannis~E. Ioannidis}.} \bibinfo{year}{2012}\natexlab{}.
\newblock \showarticletitle{Logos: a system for translating queries into narratives}. In \bibinfo{booktitle}{\emph{Proceedings of the {ACM} {SIGMOD} International Conference on Management of Data, {SIGMOD} 2012, Scottsdale, AZ, USA, May 20-24, 2012}}. \bibinfo{publisher}{{ACM}}, \bibinfo{pages}{673--676}.
\newblock


\bibitem[Leis et~al\mbox{.}(2015)]%
        {good_viktor_2015}
\bibfield{author}{\bibinfo{person}{Viktor Leis}, \bibinfo{person}{Andrey Gubichev}, \bibinfo{person}{Atanas Mirchev}, \bibinfo{person}{Peter~A. Boncz}, \bibinfo{person}{Alfons Kemper}, {and} \bibinfo{person}{Thomas Neumann}.} \bibinfo{year}{2015}\natexlab{}.
\newblock \showarticletitle{How Good Are Query Optimizers, Really?}
\newblock \bibinfo{journal}{\emph{Proc. {VLDB} Endow.}} \bibinfo{volume}{9}, \bibinfo{number}{3} (\bibinfo{year}{2015}), \bibinfo{pages}{204--215}.
\newblock


\bibitem[Leventidis et~al\mbox{.}(2020)]%
        {LZ+20}
\bibfield{author}{\bibinfo{person}{Aristotelis Leventidis}, \bibinfo{person}{Jiahui Zhang}, \bibinfo{person}{Cody Dunne}, \bibinfo{person}{Wolfgang Gatterbauer}, \bibinfo{person}{HV Jagadish}, {and} \bibinfo{person}{Mirek Riedewald}.} \bibinfo{year}{2020}\natexlab{}.
\newblock \showarticletitle{QueryVis: Logic-based diagrams help users understand complicated SQL queries faster}. In \bibinfo{booktitle}{\emph{Proceedings of the 2020 ACM SIGMOD International Conference on Management of Data}}. \bibinfo{pages}{2303--2318}.
\newblock


\bibitem[Li and Liu(2007)]%
        {normalized_yu_2007}
\bibfield{author}{\bibinfo{person}{Yujian Li} {and} \bibinfo{person}{Bi Liu}.} \bibinfo{year}{2007}\natexlab{}.
\newblock \showarticletitle{A Normalized Levenshtein Distance Metric}.
\newblock \bibinfo{journal}{\emph{{IEEE} Trans. Pattern Anal. Mach. Intell.}} \bibinfo{volume}{29}, \bibinfo{number}{6} (\bibinfo{year}{2007}), \bibinfo{pages}{1091--1095}.
\newblock


\bibitem[Liu et~al\mbox{.}(2018)]%
        {neuron_liu_2018}
\bibfield{author}{\bibinfo{person}{Siyuan Liu}, \bibinfo{person}{Sourav~S Bhowmick}, \bibinfo{person}{Wanlu Zhang}, \bibinfo{person}{Shu Wang}, \bibinfo{person}{Wanyi Huang}, {and} \bibinfo{person}{Shafiq Joty}.} \bibinfo{year}{2018}\natexlab{}.
\newblock \showarticletitle{NEURON: Query Optimization Meets Natural Language Processing For Augmenting Database Education}.
\newblock \bibinfo{journal}{\emph{arXiv preprint arXiv:1805.05670}} (\bibinfo{year}{2018}).
\newblock


\bibitem[McFee and Lanckriet(2010)]%
        {mcf2010metric}
\bibfield{author}{\bibinfo{person}{Brian McFee} {and} \bibinfo{person}{Gert Lanckriet}.} \bibinfo{year}{2010}\natexlab{}.
\newblock \showarticletitle{Metric learning to rank}.
\newblock  (\bibinfo{year}{2010}).
\newblock


\bibitem[Miao et~al\mbox{.}(2019)]%
        {MRY19}
\bibfield{author}{\bibinfo{person}{Zhengjie Miao}, \bibinfo{person}{Sudeepa Roy}, {and} \bibinfo{person}{Jun Yang}.} \bibinfo{year}{2019}\natexlab{}.
\newblock \showarticletitle{Explaining wrong queries using small examples}. In \bibinfo{booktitle}{\emph{Proceedings of the 2019 International Conference on Management of Data}}. \bibinfo{pages}{503--520}.
\newblock


\bibitem[Miedema et~al\mbox{.}(2022a)]%
        {MAF21}
\bibfield{author}{\bibinfo{person}{Daphne Miedema}, \bibinfo{person}{Efthimia Aivaloglou}, {and} \bibinfo{person}{George Fletcher}.} \bibinfo{year}{2022}\natexlab{a}.
\newblock \showarticletitle{Identifying SQL misconceptions of novices: Findings from a think-aloud study}.
\newblock \bibinfo{journal}{\emph{ACM Inroads}} \bibinfo{volume}{13}, \bibinfo{number}{1} (\bibinfo{year}{2022}), \bibinfo{pages}{52--65}.
\newblock


\bibitem[Miedema and Fletcher(2021)]%
        {MF21}
\bibfield{author}{\bibinfo{person}{Daphne Miedema} {and} \bibinfo{person}{George Fletcher}.} \bibinfo{year}{2021}\natexlab{}.
\newblock \showarticletitle{SQLVis: Visual query representations for supporting SQL learners}. In \bibinfo{booktitle}{\emph{2021 IEEE Symposium on Visual Languages and Human-Centric Computing (VL/HCC)}}. IEEE, \bibinfo{pages}{1--9}.
\newblock


\bibitem[Miedema et~al\mbox{.}(2022b)]%
        {miedemaExpertPerspectivesStudent2022}
\bibfield{author}{\bibinfo{person}{Daphne Miedema}, \bibinfo{person}{George Fletcher}, {and} \bibinfo{person}{Efthimia Aivaloglou}.} \bibinfo{year}{2022}\natexlab{b}.
\newblock \showarticletitle{Expert Perspectives on Student Errors in {{SQL}}}.
\newblock \bibinfo{journal}{\emph{ACM Trans. Comput. Educ.}} \bibinfo{volume}{23}, \bibinfo{number}{1} (\bibinfo{date}{Dec.} \bibinfo{year}{2022}), \bibinfo{pages}{11:1--11:28}.
\newblock


\bibitem[Ramakrishnan and Gehrke(2003)]%
        {DBLP:books/daglib/0011128}
\bibfield{author}{\bibinfo{person}{Raghu Ramakrishnan} {and} \bibinfo{person}{Johannes Gehrke}.} \bibinfo{year}{2003}\natexlab{}.
\newblock \bibinfo{booktitle}{\emph{Database management systems {(3.} ed.)}}.
\newblock \bibinfo{publisher}{McGraw-Hill}.
\newblock
\showISBNx{978-0-07-115110-8}


\bibitem[Ravi et~al\mbox{.}(2018)]%
        {approximation_ss_2018}
\bibfield{author}{\bibinfo{person}{S.~S. Ravi}, \bibinfo{person}{Daniel~J. Rosenkrantz}, {and} \bibinfo{person}{Giri~Kumar Tayi}.} \bibinfo{year}{2018}\natexlab{}.
\newblock \showarticletitle{Approximation Algorithms for Facility Dispersion}.
\newblock In \bibinfo{booktitle}{\emph{Handbook of Approximation Algorithms and Metaheuristics, Second Edition, Volume 2: Contemporary and Emerging Applications}}. \bibinfo{publisher}{Chapman and Hall/CRC}.
\newblock


\bibitem[Scheibli et~al\mbox{.}(2015)]%
        {SDB15}
\bibfield{author}{\bibinfo{person}{Daniel Scheibli}, \bibinfo{person}{Christian Dinse}, {and} \bibinfo{person}{Alexander Boehm}.} \bibinfo{year}{2015}\natexlab{}.
\newblock \showarticletitle{QE3D: Interactive Visualization and Exploration of Complex, Distributed Query Plans}. In \bibinfo{booktitle}{\emph{Proceedings of the 2015 ACM SIGMOD International Conference on Management of Data}}. \bibinfo{pages}{877--881}.
\newblock


\bibitem[Shawe{-}Taylor and Cristianini(2004)]%
        {kernel_jo_2004}
\bibfield{author}{\bibinfo{person}{John Shawe{-}Taylor} {and} \bibinfo{person}{Nello Cristianini}.} \bibinfo{year}{2004}\natexlab{}.
\newblock \bibinfo{booktitle}{\emph{Kernel Methods for Pattern Analysis}}.
\newblock \bibinfo{publisher}{Cambridge University Press}.
\newblock
\showISBNx{9780511809682}


\bibitem[Smola and Vishwanathan(2003)]%
        {fast_sm_2003}
\bibfield{author}{\bibinfo{person}{Alex Smola} {and} \bibinfo{person}{S.v.n. Vishwanathan}.} \bibinfo{year}{2003}\natexlab{}.
\newblock \showarticletitle{Fast Kernels for String and Tree Matching}. In \bibinfo{booktitle}{\emph{Advances in Neural Information Processing Systems}}, Vol.~\bibinfo{volume}{15}.
\newblock


\bibitem[Soliman et~al\mbox{.}(2014)]%
        {orca_mo_2014}
\bibfield{author}{\bibinfo{person}{Mohamed~A. Soliman}, \bibinfo{person}{Lyublena Antova}, \bibinfo{person}{Venkatesh Raghavan}, \bibinfo{person}{Amr El{-}Helw}, \bibinfo{person}{Zhongxian Gu}, \bibinfo{person}{Entong Shen}, \bibinfo{person}{George~C. Caragea}, \bibinfo{person}{Carlos Garcia{-}Alvarado}, \bibinfo{person}{Foyzur Rahman}, \bibinfo{person}{Michalis Petropoulos}, \bibinfo{person}{Florian Waas}, \bibinfo{person}{Sivaramakrishnan Narayanan}, \bibinfo{person}{Konstantinos Krikellas}, {and} \bibinfo{person}{Rhonda Baldwin}.} \bibinfo{year}{2014}\natexlab{}.
\newblock \showarticletitle{Orca: a modular query optimizer architecture for big data}. In \bibinfo{booktitle}{\emph{International Conference on Management of Data, {SIGMOD} 2014, Snowbird, UT, USA, June 22-27, 2014}}. \bibinfo{publisher}{{ACM}}, \bibinfo{pages}{337--348}.
\newblock


\bibitem[Taipalus and Sepp{\"a}nen(2020)]%
        {taipalusSQLEducationSystematic2020}
\bibfield{author}{\bibinfo{person}{Toni Taipalus} {and} \bibinfo{person}{Ville Sepp{\"a}nen}.} \bibinfo{year}{2020}\natexlab{}.
\newblock \showarticletitle{{{SQL}} Education: A Systematic Mapping Study and Future Research Agenda}.
\newblock \bibinfo{journal}{\emph{ACM Transactions on Computing Education}} \bibinfo{volume}{20}, \bibinfo{number}{3} (\bibinfo{date}{Sept.} \bibinfo{year}{2020}), \bibinfo{pages}{1--33}.
\newblock
\showISSN{1946-6226}


\bibitem[Taipalus et~al\mbox{.}(2018)]%
        {taipalusErrorsComplicationsSQL2018}
\bibfield{author}{\bibinfo{person}{Toni Taipalus}, \bibinfo{person}{Mikko Siponen}, {and} \bibinfo{person}{Tero Vartiainen}.} \bibinfo{year}{2018}\natexlab{}.
\newblock \showarticletitle{Errors and Complications in {{SQL}} Query Formulation}.
\newblock \bibinfo{journal}{\emph{ACM Transactions on Computing Education}} \bibinfo{volume}{18}, \bibinfo{number}{3} (\bibinfo{date}{Sept.} \bibinfo{year}{2018}), \bibinfo{pages}{1--29}.
\newblock
\showISSN{1946-6226}


\bibitem[Tan et~al\mbox{.}(2022)]%
        {mocha}
\bibfield{author}{\bibinfo{person}{Jess Tan}, \bibinfo{person}{Desmond Yeo}, \bibinfo{person}{Rachael Neoh}, \bibinfo{person}{Huey-Eng Chua}, {and} \bibinfo{person}{Sourav~S Bhowmick}.} \bibinfo{year}{2022}\natexlab{}.
\newblock \showarticletitle{MOCHA: a tool for visualizing impact of operator choices in query execution plans for database education}.
\newblock \bibinfo{journal}{\emph{Proceedings of the VLDB Endowment}} \bibinfo{volume}{15}, \bibinfo{number}{12} (\bibinfo{year}{2022}), \bibinfo{pages}{3602--3605}.
\newblock


\bibitem[Tian et~al\mbox{.}(2024)]%
        {Tian2024SQLucidGN}
\bibfield{author}{\bibinfo{person}{Yuan Tian}, \bibinfo{person}{Jonathan~K. Kummerfeld}, \bibinfo{person}{Toby~Jia{-}Jun Li}, {and} \bibinfo{person}{Tianyi Zhang}.} \bibinfo{year}{2024}\natexlab{}.
\newblock \showarticletitle{SQLucid: Grounding Natural Language Database Queries with Interactive Explanations}.
\newblock \bibinfo{journal}{\emph{CoRR}}  \bibinfo{volume}{abs/2409.06178} (\bibinfo{year}{2024}).
\newblock
\showeprint[arXiv]{2409.06178}
\href{https://doi.org/10.48550/ARXIV.2409.06178}{doi:\nolinkurl{10.48550/ARXIV.2409.06178}}


\bibitem[Wang et~al\mbox{.}(2018)]%
        {top_bin_2018}
\bibfield{author}{\bibinfo{person}{Bin Wang}, \bibinfo{person}{Rui Zhu}, \bibinfo{person}{Xiaochun Yang}, {and} \bibinfo{person}{Guoren Wang}.} \bibinfo{year}{2018}\natexlab{}.
\newblock \showarticletitle{Top-K representative documents query over geo-textual data stream}.
\newblock \bibinfo{journal}{\emph{World Wide Web}} \bibinfo{volume}{21}, \bibinfo{number}{2} (\bibinfo{year}{2018}), \bibinfo{pages}{537--555}.
\newblock


\bibitem[Wang et~al\mbox{.}(2023)]%
        {arena}
\bibfield{author}{\bibinfo{person}{Hu Wang}, \bibinfo{person}{Hui Li}, \bibinfo{person}{Sourav~S Bhowmick}, {and} \bibinfo{person}{Baochao Xu}.} \bibinfo{year}{2023}\natexlab{}.
\newblock \showarticletitle{ARENA: Alternative Relational Query Plan Exploration for Database Education}. In \bibinfo{booktitle}{\emph{Companion of the 2023 International Conference on Management of Data}}. \bibinfo{pages}{107--110}.
\newblock


\bibitem[Wang et~al\mbox{.}(2024)]%
        {wangFalsePositivesDeceptive2024}
\bibfield{author}{\bibinfo{person}{Jinshui Wang}, \bibinfo{person}{Shuguang Chen}, \bibinfo{person}{Zhengyi Tang}, \bibinfo{person}{Pengchen Lin}, {and} \bibinfo{person}{Yupeng Wang}.} \bibinfo{year}{2024}\natexlab{}.
\newblock \showarticletitle{False Positives and Deceptive Errors in {{SQL}} Assessment: A Large-Scale Analysis of Online Judge Systems}.
\newblock \bibinfo{journal}{\emph{ACM Transactions on Computing Education}} \bibinfo{volume}{24}, \bibinfo{number}{3} (\bibinfo{date}{Sept.} \bibinfo{year}{2024}), \bibinfo{pages}{1--23}.
\newblock
\showISSN{1946-6226}


\bibitem[Wang et~al\mbox{.}(2021)]%
        {towards_weiguo_2021}
\bibfield{author}{\bibinfo{person}{Weiguo Wang}, \bibinfo{person}{Sourav~S. Bhowmick}, \bibinfo{person}{Hui Li}, \bibinfo{person}{Shafiq~R. Joty}, \bibinfo{person}{Siyuan Liu}, {and} \bibinfo{person}{Peng Chen}.} \bibinfo{year}{2021}\natexlab{}.
\newblock \showarticletitle{Towards Enhancing Database Education: Natural Language Generation Meets Query Execution Plans}. In \bibinfo{booktitle}{\emph{{SIGMOD} '21: International Conference on Management of Data, Virtual Event, China, June 20-25, 2021}}. \bibinfo{publisher}{{ACM}}, \bibinfo{pages}{1933--1945}.
\newblock


\bibitem[Yuan et~al\mbox{.}(2015)]%
        {diversified_long_2015}
\bibfield{author}{\bibinfo{person}{Long Yuan}, \bibinfo{person}{Lu Qin}, \bibinfo{person}{Xuemin Lin}, \bibinfo{person}{Lijun Chang}, {and} \bibinfo{person}{Wenjie Zhang}.} \bibinfo{year}{2015}\natexlab{}.
\newblock \showarticletitle{Diversified top-k clique search}. In \bibinfo{booktitle}{\emph{31st {IEEE} International Conference on Data Engineering, {ICDE} 2015, Seoul, South Korea, April 13-17, 2015}}. \bibinfo{publisher}{{IEEE} Computer Society}, \bibinfo{pages}{387--398}.
\newblock


\bibitem[Zhao et~al\mbox{.}(2011)]%
        {broad_zhao_2011}
\bibfield{author}{\bibinfo{person}{Feng Zhao}, \bibinfo{person}{Xiaolong Zhang}, \bibinfo{person}{Anthony~KH Tung}, {and} \bibinfo{person}{Gang Chen}.} \bibinfo{year}{2011}\natexlab{}.
\newblock \showarticletitle{Broad: Diversified keyword search in databases}.
\newblock  (\bibinfo{year}{2011}).
\newblock


\bibitem[Zhu et~al\mbox{.}(2017)]%
        {sap_rui_2017}
\bibfield{author}{\bibinfo{person}{Rui Zhu}, \bibinfo{person}{Bin Wang}, \bibinfo{person}{Xiaochun Yang}, \bibinfo{person}{Baihua Zheng}, {and} \bibinfo{person}{Guoren Wang}.} \bibinfo{year}{2017}\natexlab{}.
\newblock \showarticletitle{{SAP:} Improving Continuous Top-K Queries Over Streaming Data}.
\newblock \bibinfo{journal}{\emph{{IEEE} Trans. Knowl. Data Eng.}} \bibinfo{volume}{29}, \bibinfo{number}{6} (\bibinfo{year}{2017}), \bibinfo{pages}{1310--1328}.
\newblock


\end{thebibliography}

\appendix
\renewcommand{\thefigure}{\thesection.\arabic{figure}}
\setcounter{figure}{0}
\renewcommand{\thetable}{\thesection.\arabic{table}}
\setcounter{table}{0}
\section{Proofs}\label{apx:a}
\subsection{Proof of \Cref{np-bapq}}
 \begin{proof} (Sketch).
 Let us consider the facility dispersion problem which is NP-hard~\cite{approximation_ss_2018}. In that problem, one is given a facility set $S$, a function $f$ of distance and an integer $k (k < |F|)$. The objective is to find a subset $S' (|S'| = k)$ of $S$ so that the given function on $f$ is maximized. Because the facility dispersion problem is a special case of \textsc{b-tips} where $\{\pi^*\} = \emptyset$ and $U(\cdot)=f$, \textsc{b-tips} is also NP-hard.
 \end{proof}
 
\subsection{Proof of Lemma \ref{lem:metric}}
\begin{proof}
Subtree kernel $\kappa_S\left(T_a, T_b\right)$ is a valid kernel, and we map the tree to one metric space. Consequently, we define a norm $||T||=\sqrt{\kappa\left(T, T\right)}$ based on the normalized kernel, and then the $L_2$ distance metric on this space is
\begin{align*}
 L_2(T_a, T_b) &= ||T_a-T_b||\\
               &= \sqrt{\kappa(T_a, T_a)+\kappa(T_b, T_b)-2\kappa(T_a, T_b)} \\
               &= \sqrt{2\left(1-\kappa(T_a, T_b)\right)}
\end{align*}
The above deduction relies on $\kappa(T_a, T_a) = \kappa(T_b, T_b) = 1$. Finally, we conclude that $s\_dist(T_a, T_b)$ is a metric since $L_2(T_a, T_b)=\sqrt{2}s\_dist(T_a, T_b)$.
\end{proof}

\subsection{Proof of Lemma \ref{disttri}}
\begin{proof}
 $dist(\cdot)$ is a linear combination of $s\_dist$, $c\_dist$ and $cost\_dist$. Because all of them are metric, $dist(\cdot)$ is a metric too. $dist(\cdot)$ satisfies the triangle inequality. Then we only need to prove that $\frac{1-\lambda}{2}\left(rel(a)+rel(b)\right)$ satisfies the triangle inequality. Let $R(a,b) = rel(a)+rel(b)$, then
 \begin{align*}
     R(a,b)+R(a,c) &= rel(a)+rel(b)+rel(a)+rel(c)\\
                   &= 2\cdot rel(a)+rel(b)+rel(c)\\
                   &= 2\cdot rel(a) + R(b, c)
 \end{align*}
 Since $rel(a)\ge 0$, $R(a,b)+R(a,c)\ge R(b,c)$. Similarly, we can prove that $R(a,b)-R(a,c)\le R(b,c)$. Therefore the $Dist(\cdot)$ consists of $R(\cdot)$ and $dist(\cdot)$ also satisfies the triangle inequality.
\end{proof}

\subsection{Proof of \Cref{the:approximate}}
\begin{proof}
Suppose $\Pi_{OPT}$ is the optimal solution of the \textsc{b-tips} problem. $OPT = U(\Pi_{OPT}\cup \{\pi^*\})$ is the minimum distance. Let $\Pi$ and $\Pi^{\ast}$  represent the result selected by Algorithm~\ref{alg: greedy} and all plans, respectively. Let $f(P)=\min \limits_{x, y \in P} \{Dist(x, y)\}$. We will prove the condition $f(\Pi \cup \{\pi^*\}) \ge OPT/2$ holds after each addition to $\Pi$. Since $f(\Pi)$ represents the minimum distance of result of Algorithm~\ref{alg: greedy} after the last addition to $\Pi$, the theorem would then follow. 

The first addition inserts a plan that is farthest from $\pi^*$. There are two cases about $OPT$: $OPT=dist(\pi^*, \pi_i)$ or $OPT=dist(\pi_j, \pi_i)$ where $\pi_i, \pi_j \neq \pi^*$. In either case, $f(\Pi\cup \{\pi^*\}) > OPT$, Hence, the condition clearly holds after the first addition. 

Assume that the condition holds after $k$ additions to $\Pi$, where $k\ge 1$. We will prove by contradiction that the condition holds after the $(k+1)^{th}$ addition to $\Pi$ as well. Let $\Pi_k={\pi_1, \pi_2, \ldots, \pi_k, \pi^*}$ denote the set $\Pi$ after $k$ additions\eat{ (Note that $\pi$ is added as the \textsc{qep}, not included in the alternative plans)}. Since we are assuming that the above condition does not hold after the $(k + 1)^{th}$ addition, it must be that for each $\pi_j \in \Pi^{\ast}-\Pi_k$, there is a plan $\pi_i \in \Pi_k$ such that $dist(\pi_i, \pi_j) < OPT/2$. We describe this situation by saying that $\pi_j$ is blocked by $\pi_i$.

Let $\Pi'=\Pi_{OPT}\cap (\Pi^{\ast}\setminus\Pi_k)$. Note that $k< \left|\Pi_{OPT}\right|$ since an additional plan is to be added. Thus $\Pi' \ge 1$. It is easy to verify that if $|\Pi'|=1$, then the condition will hold after the $(k+1)^{th}$ addition. Therefore, assume that $|\Pi'|=r\ge 2$. Furthermore, let $\Pi'={y_1, y_2, \ldots, y_r}$. Since $\Pi' \subseteq \Pi_{opt}$, we must have $Dist(y_i, y_j)\ge OPT$ for $1\le i, j\le r, (i\neq j)$. Since the distances satisfy the triangle inequality (Lemma~\ref{disttri}), it is possible to show that no two distinct plans in $\Pi'$ can be blocked by the same node $\pi_i \in \Pi_k$. Moreover, if there is a blocking plan in $\Pi_k$ that is also in $\Pi_{OPT}$, the minimum distance of $\Pi_{OPT}$ must be smaller than $OPT$, because this is a blocking plan whose minimum distance with $\Pi_{OPT}\setminus\Pi_k$ is smaller than $OPT/2$. Therefore, none of the blocking plans in $\Pi_k$ can be in $\Pi_{OPT}$ and $|\Pi'|=|\Pi_{OPT}|$. It is possible to show that $\Pi_k$ must contain $\left|\Pi_{OPT}\right|$ plans to block all plans in $\Pi'$.  However, this contradicts our initial assumption that $k < \left|\Pi_{OPT}\right|$. Thus, the condition must hold after the $(k+1)^{th}$ addition.
\end{proof}

\section{Final questions in the exam of 2023-2025}\label{apx:b}
\subsection{2023 -- Part V: Comprehensive Problem (20 points)}

\begin{table}[H]
  \centering \scriptsize
  \caption{Statistics of 2023 Teaching Database Relations.}
  \label{tab:2023_relations}
  \begin{tabular}{lccc}
  \toprule
  Relation & \# Tuples & \# Physical Pages & Primary Key \\
  \midrule
  Teacher(Tno, Tname, Tsex, Tage, Tdept) & 2000 & 5 & Tno \\
  Course(Cno, Cname, Credit, Cpno) & 1200 & 3 & Cno \\
  TC(Tno, Cno, Year) & 6000 & 14 & (Tno, Cno) \\
  \bottomrule
  \end{tabular}
\end{table}

\begin{figure}[H]
	\centering
	\subfloat[Optimal Plan]{
		\label{fig:2023-opt}
		\includegraphics[width=0.485\linewidth]{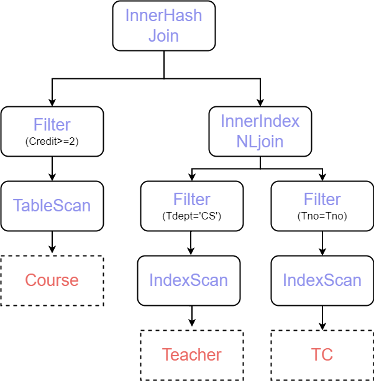}
		}
	\subfloat[Alternative Plan]{
		\label{fig:2023-alt}
		\includegraphics[width=0.475\linewidth]{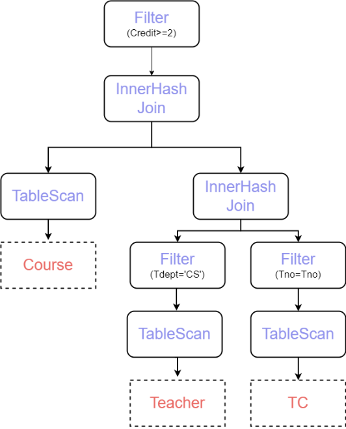}
		}
	\caption{Optimal Plan and Alternative Plan for the 2023 query.}
	\label{fig:2023-plans}
\end{figure}

Consider the teaching database in Table~\ref{tab:2023_relations}.

Assume:
\begin{itemize}
	\item TC.Tno is a foreign key referencing Teacher.Tno.
	\item TC.Cno is a foreign key referencing Course.Cno.
	\item Teachers are uniformly distributed across Tdept, and TC are uniformly distributed across Year.
	\item Teacher names are unique.
	\item Each table has an in-memory B$^+$ tree index built on its primary key.
	\item Teacher.Tdept also has a B$^+$ tree index.
	\item All intermediate results can be kept in memory during query execution.
\end{itemize}

By the end of 2022, the university must report every course with credit $\geq 2$ that the ``CS'' department teachers offered that year. Determine the qualifying course names and address the following:

\begin{enumerate}
	\item Figures~\ref{fig:2023-opt} and~\ref{fig:2023-alt} show the optimal plan and an alternative plan for the SQL query. Explain the advantages of the optimal plan over the alternative plan. (10 points)
	\item Other than the two plans in part~(1), describe at least one additional executable plan for requirement~(1) and explain why the optimizer might reject it. (10 points)
\end{enumerate}

\subsection{2024 -- Part V: Query Understanding (20 points)}

\begin{table}[H]
  \centering \scriptsize
  \caption{Statistics of 2024 Advising Database Relations.}
  \label{tab:2024_relations}
  \begin{tabular}{lccc}
  \toprule
  Relation & \# Tuples & \# Physical Pages & Primary Key \\
  \midrule
  Student(id, name, sex, age, dept) & 20000 & 50 & id \\
  Instructor(id, name, age, dept) & 1000 & 3 & Cno \\
  Advisor(s\_id, i\_id, Course) & 1000000 & 600 & (s\_id, i\_id) \\
  \bottomrule
  \end{tabular}
\end{table}

\begin{figure}[H]
	\centering
	\subfloat[Plan A.]{
		\label{fig:2024-planA}
		\includegraphics[width=0.49\linewidth]{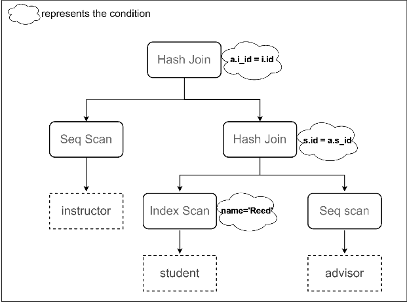}
		}
	\subfloat[Plan B.]{
		\label{fig:2024-planB}
		\includegraphics[width=0.49\linewidth]{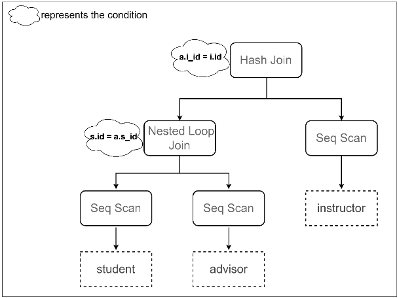}
		}
	\caption{Optimal Plan and Alternative Plan for the 2024 query.}
	\label{fig:2024-plans}
\end{figure}

Consider the database in Table~\ref{tab:2024_relations}. The three relations correspond to student information, instructor information, and advising relationships, respectively.

Assume:
\begin{itemize}
	\item Advisor.s\_id is a foreign key referencing Student.id.
	\item Advisor.i\_id is a foreign key referencing Instructor.id.
	\item Students are uniformly distributed across dept.
	\item Student names are unique.
	\item Each table has an in-memory B$^+$ tree index built on its primary key.
	\item ``Seq Scan'' denotes a table scan operator.
\end{itemize}

Consider the query below.

\begin{quote}
	\begin{verbatim}
SELECT s.name, i.name
FROM student AS s, advisor AS a, instructor AS i
WHERE s.name = 'Reed'
  AND s.id = a.s_id
  AND a.i_id = i.id;
	\end{verbatim}
\end{quote}

Two alternative plans for this query are shown in Figures~\ref{fig:2024-planA} and~\ref{fig:2024-planB}, and plan A is the optimizer's choice.

\begin{enumerate}
	\item List possible reasons why plan A is superior to plan B. (10 points)
	\item Sketch a third potential plan (plan C) for this query and explain why the optimizer is unlikely to choose it. (10 points)
\end{enumerate}

\subsection{2025 -- Part V: Query Understanding (20 points)}

\begin{table}[H]
  \centering \scriptsize
  \caption{Statistics of 2025 Teaching Database Relations.}
  \label{tab:2025_relations}
  \begin{tabular}{lccc}
  \toprule
  Relation & \# Tuples & \# Physical Pages & Primary Key \\
  \midrule
  Course(credits, course\_id, title, dept\_name) & 100 & 2 & Course\_id \\
  Instructor(id, name, age, dept) & 1000 & 10 & id \\
  Teaches(year, id, course\_id, sec\_id, semester) & 60000 & 600 & (s\_id, i\_id) \\
  \bottomrule
  \end{tabular}
\end{table}

\begin{figure}[H]
  \centering
  \includegraphics[width=\linewidth]{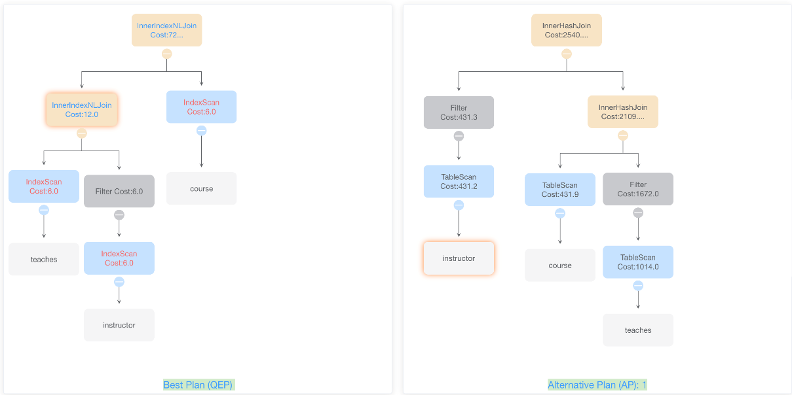}
  \caption{Optimal Plan and Alternative Plan for the 2025 query.}
  \label{fig:2025-plans}
\end{figure}
Consider the database in Table~\ref{tab:2025_relations}, where the three relations store course information, instructor information, and teaching relationships.

Assume:
\begin{itemize}
	\item Teaches.id is a foreign key referencing Instructor.id.
	\item Teaches.course\_id is a foreign key referencing Course.course\_id.
	\item Instructors are uniformly distributed across dept.
	\item Instructor names are unique.
	\item Each table has an in-memory B$^+$ tree index built on its primary key.
	\item ``Seq Scan'' denotes a table scan operator.
\end{itemize}

Consider the SQL query:

\begin{quote}
	\begin{verbatim}
-- List courses Professor Lembr teaches in the Fall semester.
SELECT title, semester
FROM course, teaches, instructor
WHERE course.course_id = teaches.course_id
  AND teaches.id = instructor.id
  AND instructor.name = 'Lembr'
  AND teaches.semester = 'Fall';
	\end{verbatim}
\end{quote}

\begin{enumerate}
	\item Figure~\ref{fig:2025-plans} depicts the optimal and an alternative plan. Analyze the advantages the optimal plan has over the alternative plan. (10 points)
	\item Describe at least one additional plan that satisfies the query and explain why it might not be chosen for execution. (10 points)
\end{enumerate}

\end{document}